\documentclass[review, 11pt, a4paper]{elsarticle}
\usepackage{amsmath}
\usepackage{graphicx,psfrag,epsf}
\usepackage{natbib}
\setcitestyle{authoryear,open={(},close={)}}
\usepackage[hidelinks]{hyperref}
\usepackage{enumerate}
\usepackage{doi}

\usepackage{url} % not crucial - just used below for the URL 
\newcommand{\blind}{0}
\usepackage[left=0.8in, right=0.8in, bottom = 0.8in, top = 0.8in]{geometry}

\makeatletter
\ifx\@journal\empty
  \def\ps@pprintTitle{%
  \let\@oddhead\@empty
  \let\@evenhead\@empty
  \def\@oddfoot{\small\itshape\hfil\today}
  \let\@evenfoot\@oddfoot
}
\fi
\makeatother
\usepackage{amsthm}
\usepackage{amsfonts}
\usepackage{bm}
\usepackage{bbm}

\makeatletter
\renewcommand*\env@matrix[1][\arraystretch]{%
  \edef\arraystretch{#1}%
  \hskip -\arraycolsep
  \let\@ifnextchar\new@ifnextchar
  \array{*\c@MaxMatrixCols c}}
\makeatother

\usepackage{mathrsfs}
\usepackage{amssymb}
\usepackage[dvipsnames]{xcolor}
\usepackage{paralist}
\usepackage{multirow, float}
\usepackage{makecell}
\usepackage{calc}
\usepackage{enumitem}
\usepackage{microtype}

\usepackage{thmtools}
\declaretheoremstyle[headfont=\sffamily\bfseries,%
 notefont=\sffamily\bfseries,%
 notebraces={}{},%
 headpunct=,%
 bodyfont={},%
 headformat=Theorem~\NUMBER.\NOTE\\,%
 ]{cstThm}
\declaretheorem[style=cstThm]{thm}

\declaretheoremstyle[headfont=\sffamily\bfseries,%
 notefont=\sffamily\bfseries,%
 notebraces={}{},%
 headpunct=,%
 bodyfont={},%
 headformat=Corollary~\NUMBER.\NOTE\\,%
 ]{cstCrl}
\declaretheorem[style=cstCrl]{crl}

\declaretheoremstyle[headfont=\sffamily\bfseries,%
 notefont=\sffamily\bfseries,%
 notebraces={}{},%
 headpunct=,%
 bodyfont={},%
 headformat=Remark~\NUMBER.,%
 ]{cstRmk}
\declaretheorem[style=cstRmk]{rml}

\usepackage[nodisplayskipstretch]{setspace}
\usepackage[bottom,hang,flushmargin]{footmisc}
\usepackage{adjustbox}

\usepackage{caption}
\DeclareCaptionStyle{italic}[justification=centering]
 {labelfont={bf},font=small,labelsep=colon}
\usepackage{subcaption}

\usepackage{longtable}
\usepackage{booktabs}
\usepackage{array}
\newcolumntype{M}[1]{>{\centering\arraybackslash}m{#1}}
\newcolumntype{L}[1]{>{\raggedright\arraybackslash}m{#1}}
\newcolumntype{R}[1]{>{\raggedleft\arraybackslash}m{#1}}

\newcommand{\bvet}{\bm{b}}

\newcommand{\uvet}{\bm{u}}

\newcommand{\xvet}{\bm{x}}
\newcommand{\yvet}{\bm{y}}

\newcommand{\Avet}{\bm{A}}
\newcommand{\Bvet}{\bm{B}}
\newcommand{\Cvet}{\bm{C}}
\newcommand{\Dvet}{\bm{D}}
\newcommand{\Evet}{\bm{E}}
\newcommand{\Fvet}{\bm{F}}
\newcommand{\Gvet}{\bm{G}}
\newcommand{\Hvet}{\bm{H}}
\newcommand{\Ivet}{\bm{I}}
\newcommand{\Jvet}{\bm{J}}
\newcommand{\Kvet}{\bm{K}}
\newcommand{\Lvet}{\bm{L}}
\newcommand{\Mvet}{\bm{M}}

\newcommand{\Pvet}{\bm{P}}

\newcommand{\Rvet}{\bm{R}}
\newcommand{\Svet}{\bm{S}}

\newcommand{\Wvet}{\bm{W}}
\newcommand{\Xvet}{\bm{X}}
\newcommand{\Yvet}{\bm{Y}}

\newcommand{\Unovet}{\bm{1}}
\newcommand{\Zerovet}{\bm{0}}

\newcommand{\betavet}{\bm{\beta}}
\newcommand{\gammavet}{\bm{\gamma}}
\newcommand{\muvet}{\bm{\mu}}
\newcommand{\epsvet}{\bm{\varepsilon}}
\newcommand{\etavet}{\bm{\eta}}
\newcommand{\lambdavet}{\bm{\lambda}}
\newcommand{\xivet}{\bm{\xi}}
\newcommand{\omegavet}{\bm{\omega}}

\newcommand{\Gammavet}{\bm{\Gamma}}
\newcommand{\Deltavet}{\bm{\Delta}}
\newcommand{\Lambdavet}{\bm{\Lambda}}
\newcommand{\Omegavet}{\bm{\Omega}}
\newcommand{\Phivet}{\bm{\Phi}}
\newcommand{\Psivet}{\bm{\Psi}}
\newcommand{\Sigmavet}{\bm{\Sigma}}
\newcommand{\Thetavet}{\bm{\Theta}}

\DeclareMathOperator*{\argmin}{arg\,min}

\definecolor{mybluehl}{HTML}{cbd3ff}

\makeatletter
\def\@endtheorem{\endtrivlist}
\makeatother

\usepackage{tikz}
\usetikzlibrary{arrows,shapes,positioning,shadows,trees}
\usetikzlibrary{matrix, decorations.pathreplacing, arrows, calc, fit, arrows.meta, decorations.pathmorphing, decorations.markings}

\tikzset{
  basic/.style  = {draw, text width=2cm, drop shadow, font=\sffamily, rectangle},
  root/.style   = {basic, rounded corners=2pt, thin, align=center,
                   fill=green!30},
  level 2/.style = {basic, rounded corners=6pt, thin,align=center, fill=green!60,
                   text width=4em},
  level 3/.style = {basic, thin, align=left, fill=pink!60, text width=1.5em}
}

\theoremstyle{definition}

\makeatletter
\newcommand{\githuburl}{\begingroup%
\if1\blind
{
\url{https://github.com/danigiro/cfc-project}
}\fi
\if0\blind
{
\url{https://github.com/danigiro/cfc-project}
}\fi
\endgroup}
\makeatother

\makeatletter
\newcommand{\Rpackage}{\begingroup%
	\if1\blind
	{
		\textsf{R} packages \texttt{FoReco} \citep{FoReco} and \texttt{FoCo2} \citep{Foco2}
	}\fi
	\if0\blind
	{
		\textsf{R} packages \texttt{FoReco} \citep{FoReco} and \texttt{FoCo2} \citep{Foco2}
	}\fi
	\endgroup}
\makeatother

\expandafter\def\expandafter\normalsize\expandafter{%
	\normalsize%
	\setlength\abovedisplayskip{3pt}%
	\setlength\belowdisplayskip{3pt}%
	\setlength\abovedisplayshortskip{0pt}%
	\setlength\belowdisplayshortskip{0pt}%
}
\usepackage{titlesec}
\titlespacing*{\subsection}{0pt}{0.1\baselineskip}{0\baselineskip}
\titlespacing*{\section}{0pt}{0.1\baselineskip}{0\baselineskip}

\begin{document}

\begin{frontmatter}

\title{\textbf{Coherent forecast combination for linearly constrained \\multiple time series}}

\if0\blind
{
  \author{\vspace*{-1em}Daniele Girolimetto\corref{mycorrespondingauthor}}
  \cortext[mycorrespondingauthor]{Corresponding author, \href{mailto:daniele.girolimetto@unipd.it}{\texttt{daniele.girolimetto@unipd.it}}}
  %\ead{daniele.girolimetto@unipd.it}
  \author{Tommaso Di Fonzo}
  
  \address{Department of Statistical Sciences, University of Padua, Padova 35121, Italy \vspace*{-3.5em}}
} \fi

\if1\blind
{
	\author{\vspace*{-4em}}
} \fi

%\bigskip
\begin{abstract}
Forecasting linearly constrained multiple time series, such as those arising in national accounts, energy systems, and supply chain planning, requires that the forecasts of all variables jointly satisfy the underlying accounting or resource-balance identities. When several experts provide forecasts for some or all of the variables, classical forecast combination operates on each variable and ignores the linear constraints, while forecast reconciliation enforces such constraints by using a single forecast per variable. To bridge this gap, we develop an original coherent forecast combination framework in which the joint problem of combination across experts and coherence across variables is formulated as a linearly constrained optimization problem. Closed-form expressions are derived for the coherent combined forecast and its error covariance matrix, and the statistical properties of the resulting estimators are discussed. A Monte Carlo experiment and a forecasting application to an Australian daily electricity generation dataset confirm that the proposed approach significantly improves forecast accuracy with respect to standard forecast combination and reconciliation alternatives.
\end{abstract}

\end{frontmatter}

%\vspace*{-0.75em}
%\newpage
\def\spacingset#1{\renewcommand{\baselinestretch}%
{#1}\small\normalsize} \spacingset{1}

\newpage
\spacingset{1.5} 
\section{Introduction}

In many operational contexts, forecasts of several interrelated variables are required to satisfy the exact linear relationships imposed by accounting or balance identities. For example, in national accounts, the gross domestic product can be measured from the income, expenditure, and output sides, and these three measures must agree \citep{Difonzo2023}. Analogous identities arise in energy systems, where total generation equals the sum of generation by region or technology \citep{Ben_Taieb2021, Panagiotelis2023}, and in supply chain and inventory planning, where aggregate demand decomposes across lower levels of the product hierarchy \citep{Makridakis2022-we, Petropoulos2025-uq}. 
Forecasts that fail to satisfy these relationships cannot be used directly for any decisions that require the constraints to hold exactly, such as resource allocation, accounting balancing, or institutional reporting. 

At the same time, more than one forecast per variable is usually available in practice. Statistical specifications, machine learning algorithms, and expert judgment are used in parallel, each producing predictions for all or part of the system. In these contexts, determining how to optimally combine forecasts from multiple sources while respecting the linear relationships among series remains a long-standing and fundamental problem in forecasting \citep{Rosas1994-po, Guerrero1998-ia, Guerrero2000-ve, Guerrero2003}.

Two well-established literatures address parts of this problem, but each from a different angle. The first, forecast combination \citep{Bates1969, Clemen1989, Timmermann2006, Wang2024}, combines several forecasts of the same variable into a single, possibly more accurate prediction. In its classical formulation, the combination is performed one variable at a time: the weights assigned to the forecasts of a variable are estimated using only forecasts of that variable, and the constraints relating it to the other variables play no role. The second, forecast reconciliation \citep{Hyndman2011, Athanasopoulos2024}, starts from a single base forecast per variable, typically produced by one source, and adjusts these forecasts so that the linear constraints hold exactly, possibly exploiting the cross-variable error covariance structure \citep[see, among many others,][]{Wickra2019, Girolimetto2024-jm}. In practice, however, the two requirements arise together. %When several sources forecast several linearly constrained variables, the resulting predictions need to be both combined across sources and coherent across variables, and neither the combination nor the reconciliation literature delivers this jointly.

To describe these frameworks in a uniform way, we employ terminology inspired by the recent literature on global forecast combination \citep{Thompson2024} and by the data fusion tradition \citep{Sun2004}. The \textit{task} is the problem of forecasting a single individual variable. An \textit{expert} is any source of base forecasts for one or more variables of the system. The term is meant to be inclusive: an expert can be a statistical specification, a machine learning algorithm, a forecasting service, an institutional projection, or an individual human forecaster, and in this paper experts are treated only through the unbiased point predictions they provide. A combination is \textit{local}, or \textit{single-task}, when the combined forecast of a variable is a function of base forecasts of that variable only. It is \textit{global}, or \textit{multi-task}, when it depends on the base forecasts of all variables\footnote{This distinction is closely related to the classification of \citet{Chen2025}, who distinguished between \textit{separate inference} methods and \textit{pooled inference} methods.}. It is \textit{single-expert} when only one source of base forecasts is available, and \textit{multi-expert} otherwise. The available base forecasts can be stacked into a single vector in two equivalent ways, grouped by the expert that produced them, which we call \textit{by-expert} organization, or by the variable they refer to, which we call \textit{by-variable} organization. A vector of forecasts of all constrained variables is \textit{coherent} when it satisfies the linear constraints exactly. Under this notation, classical forecast combination is local, multi-expert, and incoherent, while classical forecast reconciliation is global and coherent, and usually applied in a single-expert setting in the sense that a single base forecast per variable enters the procedure, regardless of how it has been obtained. Combining these two perspectives, we here develop an original forecasting framework that is simultaneously global, multi-expert, and coherent, which we refer to as \textit{coherent forecast combination}.

Recent contributions have begun to explore whether combining forecasts from different sources can improve accuracy in hierarchical settings \citep{Yang2019, Goehri2020, Mohamed2023, Zhang2024}. Some of them reinterpret reconciliation itself as a combination: \cite{Hollyman2021} and \cite{Difonzo2024} read the reconciled forecast of a variable as a weighted combination of the base forecasts of the levels at which that variable appears, an approach that requires a genuinely hierarchical structure. Others apply the two operations in sequence: \cite{Spiliotis2019}, for temporal hierarchies, and \cite{Rostami2024}, in a cross-sectional setting, first combine the base forecasts of each variable across experts and then reconcile them. 

We refer to this scheme as sequential combine-then-reconcile (scr), and we consider a reverse sequential reconcile-then-combine (src) scheme, in which the forecasts of each expert are first reconciled and the reconciled vectors are then combined. Both are formalized and used as benchmarks throughout the paper, together with the conditions under which they are available: src returns coherent forecasts only when the combination weights are common to all variables, and does not apply in the unbalanced case, in which different experts forecast different subsets of the variables. In addition, neither scheme exploits the cross-expert and cross-variable covariance information simultaneously, since the weights in scr are chosen one variable at a time and the reconciliation in src is performed one expert at a time. And neither solves a single optimality problem over the whole panel of base forecasts.

%The existing literature has thoroughly investigated two of the possible regimes in which the combination problem can arise. Classical forecast combination considers several experts producing forecasts for the same variable, with no role for the linear constraints, and treats each variable in isolation. Forecast reconciliation considers several variables linked by linear constraints, but starts from a single forecast per variable provided by one expert. The regime in which several experts forecast several linearly constrained variables, and the resulting predictions are required to be both combined and coherent, has not been treated in a unified way. Recent contributions have begun to explore whether combining forecasts from different sources can improve accuracy in hierarchical settings \citep{Spiliotis2019, Yang2019, Goehri2020, Mohamed2023, Rostami2024, Zhang2024}, with partial steps focused on specific hierarchical structures taken by \cite{Hollyman2021} and \cite{Difonzo2024}. The natural ad hoc remedy, applying combination and reconciliation in sequence in either order (combine-then-reconcile, scr; reconcile-then-combine, src), provides useful benchmarks but does not solve a single optimality problem and does not extend in closed form to the empirically common case in which different experts forecast different subsets of the variables; none of the existing contributions integrates combination across experts and reconciliation across variables within a single optimisation framework.

Building on the constrained least-squares adjustment of \cite{Stone1942-fa} and \cite{Byron1978-ws, Byron1979-hv}, and on the multi-source data fusion result of \cite{Sun2004}, we formulate the problem of global, multi-expert, coherent forecasting as a linearly constrained quadratic program using all available base forecasts. The solution may be expressed in closed form, and the resulting forecasts are the minimum variance linear unbiased coherent combination of the base forecasts. The resulting weight matrix admits an algebraic factorization into a projection onto the constraint subspace and a multi-task combination matrix. This is a property of the closed-form solution, not a sequential algorithm: both factors depend jointly on the full error covariance and on the constraint matrix. The framework also handles natively the empirically common \textit{unbalanced} case, in which different experts forecast different subsets of the variables, a setting for which sequential procedures are not always suitable. Moreover, in this comprehensive setting, classical forecast reconciliation is the special case in which only one expert is available, and classical local forecast combination is the special case of a single unconstrained variable.

We explore this framework through a Monte Carlo experiment that extends the simulation design of \cite{Capistran2009} to a hierarchical multivariate setting with cross-sectional dependencies and linear constraints, and through an empirical application on a daily Australian electricity generation dataset. Across both, the proposed coherent combination outperforms the base forecasts, as well as single-task combination schemes including the equal-weighted benchmark, single-expert reconciliation approaches including MinT \citep{Wickra2019}, and the sequential alternatives. The analysis treats the base forecasts as given inputs, regardless of how they are produced, as the forecast-generation process lies outside our scope.

The paper is organized as follows. \autoref{sec:Notation} introduces the notation and the alternative representations of a linearly constrained multiple time series, and defines the by-expert and by-variable organizations of the base forecasts. \autoref{sec:CFC} presents the formal representation linking multiple base forecasts to the target vector, develops the optimal coherent forecast combination and its covariance matrix, and discusses how reconciliation and classical combination may be recovered as special cases. \autoref{sec:on the covmat} addresses practical issues related to the estimation of the error covariance matrix. \autoref{sec:SimulationExperiment} reports the Monte Carlo experiment, and \autoref{sec:energy} the empirical application to Australian daily electricity generation. Finally, \autoref{sec:conclusion} presents conclusions and indications for future research. Supplementary theoretical material and additional tables and figures are available in Appendices~\ref{app:lcts}--\ref{app:fexp}.

\section{Notation and preliminaries}
\label{sec:Notation}

We use lower-case bold letters to denote column vectors, e.g., $\xvet \in \mathbb{R}^{N}$, and upper-case bold letters to denote matrices, e.g., $\Avet \in \mathbb{R}^{N \times K}$. An all-one vector of dimension $N$ is denoted as $\Unovet_N$. We use $\Ivet_N$ to denote the $(N \times N)$ identity matrix, and $\Zerovet_{(N \times K)}$ to denote the $(N \times K)$ zero matrix. %To specify that a vector $\xvet$ is non-negative for all its elements, we write $\xvet \succeq \Zerovet_{(N \times 1)}$. 
For matrix $\Avet$, $\Avet \succeq \Zerovet$ means that $\Avet$ is positive semi-definite. \autoref{tab:notation} summarizes the notation used throughout the paper.

\begingroup
\spacingset{1.1} 
\centering
	\small
\begin{longtable}{c|p{0.85\linewidth}}
\caption{List of symbols used in the paper.} \label{tab:notation} \\

		\toprule
		{\bf Symbol} & {\bf Description} \\
		\midrule 
\endfirsthead

\multicolumn{2}{c}%
{{\bfseries \tablename\ \thetable{} -- continued from previous page}} \\
		\toprule
		{\bf Symbol} & {\bf Description} \\
		\midrule 
\endhead

\midrule \multicolumn{2}{r}{{Continued on next page}} \\
\endfoot

\bottomrule
\endlastfoot

		\makecell[t]{$n_u$, $n_b$,\\[-0.15cm] $n$, $p$}  & Scalars denoting the number of constrained (upper), free (bottom), total variables, respectively, and the number of forecast experts.\\
		$i, k$ & Indices running on the variables, $i, k = 1, \dots, n$. \\
		$j, l$ & Indices running on the forecast experts, $j,l = 1, \dots, p$. \\
		$n_j$ & Scalar denoting the number of variables for which forecasts produced by the $j$-th expert are available, $1 \le n_j \le n$. \\
		$p_i$ & Scalar denoting the number of forecast experts available for the $i$-th variable, $1 \le p_i \le p$.\\
		$m$ & Scalar denoting the total number of available forecasts produced by $p$ experts for $n$ variables: $m = \displaystyle\sum_{j=1}^{p} n_j = \displaystyle\sum_{i=1}^{n} p_i$, $n \le m \le np$. When $n_j = n \; \forall j$, i.e., $p_i=p \; \forall i$, $m=np$ (`balanced' case). \\
		$y_i$ & Target forecast for the $i$-th variable. \\
$\uvet$ & $(n_u \times 1)$  vector of constrained (upper) variables.\\
$\bvet$ & $(n_b \times 1)$ vector of free (bottom) variables.\\
$\Avet$ & $(n_u \times n_b)$ linear combination matrix mapping $\bvet$ into $\uvet$: $\uvet = \Avet\bvet$.\\
		$\Cvet$ & $(n_u \times n)$ zero-constraints matrix: $\Cvet = 
		\begin{bmatrix}[0.7]
			\Ivet_{n_u} & -\Avet
		\end{bmatrix}$. \\
$\yvet$ & $(n \times 1)$ target forecast vector: $\yvet = \begin{bmatrix}[0.7]
			\uvet^\top & \bvet^\top
		\end{bmatrix}^\top$, by definition coherent ($\Cvet\yvet = \Zerovet_{(n_u \times 1)}$).\\
		$\Svet$ & $(n \times n_b)$ structural-like matrix: $\Svet = \begin{bmatrix}[0.7]
\Avet \\ \Ivet_{n_b}
		\end{bmatrix}$, such that $\yvet = \Svet\bvet$.\\
		$\widehat{y}_i^j$ & Unbiased base forecast of the $i$-th variable produced by the $j$-th expert. \\
		$\widehat{\yvet}^j$ &  $(n_j \times 1)$ vector of the base forecasts produced by the $j$-th expert.\\
		$\widehat{\yvet}_i$ & $(p_i \times 1)$ vector of the base forecasts of the $i$-th variable. \\
		$\widehat{\yvet} \equiv \widehat{\yvet}_{\text{be}}$ & $\left[
	\widehat{\yvet}^{1\top} \dots \;\widehat{\yvet}^{j\top} \dots \;\widehat{\yvet}^{p\top}
\right]^\top$: $(m \times 1)$ vector of the base forecasts stacked \textit{by-expert} (be).\\
$\widehat{\yvet}_{\text{bv}}$ & $\left[
	\widehat{\yvet}_1^\top \; \ldots \; \widehat{\yvet}_i^\top \; \ldots \; \widehat{\yvet}_n^\top
\right]^\top$: \makecell[{{p{0.9\linewidth}}}]{$(m \times 1)$ vector of the base forecasts stacked \textit{by-variable} (bv). } \\
		$\widehat{\Yvet}$ & $\begin{bmatrix}[0.7]
		\widehat{\yvet}^1 & \ldots & \widehat{\yvet}^j & \ldots & \widehat{\yvet}^p
		\end{bmatrix} $.
		If $n_j=n \; \forall j$ (i.e., `balanced' case): ($n \times p$) matrix containing the base forecasts produced by $p$ different experts for the target vector $\yvet$. Each column denotes an expert $j = 1,\dots,p$, while each row denotes a variable $i = 1, \dots, n$. In this case, $\widehat{\yvet}=\text{vec}\left(\widehat{\Yvet}\right)$.\\
		$\widehat{\Yvet}^\top$ & $\begin{bmatrix}[0.7]
			\widehat{\yvet}_1 & \ldots & \widehat{\yvet}_i & \ldots & \widehat{\yvet}_n
		\end{bmatrix}$.
		If $n_j=n \; \forall j$ (i.e., `balanced' case):  ($p \times n$) matrix containing the base forecasts produced by $p$ different experts for the target vector $\yvet$. Each column denotes a variable $i = 1, \dots, n$, while each row denotes an expert $j = 1,\dots,p$. In this case, $\widehat{\yvet}_{\text{bv}}=\text{vec}\left(\widehat{\Yvet}^\top\right)$.\\
		$\Pvet$  & $(m \times m)$ permutation matrix such that $\Pvet\widehat{\yvet} = \widehat{\yvet}_{\text{bv}}$. As $\Pvet^{-1} = \Pvet^\top$, $\widehat{\yvet} = \Pvet^\top\widehat{\yvet}_{\text{bv}}$. If $n_j=n \; \forall j$ (i.e., `balanced' case): $(np \times np)$ commutation matrix such that $\Pvet\text{vec}\left(\widehat{\Yvet}\right) = \text{vec}\left(\widehat{\Yvet}^\top\right)$. \\
	$\Lvet_j$ & ($n_j \times n$) %zero-one 
	matrix selecting the entries of $\yvet$ for which forecasts of expert $j$ are available.\\
$\Lvet$ & $\text{Diag}\left(\Lvet_1, \ldots, \Lvet_j, \ldots, \Lvet_p\right) \in \{0,1\}^{m \times np}$. If $n_j=n \; \forall j$ (i.e., `balanced' case): $\Lvet = \Ivet_{np}$.\\
$\Kvet$  & $\Lvet\left(\Unovet_p \otimes \Ivet_n\right) = \left[\Lvet_1^\top \ldots \Lvet_j^\top \ldots \Lvet_p^\top\right]^\top \in \{0,1\}^{m \times n}$. If $n_j=n \; \forall j$ (i.e., `balanced' case):  $\Kvet = \Unovet_p \otimes \Ivet_n$, $(np \times n)$ matrix such that $\Kvet\yvet = \big[\;\underbrace{\yvet^\top \; \ldots^{\phantom{\top}} \; \yvet^\top}_{p \; \text{times}}\;\big]^\top =
		\begin{bmatrix}[0.7] \yvet \\ \vdots \\ \yvet \end{bmatrix}$. \\
		$\Jvet$ & $\Pvet\Kvet = \Pvet\Lvet\left(\Unovet_p \otimes \Ivet_n\right) \in \{0,1\}^{m \times n}$, such that $\Jvet\yvet =  \begin{bmatrix}[0.7]
			\Unovet_{p_1} y_1 \\ \vdots \\ \Unovet_{p_n} y_n
		\end{bmatrix}$.
If $n_j=n \; \forall j$ (i.e., `balanced' case):  $\Jvet = \Ivet_n \otimes \Unovet_p$: $(np \times n)$ matrix such that $\Jvet\yvet = \yvet \otimes \Unovet_p$.	\\
		$\Wvet$ & $(m \times m)$ error covariance matrix of $\widehat{\yvet}$.\\
		$\Wvet_{jl}$ & $(n_j \times n_l)$ error cross-covariance matrix of $\widehat{\yvet}^j$ and $\widehat{\yvet}^l$.\\
$\Sigmavet$ & $(m \times m)$ error covariance matrix of $\widehat{\yvet}_{\text{bv}}$. 
 $\Sigmavet = \Pvet\Wvet\Pvet^\top$ and $\Wvet = \Pvet^\top\Sigmavet\Pvet$.\\
		$\Sigmavet_{ik}$ & $(p_i \times p_k)$ error cross-covariance matrix of $\widehat{\yvet}_i$ and $\widehat{\yvet}_k$.\\
$\widehat{\yvet}^c$ & $(n \times 1)$ vector of the multi-task combined forecasts.
In general, these forecasts are incoherent: $\Cvet\widehat{\yvet}^c \neq \Zerovet_{(n_u \times 1)}$. \\
		$\widetilde{\yvet}^j$ & If $n_j=n$, $(n \times 1)$ vector of reconciled forecasts using the base forecasts produced by the $j$-th expert. $\widetilde{\yvet}^j$ is by construction coherent: $\Cvet\widetilde{\yvet}^j = \Zerovet_{(n_u \times 1)}$.\\
$\widetilde{\yvet}^c$ & $(n \times 1)$ vector of the optimal coherent combined forecasts, by construction coherent: $\Cvet\widetilde{\yvet}^c = \Zerovet_{(n_u \times 1)}$. \\
		$\Omegavet$ & $\begin{bmatrix}[0.7]
			\Omegavet_1 \; \ldots \; \Omegavet_j \; \ldots \; \Omegavet_p
		\end{bmatrix}^\top$:
		$(m \times n)$ matrix given by the concatenation of the matrix combination weights $\Omegavet_j \in \mathbb{R}^{n \times n_j}$ of the multi-task forecast combination $\widehat{\yvet}^c = \Omegavet^\top\widehat{\yvet} = \displaystyle\sum_{j=1}^{p}\Omegavet_j\widehat{\yvet}^j$ (by-expert formulation).\\
		$\Gammavet$ & $\begin{bmatrix}[0.7]
	\Gammavet_1 \; \ldots \; \Gammavet_i \; \ldots \; \Gammavet_n
\end{bmatrix}^\top$:
$(m \times n)$ matrix given by the concatenation of the matrix combination weights $\Gammavet_i \in \mathbb{R}^{n \times p_i}$ of the multi-task forecast combination $\widehat{\yvet}^c = \Gammavet^\top\widehat{\yvet}_{\text{bv}} = \displaystyle\sum_{i=1}^{n}\Gammavet_i\widehat{\yvet}_i$ (by-variable formulation).\\
		$\Psivet$ & $\begin{bmatrix}[0.7]
			\Psivet_1 \; \ldots \; \Psivet_j \; \ldots \; \Psivet_p
		\end{bmatrix}^\top$:
		$(m \times n)$ matrix given by the concatenation of the matrix reconciliation weights $\Psivet_j \in \mathbb{R}^{n \times n_j}$ of the coherent multi-task forecast combination $\widetilde{\yvet}^c = \Psivet^\top\widehat{\yvet} = \displaystyle\sum_{j=1}^{p}\Psivet_j\widehat{\yvet}^j$ (by-expert formulation).\\
		$\Phivet$ & $\begin{bmatrix}[0.7]
	\Phivet_1 \; \ldots \; \Phivet_i \; \ldots \; \Phivet_n
\end{bmatrix}^\top$:
$(m \times n)$ matrix given by the concatenation of the matrix reconciliation weights $\Phivet_i \in \mathbb{R}^{n \times p_i}$ of the coherent multi-task forecast combination $\widetilde{\yvet}^c = \Phivet^\top\widehat{\yvet}_{\text{bv}} = \displaystyle\sum_{i=1}^{n}\Phivet_i\widehat{\yvet}_i$ (by-variable formulation).\\
%		src & Sequential first-reconciliation-then-combination (simple-average). \\
%		scr & Sequential first-combination-then-reconciliation.\\
%		occ & Optimal (minimum mean square error, MMSE) coherent linear combination.\\
\end{longtable}
\endgroup

\vspace*{-0.5\baselineskip}

\subsection{Zero-constrained and structural representations of a linearly constrained time series}
\label{sec:general_lcts}

The framework developed in this paper applies to any multiple time series subject to linear constraints, and not only to genuine hierarchical or grouped structures. This section fixes the two equivalent representations of such constraints that are used in the following: the \textit{zero-constrained form}, on which the constrained quadratic program of \autoref{sec:OCC} is built; and the \textit{structural form}, through which the solution is expressed as a multi-expert extension of MinT.

Let $\yvet_t = \begin{bmatrix}[0.7] y_{1,t} & \ldots & y_{i,t} & \ldots \; y_{n,t} \end{bmatrix}^\top \in \mathbb{R}^{n}$ be a vector of the observed time series we are interested in forecasting, and assume that the time series is linearly constrained, in the sense that at every time $t$ the $n$ individual variables of $\yvet_t$ must satisfy $n_u < n$ independent linear constraints. Without loss of generality, we assume that the $n_u$ constraints are expressed in \textit{zero-constrained form}, according to the homogeneous linear system:
\begin{equation}
\label{eq:zero-constraint-formulation}
\Cvet \yvet_t = \Zerovet_{(n_u \times 1)} ,
\end{equation}
where $\Cvet \in \mathbb{R}^{n_u \times n}$ is a zero-constraints matrix with full row rank\footnote{Full row rank rules out redundant constraints. That is, rows of $\Cvet$ that can be expressed as a linear combination of the others.}.

\cite{Girolimetto2024} showed that every general linearly constrained time series $\yvet_t$ may be expressed as $\yvet_t = \begin{bmatrix}[0.7] \uvet_t^\top & \bvet_t^\top \end{bmatrix}^\top$, i.e., stacking an `upper' vector $\uvet_t \in \mathbb{R}^{n_u}$ and a `bottom' vector $\bvet_t \in \mathbb{R}^{n_b}$, with $n = n_u + n_b$, where $\uvet_t$ is a linear combination of the bottom time series $\bvet_t$: $\uvet_t = \Avet\bvet_t$. $\Avet \in \mathbb{R}^{n_u \times n_b}$ is the \textit{linear combination matrix} mapping $n_b$ free variables $\bvet_t$ into $n_u$ constrained variables $\uvet_t$. Then, we denote by $\Cvet = \begin{bmatrix}[0.7] \Ivet_{n_u} & -\Avet \end{bmatrix} \in \mathbb{R}^{n_u \times n}$ the zero-constraints matrix, and by $\Svet = \begin{bmatrix}[0.85] \Avet \\ \Ivet_{n_b} \end{bmatrix} $ the \textit{structural-like matrix} \citep{Girolimetto2024}, the constraints on $\yvet_t$ may be equivalently expressed either as Eq.~(\ref{eq:zero-constraint-formulation}), above, or in \textit{structural form}:
\begin{equation}
\label{eq:structural-like-formulation}
\yvet_t = \Svet\bvet_t	.
\end{equation}
The two representations are equivalent. Each one can be obtained from the other, and the results derived from either lead to the same coherent forecasts.

For example, the left panel of \autoref{fig:example} shows a three-level hierarchical time series ($n = 7$, $n_u = 3$) $\yvet_t = \begin{bmatrix}[0.7] \texttt{X}_t & \texttt{A}_t & \texttt{B}_t & \texttt{AA}_t & \texttt{AB}_t & \texttt{BA}_t & \texttt{BB}_t \end{bmatrix}^\top$, with 
\begin{equation}\label{eq:hsr}
\addtolength{\jot}{-0.4em}
\left\{\rule{0cm}{1cm}\right.\begin{aligned}
		\texttt{X} &= \texttt{AA} + \texttt{AB} + \texttt{BA} + \texttt{BB} \\[-0.6ex]
	\texttt{A} &= \texttt{AA} + \texttt{AB} \\[-0.6ex]
	\texttt{B} &= \texttt{BA} + \texttt{BB}
\end{aligned}\ \quad\mbox{and}\quad \Cvet = \begin{bmatrix}[0.8]
	1 & 0 & 0 & -1 & -1 & -1 & -1\\[-0.4ex]
	0 & 1 & 0 & -1 & -1 & 0 & 0\\[-0.4ex]
	0 & 0 & 1 & 0 & 0 & -1 & -1\\
\end{bmatrix}.	
\end{equation}
In this case, $\uvet_t = \begin{bmatrix}[0.7] \texttt{X}_t & \texttt{A}_t & \texttt{B}_t \end{bmatrix}^\top$, $\bvet_t = \begin{bmatrix}[0.7] \texttt{AA}_t & \texttt{AB}_t & \texttt{BA}_t & \texttt{BB}_t \end{bmatrix}^\top$, and matrices $\Avet$ and  $\Svet$ are easily deduced from Eq.~(\ref{eq:hsr}):
\begin{equation}
	\label{eq:SimpleExampleMatrices}
\Avet = \begin{bmatrix}[0.8]
	1 & 1 & 1 & 1\\[-0.4ex]
	1 & 1 & 0 & 0\\[-0.4ex]
	0 & 0 & 1 & 1\\
\end{bmatrix} , \qquad
\Svet = \begin{bmatrix}[0.8]
	1 & 1 & 1 & 1\\[-0.4ex]
	1 & 1 & 0 & 0\\[-0.4ex]
	0 & 0 & 1 & 1\\[-0.4ex]
\multicolumn{4}{c}{\Ivet_4}
\end{bmatrix} .
\end{equation}

Here, $\Cvet$ already has the required $\begin{bmatrix}[0.7] \Ivet_{n_u} & -\Avet \end{bmatrix}$ form, and $\Avet$ and $\Svet$ can be read off directly from the constraints. This is not so for a general linearly constrained time series, such as the one in the right panel of \autoref{fig:example}, which consists of two simple hierarchies sharing the same top-level series: the zero-constraints matrix does not come in the required form, and $\Avet$ is not directly readable, although it can always be recovered through elementary row operations \citep{Girolimetto2024}. This example is worked out in detail in \autoref{app:lcts}.

\begin{figure}[tb]\centering
      \includegraphics[width=0.75\linewidth]{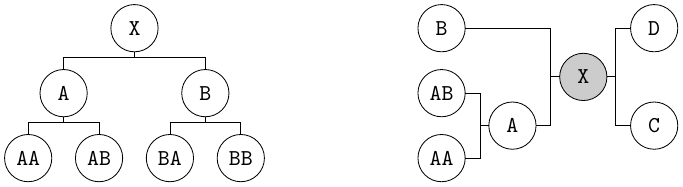}
      \caption{Two examples of linearly constrained time series. Left: a three-level genuine hierarchical structure. Right: a general linearly constrained multiple time series consisting of two simple hierarchies sharing the same top-level series.\label{fig:example}}
      \vspace*{-\baselineskip}
\end{figure}

Starting with \cite{Athanasopoulos2009} and \cite{Hyndman2011}, the \textit{structural form} (\ref{eq:structural-like-formulation}) is widely used in the literature on forecast reconciliation for genuine hierarchical/grouped time series \citep[for a recent review, see][]{Athanasopoulos2024}. In fact, Eq.~(\ref{eq:structural-like-formulation}) has the merit of describing the relationships linking the free (bottom) and the constrained (upper) time series in a way that visually recalls a genuine hierarchical or grouped structure, by placing the aggregated series at higher levels than the bottom one, where the most disaggregated series take place. Moreover, this formulation allows the forecast reconciliation results to be developed through a simple unconstrained regression model. The zero-constrained form, on the other hand, is the more convenient starting point for general constraints, since it does not require the free variables to be identified in advance \citep{Girolimetto2024}. The framework is therefore developed from the zero-constrained form, and the equivalent structural expressions are reported throughout the paper.

\subsection{Base forecast organization: By expert and by variable}
\label{sec:be_vs_bv}
Suppose we have base forecasts of the $n$ individual variables of the target vector $\yvet_{t+h} \in \mathbb{R}^n$, where $h \ge 1$ is the forecast horizon, produced by $p \ge 2$ experts. In general, we admit that the forecasts provided by each expert might refer to different sets of individual variables, and we denote them by the vectors $\widehat{\yvet}_{t+h}^1 \in \mathbb{R}^{n_1}$, \dots, $\widehat{\yvet}_{t+h}^p \in \mathbb{R}^{n_p}$, with $1 \le n_j \le n$, $j=1,\ldots,p$. In analogy to the literature on panel data models \citep{Baltagi2021}, we refer to this situation as the `unbalanced case'. To simplify the notation, in the following we consider a single prediction horizon, say $h=1$, and omit the subscript $t+1$. In other terms, $y_i \equiv y_{i,t+1}$, $\widehat{y}_i^j \equiv \widehat{y}^j_{i,t+1}$, $\yvet \equiv \yvet_{t+1}$ and $\widehat{\yvet}^j \equiv \widehat{\yvet}^j_{t+1}$. Let $\widehat{\yvet}_i \in \mathbb{R}^{p_i}$ be the vector of the base forecasts available for the $i$-th variable, where $1 \le p_i \le p$, and denote by $m = \displaystyle\sum_{j=1}^{p} n_j = \sum_{i=1}^{n}p_i$ the total number of available forecasts (i.e., produced by $p$ experts for $n$ individual variables).

All the available base forecasts may be grouped into a single $(m \times 1)$ vector in two equivalent ways. In the \textit{by-expert} (be) organization the forecasts are stacked in $p$ consecutive blocks, each containing the $n_j$ forecasts produced by a single expert. In the \textit{by-variable} (bv) organization they are arranged in $n$ blocks, each containing the $p_i$ forecasts available for one variable, regardless of which expert produced them. In what follows, we mainly refer to the by-expert data organization, so we omit the subscript `be' and use $\widehat{\yvet}$ instead of $\widehat{\yvet}_{\text{be}}$:
\begin{equation}
\label{eq:yhat_be_bv}
\widehat{\yvet} \mathrel{\mathop:}= \widehat{\yvet}_{\text{be}} = \begin{bmatrix}[0.85]
	\widehat{\yvet}^1 \\[-0.4ex] \vdots \\[-0.4ex] \widehat{\yvet}^j \\[-0.4ex] \vdots \\[-0.4ex] \widehat{\yvet}^p
\end{bmatrix} \in \mathbb{R}^m ,
\qquad
\widehat{\yvet}_{\text{bv}} = \begin{bmatrix}[0.85]
\widehat{\yvet}_1 \\[-0.4ex] \vdots \\[-0.4ex] \widehat{\yvet}_i \\[-0.4ex] \vdots \\[-0.4ex] \widehat{\yvet}_n
\end{bmatrix}  \in \mathbb{R}^m .
\end{equation}
For instance, with $p=2$ experts and $n=3$ variables (and forecasts available for every expert--variable pair), $\widehat{\yvet} = \left[\widehat{y}_1^1\; \widehat{y}_2^1\; \widehat{y}_3^1\; \widehat{y}_1^2\; \widehat{y}_2^2\; \widehat{y}_3^2\right]^\top$ and $\widehat{\yvet}_{\text{bv}} = \left[\widehat{y}_1^1\; \widehat{y}_1^2\; \widehat{y}_2^1\; \widehat{y}_2^2\; \widehat{y}_3^1\; \widehat{y}_3^2\right]^\top$.
Passing from a by-expert to a by-variable data organization can be achieved through a permutation matrix \citep{Magnus2019-lo}, i.e., a square matrix $\Pvet\in \{0,1\}^{m \times m}$ obtained from the same size identity matrix by a permutation of rows, such that $\Pvet\widehat{\yvet} = \widehat{\yvet}_{\text{bv}}$. In addition, as $\Pvet^{-1} = \Pvet^\top$, it is $\widehat{\yvet} = \Pvet^\top\widehat{\yvet}_{\text{bv}}$.

%\vspace{.25cm}

\noindent \textbf{Balanced case}

\noindent The notation so far encompasses the complete case where $n_j = n$ $\forall j$ (i.e., $p_i = p$ $\forall i$, $m=np$). In this case, which we call `balanced', the available base forecasts may be organized as an $(n \times p)$ matrix $\widehat{\Yvet}$, whose $j$-th column contains the $n$ forecasts provided by the $j$-th individual expert (by-expert organization):
\begin{equation}
	\label{Yhat}
	\widehat{\Yvet} = 
	\begin{bmatrix}[0.85]
		\widehat{y}_{1}^1 & \ldots & \widehat{y}_{1}^j & \ldots & \widehat{y}_{1}^p\\[-0.4ex]
		\vdots & \ddots & \vdots & \ddots & \vdots \\[-0.4ex]
		\widehat{y}_{i}^1 & \ldots & \widehat{y}_{i}^j & \ldots & \widehat{y}_{i}^p\\[-0.4ex]
		\vdots & \ddots & \vdots & \ddots & \vdots \\[-0.4ex]
		\widehat{y}_{n}^1 & \ldots & \widehat{y}_{n}^j & \ldots & \widehat{y}_{n}^p\\
	\end{bmatrix} =
	\begin{bmatrix}[0.7]
		\widehat{\yvet}^1 & \ldots & \widehat{\yvet}^j & \ldots & \widehat{\yvet}^p
	\end{bmatrix} .
\end{equation}
An equivalent by-variable organization of the available base forecasts is obtained by considering the transpose of matrix $\widehat{\Yvet}$: $\widehat{\Yvet}^\top = 
	\begin{bmatrix}[0.7]
		\widehat{\yvet}_{1} & \ldots & \widehat{\yvet}_{i} & \ldots & \widehat{\yvet}_{n}
	\end{bmatrix}$,
where $\widehat{\yvet}_{i} \in \mathbb{R}^p$ contains the base forecasts of the $i$-th individual variable provided by the $p$ experts. In this case, moving from a by-expert to a by-variable data organization can be achieved through a commutation matrix $\Pvet\in \{0,1\}^{np \times np}$ \citep{Magnus2019-lo}, such that $\Pvet\text{vec}\left(\widehat{\Yvet}\right) = \text{vec}\left(\widehat{\Yvet}^\top\right)$.

\subsection{Local, global, and coherent forecast combination}
\label{sec:cases}

If only the $i$-th scalar variable is considered (i.e., $n=1$ and $p \ge 2$) and a linear solution $\widehat{y}^c_i$ is sought, we face a classical problem of single-task linear forecast combination \citep{Bates1969}, without any coherence issue. In this case, the combined forecast $\widehat{y}^c_i$ is the result of a \textit{local} \citep{Thompson2024} linear combination across the $p$ experts of the same variable:
\begin{equation}
\label{eq:single task fc}
\widehat{y}^c_i = \omegavet_i^\top\widehat{\yvet}_i =
\displaystyle\sum_{j=1}^{p}\omega_{ij}\widehat{y}^j_i,
\end{equation}
where $\widehat{\yvet}_i \in \mathbb{R}^{p}$ and $\omegavet_i \in \mathbb{R}^p$ is a vector of combination weights\footnote{For comprehensive reviews of the many weighting schemes proposed in the forecasting literature, see \cite{Timmermann2006} and \cite{Wang2024}.}.

On the other hand, if only the base forecasts of the $j$-th single expert are available (i.e., $p=1$) for all the $n > 1$ individual variables of $\yvet$, we face a classical problem of single-expert forecast reconciliation. When looking for a linear solution, the reconciled (i.e., coherent) forecast vector may be expressed as 
\begin{equation}
	\label{eq:single model fr}
	\widetilde{\yvet} = \Mvet\widehat{\yvet}^j
	\quad \rightarrow \quad 
	\widetilde{y}_i = m_{ii} \widehat{y}_{i}^j + \sum_{\substack{k=1 \\ k \ne i}}
	^{n} m_{ik} \widehat{y}_{k}^j, \quad i=1, \ldots, n,
\end{equation}
where $\widehat{\yvet}^j \in \mathbb{R}^n$, and $\Mvet = \{m_{ik}\} \in \mathbb{R}^{n \times n}$  is a projection matrix depending on the base forecast errors' covariance matrix and the zero constraints matrix $\Cvet$ \citep{Panagiotelis2021, Difonzo2023}. The single-variable formulation appearing on the right-hand side of Eq.~(\ref{eq:single model fr}) highlights the \textit{global} nature \citep{Thompson2024} of the forecast combination underlying the single-expert linear reconciliation, in the sense that the reconciled forecast of the $i$-th individual variable is given by the linear combination of the base forecasts of the same variable, $\widehat{y}_{i}^j$, \textbf{and} of all the remaining $n-1$ individual variables.

Following this observation, we might consider the extension to the multi-task framework of the single-task linear forecast combination approach in Eq.~(\ref{eq:single task fc}). Useful references are \cite{Sun2004}, who dealt with data assimilation problems, and \cite{Keller2004} and \cite{Lavancier2016}, who extended the statistical theory on the combination of multiple estimators of the same vector of parameters.
The results found by \cite{Sun2004} and \cite{Lavancier2016} may be used to express a global multi-task combined forecast $\widehat{\yvet}^c$ which exploits all the base forecasts of all variables, and is more accurate than any individual or multiple base forecasts.
Surprisingly, these papers appear to have had a limited impact on research on forecast combination. Here, we aim to address this gap.

However, it should be noted that, if the variable to be forecast is linearly constrained, there is no guarantee that $\widehat{\yvet}^c$ is coherent. In other terms, in general it is $\Cvet\widehat{\yvet}^c \ne \Zerovet_{(n_u \times 1)}$. A possible, practical solution to obtain coherent forecasts is to adopt a two-step procedure (\textit{sequential combination-first-then-reconciliation}, scr), consisting of performing the reconciliation of $\widehat{\yvet}^c$ in order to finalize the coherent combination forecasting process \citep[see][]{Spiliotis2019, Rostami2024}. Another simple two-step procedure, \textit{sequential reconciliation-first-then-equal-weighted-combination} (src) may be developed as well. It is worth mentioning that the src approach is limited to the balanced case, and that equal weights are sufficient for coherence but they are not necessary: since $\Cvet\widetilde{\yvet}^j = \Zerovet_{(n_u \times 1)}$ for every expert $j$, any set of weights common to all variables preserves coherence, which fails only when the combination weights vary across variables.

In this paper we address the problem of finding the optimal (in a least-squares sense) coherent-and-combined forecast vector $\widetilde{\yvet}^c \in \mathbb{R}^n$ of a linearly constrained time series, where the coherence condition is expressed as $\Cvet\widetilde{\yvet}^c = \Zerovet_{(n_u \times 1)}$. \autoref{tab:forcombin} schematically shows combination and reconciliation approaches in relation to the forecast coherence property and to the local/global nature of the final result, with the last row anticipating the framework developed in \autoref{sec:OCC}.

\begin{table}[tb]
	\centering
		\caption{Combination approaches for single- and multi-task forecast combination, single-expert forecast reconciliation, and coherent forecast combination. The term \textbf{local} denotes the combination of multiple base forecasts of one individual variable, while  \textbf{global} refers to the combination of either single or multiple expert base forecasts of all $n$ individual variables.}\label{tab:forcombin}
		\begingroup
	\small
\setlength{\tabcolsep}{4pt}
\setlength\extrarowheight{5pt}
\resizebox{\linewidth}{!}{
		\begin{tabular}{lcc@{\hspace{-8pt}}c@{\hspace{5pt}}l}
			\toprule
			{\bf Combination approach} & {\bf Forecast} & 
			{\bf Coherence} & & {\bf Nature of the combination} \\
			\midrule
			Single-task forecast combination &  $\widehat{y}_i^c$ & {\color{red}No} & $\longrightarrow$ & Local, multiple experts \\
			Forecast reconciliation & $\widetilde{\yvet}$ & {\color{blue}Yes} & $\longrightarrow$ & Global, single expert \\
			Multi-task forecast combination & $\widehat{\yvet}^c$ & {\color{red}No} & \multirow{4}{*}{$\left.\rule{0cm}{3.5em}\right\}\hspace{-5pt}\rightarrow$} & \multirow{4}{*}{Global, multiple experts}\\
			Sequential combination $+$ reconciliation & $\widetilde{\yvet}^c_{\text{scr}}$ & {\color{blue}Yes} &   \\
			Sequential reconciliation $+$ equal-weighted combination & $\widetilde{\yvet}^c_{\text{src}}$ & {\color{blue}Yes} &   \\
			Optimal coherent combination & $\widetilde{\yvet}^c$ & {\color{blue}Yes} &  \\
			\bottomrule
		\end{tabular}}
		\endgroup
		\vspace*{-\baselineskip}
	\end{table}

\section{Combining forecasts under linear constraints}
\label{sec:CFC}

To discuss the optimal coherent forecast combination methodology, we slightly abstract the problem and consider a data-fitting issue with equality constraints. As in the classical frameworks of \cite{Stone1942-fa} and \cite{Bates1969}, we assume that each individual base forecast is unbiased, i.e., $E\big(\widehat{y}_i^j\big) = E\big(y_i\big)$ for $i=1,\ldots,n_j$, $j=1,\ldots,p_i$. Following \cite{Stone1942-fa} \citep[see also ][]{Byron1978-ws, Byron1979-hv}, we state that the base forecast of the $i$-th variable produced by the $j$-th expert is the sum of the unknown target value $y_i$ and a zero-mean forecast error $\varepsilon_i^j$:
\begin{equation}
\label{eq:Stonemodel}
\widehat{y}_i^j = y_i + \varepsilon_i^j, \quad i=1,\ldots,n_j, \; j=1,\ldots,p_i.
\end{equation}

In the following, we discuss the unbalanced case\footnote{The balanced case (i.e., $n_j=n$ $\forall j$, $p_i=p$ $\forall i$) is presented in \autoref{app:rect}.}, occurring when for at least one index $j$, it is $n_j < n$, which means that the $p$ experts do not produce base forecasts for each of the $n$ variables. We denote by $\Lvet = \text{Diag}\left(\Lvet_1, \ldots, \Lvet_j, \ldots, \Lvet_p\right)$ the $(m \times np)$ selection matrix, where $\Lvet_j \in \{0,1\}^{n_j \times n}$, $j=1, \ldots, p$, selects the $n_j \le n$ entries of $\yvet$ for which base forecasts of the $j$-th expert are available\footnote{Clearly, if $n_j =n \; \forall j$, it is $\Lvet_j=\Ivet_n \;\forall j$, and $\Lvet = \Ivet_{np}$.}. Model (\ref{eq:Stonemodel}) may now be grouped into the $p$ linear models:
\begin{equation}
	\label{eq:hatyregselect}
	\widehat{\yvet}^j = \Lvet_j\yvet + \epsvet^j, \quad j=1,\ldots,p,
\end{equation}
where $\widehat{\yvet}^j \in \mathbb{R}^{n_j}$, and the base forecast errors $\epsvet^j \in \mathbb{R}^{n_j}$ are $(n_j \times 1)$ zero-mean random vectors, with variance--covariance matrices $\Wvet_{jl} = E\left[\epsvet^j(\epsvet^l)^\top\right]  \in \mathbb{R}^{n_j\times n_l}$, $j,l = 1, \ldots, p$. Then, we denote by
\begin{equation}
	\label{eq:yhat_epshat}
\widehat{\yvet} = \begin{bmatrix}[0.85]
	\widehat{\yvet}^1\\[-0.4ex] \vdots \\[-0.4ex] \widehat{\yvet}^j \\[-0.4ex] \vdots \\[-0.4ex] \widehat{\yvet}^p
\end{bmatrix} \in \mathbb{R}^m, \quad
\epsvet = \begin{bmatrix}[0.85]
	\epsvet^1\\[-0.4ex] \vdots \\[-0.4ex] \epsvet^j \\[-0.4ex] \vdots \\[-0.4ex] \epsvet^p
\end{bmatrix} \in \mathbb{R}^m .
\end{equation}
The linear relationship linking all the available base forecasts $\widehat{\yvet}$ and the unknown target value vector $\yvet$ can be expressed through the multiple regression model:
\begin{equation}\label{mod:stone_unconstrained_selection}
	\widehat{\yvet} = \Kvet\yvet + \epsvet ,
\end{equation}
where $\Kvet = \Lvet\left(\Unovet_p \otimes \Ivet_n\right) \in \{0,1\}^{m \times n}$, and $\Wvet =E(\epsvet\epsvet^\top) \in \mathbb{R}^{m \times m}$ is a positive definite block-matrix:
\begin{equation}
\label{eq:Wmatrix}
\Wvet = \begin{bmatrix}[0.85]
	\Wvet_{1} & \cdots & \Wvet_{1j} & \cdots & \Wvet_{1p} \\[-0.4ex]
	\vdots & \ddots & \vdots & \ddots & \vdots \\[-0.4ex]
	\Wvet_{j1} & \cdots & \Wvet_{j} & \cdots & \Wvet_{jp} \\[-0.4ex]
	\vdots & \ddots & \vdots & \ddots & \vdots \\[-0.4ex]
	\Wvet_{p1} & \cdots & \Wvet_{pj} & \cdots & \Wvet_{p}
\end{bmatrix},
\end{equation}
where $\Wvet_j \equiv \Wvet_{jj}$, $j=1,\ldots,p$. A simple numerical example is shown in \autoref{app:simpex}.

\subsection{Optimal coherent forecast combination}
\label{sec:OCC}

Let $\yvet = \left[y_1 \; \ldots \; y_i \; \ldots \; y_n\right]^\top \in \mathbb{R}^n$ be the target forecast vector of a linearly constrained time series $\yvet_t$ such that $\Cvet\yvet_t = \Zerovet_{(n_u \times 1)}$. We are looking for a coherent forecast vector $\widetilde{\yvet}^c$, i.e., $\Cvet\widetilde{\yvet}^c = \Zerovet_{(n_u \times 1)}$, which exploits all the available base forecasts and improves their accuracy. Forecast combination and reconciliation can be simultaneously dealt with through an optimization-based technique that combines the base forecasts of multiple individual experts $\widehat{\yvet}^j$. By extending the well-known procedure of least-squares adjustment for a single vector of preliminary incoherent estimates (i.e., base forecasts) \citep{Stone1942-fa, Byron1978-ws, Byron1979-hv, Difonzo2023} to the case of $p \ge 2$ vectors, the coherent combined forecast vector $\widetilde{\yvet}^c$ can be expressed as a weighted sum of base forecasts from individual experts. This interesting result is shown by \autoref{thm:mmse}.

\begin{thm}[Optimal linear coherent forecast combination $\widetilde{\yvet}^c$] 
\label{thm:mmse}
The linear minimum variance coherent combined forecast vector $\widetilde{\yvet}^c$, obtained as solution to the linearly constrained quadratic program
\begin{equation}
\label{eq:lcquadprog_be}
\widetilde{\yvet}^c = \argmin_{\yvet}\left(\widehat{\yvet} - \Kvet\yvet\right)^\top\Wvet^{-1}\left(\widehat{\yvet} - \Kvet\yvet\right) \qquad \text{s.t. } \Cvet\yvet=\Zerovet_{(n_u \times 1)} ,
\end{equation}
is given by
\begin{equation}
\label{eq:ytilde_occ}
\widetilde{\yvet}^c = \Psivet^\top\widehat{\yvet} = \Mvet\Omegavet^\top\widehat{\yvet},
\end{equation}
with weight matrix $\Psivet^\top = \Mvet\Omegavet^\top \in \mathbb{R}^{n \times m}$, where
\begin{equation}
		\Mvet  = \left[\Ivet_n - \Wvet_c\Cvet^\top\left(\Cvet\Wvet_c\Cvet^\top\right)^{-1}\Cvet\right],  \;
\Omegavet  = \Wvet^{-1}\Kvet\Wvet_c, \; \text{and} \;
\Wvet_c  = \left(\Kvet^\top\Wvet^{-1}\Kvet\right)^{-1}.\label{eq:Wc}
\end{equation}
%\begin{align}
%		\Mvet & = \left[\Ivet_n - \Wvet_c\Cvet^\top\left(\Cvet\Wvet_c\Cvet^\top\right)^{-1}\Cvet\right],\label{eq:Mvet}\\
%\Omegavet & = \Wvet^{-1}\Kvet\Wvet_c,\label{eq:Omegavet} \\
%\Wvet_c & = \left(\Kvet^\top\Wvet^{-1}\Kvet\right)^{-1}.\label{eq:Wc}
%\end{align}
\end{thm}
\begin{proof}\let\qed\relax
	See \autoref{app:proof}.
\end{proof}

\begin{rml}
The coherency property of $\widetilde{\yvet}^c$ can be easily verified by observing that $\Cvet\Mvet = \Zerovet_{(n_u \times n)}$. Then, it follows that $\Cvet\widetilde{\yvet}^c = \Cvet\Mvet\Omegavet^\top\widehat{\yvet} = \Zerovet_{(n_u \times 1)}$.	
\end{rml}

\begin{rml}\label{remark:globalforecast}
Denoting $\Psivet = \left[\Psivet_1 \; \ldots \; \Psivet_j \; \ldots \; \Psivet_p\right]^\top$, with $\Psivet_j \in \mathbb{R}^{n \times n_j}$, $j=1, \ldots, p$, the linear minimum variance coherent combined forecast vector can be expressed as
\begin{equation}
\label{eq:ytilde_matrx_weights}
\widetilde{\yvet}^c = \Psivet_1\widehat{\yvet}^1 + \ldots + \Psivet_j\widehat{\yvet}^j + \ldots + \Psivet_p\widehat{\yvet}^p =
\displaystyle\sum_{j=1}^{p}\Psivet_j\widehat{\yvet}^j .
\end{equation}
It is worth noting that, due to the unbiasedness of the base forecasts, the weight matrix $\Psivet_j$ has the following interesting property\footnote{In fact, since $E\left(\widetilde{\yvet}^c\right) = \Psivet^\top E\left(\widehat{\yvet}\right) = \Psivet^\top\Kvet E\left(\yvet\right)=E\left(\yvet\right)$, it follows that $\Psivet^\top\Kvet = \Ivet_n$, corresponding to Eq.~(\ref{eq:Psisumto1}).}:
\begin{equation}
	\label{eq:Psisumto1}
	\displaystyle\sum_{j=1}^{p}\Psivet_j\Lvet_j = \Ivet_n .
\end{equation}
In the balanced case, i.e. $n_j = n$ $\forall j$ and $p_i = p$ $\forall i$, Eq.~(\ref{eq:Psisumto1}) simplifies to $\displaystyle\sum_{j=1}^{p}\Psivet_j = \Ivet_n$. Again in this case, looking at the entries of the square weight matrices $\Psivet_j$ as
$$
\Psivet_j = \begin{bmatrix}[0.85]
	\psi_{11,j} & \ldots & \psi_{1i,j} & \ldots & \psi_{1n,j}\\[-0.4ex]
	\vdots        & \ddots & \vdots        & \ddots & \vdots \\[-0.4ex]
	\psi_{i1,j} & \ldots & \psi_{ii,j} & \ldots & \psi_{in,j}\\[-0.4ex]
	\vdots        & \ddots & \vdots        & \ddots & \vdots \\[-0.4ex]
	\psi_{n1,j} & \ldots & \psi_{ni,j} & \ldots & \psi_{nn,j}
\end{bmatrix}, \quad j=1, \ldots, p,
$$
the reconciled forecast of the $i$-th variable may be expressed as
\begin{equation}\label{eq:ccf_split}
	\widetilde{y}_i^c = \displaystyle\sum_{j=1}^{p}\sum_{k=1}^{n}\psi_{ik,j}\widehat{y}_k^j =
	\displaystyle\sum_{j=1}^{p}\psi_{ii,j}\widehat{y}_i^j +
	\displaystyle\sum_{j=1}^{p}
	\sum_{\substack{k=1 \\ k \ne i}}
	^{n}\psi_{ik,j}\widehat{y}_k^j, \quad i=1, \ldots, n.
\end{equation}
According to Eq.~(\ref{eq:ccf_split}), the multi-task coherent combined forecast $\widetilde{y}_i^c$ is the sum of a local forecast combination \citep{Thompson2024}, $\displaystyle\sum_{j=1}^{p}\psi_{ii,j}\widehat{y}_i^j$, computed by using only the base forecasts of the $i$-th variable, plus an adjustment term that takes into account the base forecasts of all the remaining variables. This suggests that $\widetilde{y}_i^c$ can be interpreted as the result of a simple global linear forecast combination method \citep{Thompson2024}.
\end{rml}

\begin{rml}\label{remark:twotransform}
Another interesting interpretation of $\widetilde{\yvet}^c$  is obtained by noting that Eq.~(\ref{eq:ytilde_occ}) may be re-stated as $\widetilde{\yvet}^c = \Mvet\widehat{\yvet}^c$, where
\begin{equation}
\label{eq:yhatc}
\widehat{\yvet}^c = \Omegavet^\top\widehat{\yvet} .
\end{equation}
\cite{Sun2004} and \cite{Lavancier2016} proved that $\widehat{\yvet}^c$ is the unbiased\footnote{See \autoref{app:proof} for the proof of the unbiasedness of $\widehat{\yvet}^c$.} linear minimum variance combination forecast of $\yvet$, with the error covariance matrix given by Eq.~(\ref{eq:Wc}), i.e.,
$\Wvet_c = E\left[\big(\widehat{\yvet}^c - \yvet\big)\big(\widehat{\yvet}^c - \yvet\big)^\top\right] = \left(\Kvet^\top\Wvet^{-1}\Kvet\right)^{-1}$.
In addition, it should be noted that $\widehat{\yvet}^c$ is in general incoherent, i.e., $\Cvet\widehat{\yvet}^c \ne \Zerovet_{(n_u \times 1)}$, and matrix $\Mvet \in \mathbb{R}^{n \times n}$ is an oblique projector from $\mathbb{R}^{n}$ into the linear sub-space  $\mathcal{S} \subset \mathbb{R}^n$ spanned by the relationship $\Cvet\yvet=\Zerovet_{(n_u \times 1)}$. Vector $\widetilde{\yvet}^c$ can thus be seen as the projection of the incoherent forecast vector $\widehat{\yvet}^c$  into $\mathcal{S} = \{\yvet \in \mathbb{R}^n \; | \; \Cvet\yvet = \Zerovet_{(n_u \times 1)}\}$ \citep{Panagiotelis2021}. More precisely, $\widetilde{\yvet}^c$ is the result of two successive transformations, the first mapping $\widehat{\yvet} \in \mathbb{R}^m$ into the combined forecast vector  $\widehat{\yvet}^c \in \mathbb{R}^n$, the second projecting $\widehat{\yvet}^c$ into the coherent subspace $\mathcal{S}$. The key point is that the error covariance matrix used in the latter projection is related to the error covariance matrix used in the former transformation, i.e., $\Wvet_c = \Fvet\left(\Wvet\right)$, where for any non-singular matrix $\Xvet \in \mathbb{R}^{m \times m}$, $\Fvet : \mathbb{R}^{m \times m} \rightarrow \mathbb{R}^{n \times n}$ is the matrix function $\Fvet\left(\Xvet\right) = \left(\Kvet^\top\Xvet^{-1}\Kvet\right)^{-1}$.

\noindent It is worth emphasizing that the factorization $\Psivet^\top = \Mvet\Omegavet^\top$ established in Theorem~\ref{thm:mmse} is a property of the closed-form solution (\ref{eq:ytilde_occ}), not a sequential algorithm: applying a local combination followed by a standard reconciliation projection does not return $\Psivet$ in general, because the joint information carried by $\Wvet$ would be partially discarded. Only for particular structures of $\Wvet$ does $\Psivet$ admit a genuine two-step decomposition equivalent to the sequential approaches (see \autoref{app:case}).\footnote{A similar result in the single-expert cross-temporal reconciliation framework is shown by \cite{Girolimetto2026}.}
\end{rml}

We now explore the statistical properties of $\widetilde{\yvet}^c$, namely its unbiasedness and the ordering of its error covariance matrix with respect to those of $\widehat{\yvet}^j$ and $\widehat{\yvet}^c$.

\begin{crl}[Unbiasedness of $\widetilde{\yvet}^c$ and a property of its error covariance matrix]\label{crl:unb}
Denoting $\muvet = E\left(\yvet\right)$, and provided that the base forecasts are unbiased, the combined forecast vector $\widetilde{\yvet}^c$ is unbiased, i.e., $E\left(\widetilde{\yvet}^c\right) = \muvet$, with error covariance matrix equal to
\begin{equation}
	\label{eq:Wtilde}
	\widetilde{\Wvet}_c = E\left[\left(\widetilde{\yvet}^c - \yvet\right)\left(\widetilde{\yvet}^c - \yvet\right)^\top\right] = \Mvet\Wvet_c .
\end{equation}
In addition,
\begin{equation}
\label{eq:Wtildebest}
		\Lvet_j\widetilde{\Wvet}_c\Lvet_j^\top \preceq \Lvet_j\Wvet_c\Lvet_j^\top \preceq \Wvet_{j}, \quad j=1,\ldots,p.
\end{equation}
\end{crl}
\begin{proof}\let\qed\relax
	See \autoref{app:proof}.
\end{proof}

\begin{rml} Expression (\ref{eq:Wtildebest}) states that the base forecasts produced by a single expert ($\widehat{\yvet}^j$) are always not better (i.e., the error covariance matrix is not `smaller') than the corresponding multi-task combined forecast vector $\Lvet_j\widehat{\yvet}^c$, which in turn is not better than the corresponding subset of $\widetilde{\yvet}^c$, i.e., $\Lvet_j\widetilde{\yvet}^c$. Moreover, since $\Lvet_j = \Ivet_n$ when $n_j=n$, then in the balanced case we have $\widetilde{\Wvet}_c \preceq \Wvet_c \preceq  \Wvet_{j}$. In other terms---assuming that the base forecasts are unbiased and the error covariance matrix $\Wvet$ is known---simultaneously considering multiple incoherent forecasts and the constraints operating on the components of a multivariate time series does not worsen (and in fact, hopefully improves on) the precision of both the original (base) and combined (incoherent) forecasts.
\end{rml}

 \autoref{crl:unb} compares $\widetilde{\yvet}^c$ pairwise with the base forecasts and with the incoherent combined forecast $\widehat{\yvet}^c$. A stronger optimality property holds within the entire class of linear unbiased coherent estimators of $\yvet$, as stated in the following corollary.

\begin{crl}[Best linear unbiased coherent estimator]\label{crl:blue}
For any linear unbiased coherent estimator $\check{\yvet}$ of $\yvet$, the difference 
$$
E\left[\left(\check{\yvet} - \yvet\right)\left(\check{\yvet} - \yvet\right)^\top\right] - E\left[\left(\widetilde{\yvet}^c - \yvet\right)\left(\widetilde{\yvet}^c - \yvet\right)^\top\right]
$$ 
is nonnegative definite. Therefore, $\widetilde{\yvet}^c$ is the best linear unbiased coherent estimator of $\yvet$ in the sense of the nonnegative-definite partial ordering between forecast-error dispersion matrices.
\end{crl}
\begin{proof}\let\qed\relax
See \autoref{app:proof}.
\end{proof}

\begin{table}[!tb]
	\centering
	\footnotesize
	\caption{Equivalent formulations of Theorem \ref{thm:mmse} distinct by model representation and organization of base forecasts. The symbols are fully described in \autoref{tab:notation}. The proof for the zero-constrained representation with a by-expert formulation is in \autoref{app:proof}. The remaining proofs are in \autoref{app:alternative_proofs}.}\label{tab:equiv_solutions}
		\begin{tabular}{cc|p{0.415\linewidth}|p{0.415\linewidth}}
		\toprule
		& & \multicolumn{2}{c}{\textbf{Model representation}}\\
		& & \multicolumn{1}{c|}{\textit{zero-constrained}} & \multicolumn{1}{c}{\textit{structural}} \\
		\midrule
		\parbox[t]{4mm}{\multirow{13}{*}{\rotatebox[origin=c]{90}{\textbf{Base forecasts' organization}}}} & \parbox[t]{4mm}{\rotatebox[origin=c]{90}{\textit{by-expert}}} &
		$\begin{aligned}
			&\widehat{\yvet} = \Kvet\yvet + \epsvet, \; \text{s.t. }\Cvet\yvet = \Zerovet \\[-0.5ex]
			&\begin{cases}
				\widetilde{\yvet}^c = \displaystyle\argmin_{\yvet} \left(\widehat{\yvet} - \Kvet\yvet\right)^\top\Wvet^{-1}\left(\widehat{\yvet} - \Kvet\yvet\right)\\[-0.5ex]
				\text{s.t. }\Cvet\yvet = \Zerovet
			\end{cases}\\
			&\widetilde{\yvet}^c =\Mvet\widehat{\yvet}^c = \Mvet\Omegavet^\top\widehat{\yvet}\\[-0.5ex]
			&\Mvet = \left[\Ivet_n - \Wvet_c\Cvet^\top\left(\Cvet\Wvet_c\Cvet^\top\right)^{-1}\Cvet\right] \\[-1ex]
			&\Omegavet = \Wvet^{-1}\Kvet\Wvet_c, \quad \Wvet_c = \left(\Kvet^\top\Wvet^{-1}\Kvet\right)^{-1} \\[-1ex]
			& \widetilde{\Wvet}_c = \Mvet\Wvet_c \\[-0.5ex]
			&\Lvet_j\widetilde{\Wvet}_c\Lvet_j^\top \preceq \Lvet_j\Wvet_c\Lvet_j^\top \preceq \Wvet_{j}, \quad j=1, \ldots, p
		\end{aligned}$
		& 
		$\begin{aligned}
			&\widehat{\yvet} = \Kvet\Svet\bvet + \epsvet \\[-0.5ex]
			&\begin{cases}
				\Gvet = \displaystyle\argmin_{\Gvet} \text{tr}\left(\Svet\Gvet\Wvet\Gvet^\top \Svet^\top\right)\\[-0.5ex]
				\text{s.t. }\Gvet\Kvet\Svet = \Ivet_{n_b}
			\end{cases}\\[-0.5ex]
			&\widetilde{\yvet}^c = \Svet\Gvet\widehat{\yvet} = \Svet\widetilde{\bvet}, \quad \widetilde{\bvet} = \Gvet\widehat{\yvet}\\[-1ex]
			&\Gvet = \left(\Svet^\top\Wvet^{-1}_c\Svet\right)^{-1}\Svet^\top\Kvet^\top\Wvet^{-1} \\[-1ex]
			&\Wvet_c = \left(\Kvet^\top\Wvet^{-1}\Kvet\right)^{-1}\\[-1ex]
			&\widetilde{\Wvet}_c = \Svet\left(\Svet^\top\Wvet^{-1}_c\Svet\right)^{-1}\Svet^\top\\[-0.5ex]
			&\Lvet_j\widetilde{\Wvet}_c\Lvet_j^\top \preceq \Wvet_{j}, \quad j=1, \ldots, p
		\end{aligned}$ \\
		\cmidrule{2-4}
		& \parbox[t]{4mm}{\rotatebox[origin=c]{90}{\textit{by-variable}}} & 
		$\begin{aligned}
			&\widehat{\yvet}_{\text{bv}} = \Jvet\yvet + \epsvet_{\text{bv}}, \; \text{s.t. }\Cvet\yvet = \Zerovet \\[-0.5ex]
			&\begin{cases}
				\widetilde{\yvet}^c = \displaystyle\argmin_{\yvet} \left(\widehat{\yvet}_{\text{bv}} - \Jvet\yvet\right)^\top\Sigmavet^{-1}\left(\widehat{\yvet}_{\text{bv}} - \Jvet\yvet\right)\\[-0.5ex]
				\text{s.t. }\Cvet\yvet = \Zerovet
			\end{cases}\\
			&\widetilde{\yvet}^c = \Mvet\widehat{\yvet}^c = \Mvet\Gammavet^\top\widehat{\yvet}_{\text{bv}}\\[-0.5ex]
			&\Mvet = \left[\Ivet_n - \Sigmavet_c\Cvet^\top\left(\Cvet\Sigmavet_c\Cvet^\top\right)^{-1}\Cvet\right] \\[-1ex]
			&\Gammavet = \Sigmavet^{-1}\Jvet\Sigmavet_c \\[-1ex]
			&\Sigmavet_c = \left(\Jvet^\top\Sigmavet^{-1}\Jvet\right)^{-1}, \quad
			\widetilde{\Wvet}_c = \Mvet\Sigmavet_c \\[-1ex]
			&\Lvet_j\widetilde{\Wvet}_c\Lvet_j^\top \preceq \Lvet_j\Wvet_c\Lvet_j^\top \preceq \Wvet_{j}, \quad j=1, \ldots, p
		\end{aligned}$
		&
		$\begin{aligned}
			&\widehat{\yvet}_{\text{bv}} = \Jvet\Svet\bvet + \epsvet_{\text{bv}} \\[-0.5ex]
			&\begin{cases}
				\Gvet_{\text{bv}} = \displaystyle\argmin_{\Gvet} \text{tr}\left(\Svet\Gvet\Sigmavet\Gvet^\top \Svet^\top\right)\\[-0.5ex]
				\text{s.t. }\Gvet\Jvet\Svet = \Ivet_{n_b}
			\end{cases}\\[-0.5ex]
			&\widetilde{\yvet}^c = \Svet\Gvet_{\text{bv}}\widehat{\yvet}_{\text{bv}} = \Svet\widetilde{\bvet}, \quad \widetilde{\bvet}=\Gvet_{\text{bv}}\widehat{\yvet}_{\text{bv}}\\[-1ex]
			&\Gvet_{\text{bv}} = \left(\Svet^\top\Sigmavet_c^{-1}\Svet\right)^{-1}\Svet^\top\Jvet^\top\Sigmavet^{-1}\\[-1ex]
			&\Sigmavet_c = \left(\Jvet^\top\Sigmavet^{-1}\Jvet\right)^{-1}\\[-1ex]
			&\widetilde{\Wvet}_c = \Svet\left(\Svet^\top\Sigmavet_c^{-1}\Svet\right)^{-1}\Svet^\top\\[-0.5ex]
			&\Lvet_j\widetilde{\Wvet}_c\Lvet_j^\top \preceq \Wvet_{j}, \quad j=1, \ldots, p
		\end{aligned}$ \\
		\bottomrule
	\end{tabular}\\
	\vspace*{-\baselineskip}
\end{table}

\begin{rml}
All these results were proved by adopting a constrained data fitting approach with base forecasts organized by expert. Indeed, the same result may be obtained in other ways, depending on the base forecast organization (by expert or by variable), and the problem formulation (constrained data fitting problem or unconstrained structural form). In these equivalent proofs, summarized in \autoref{tab:equiv_solutions} and developed in \autoref{app:alternative_proofs}, the choice of the forecast organization gives a different emphasis to the `objects' that are combined in the final formula, i.e., base forecasts either from the $j$-th expert, $\widehat{\yvet}^j$, collected into $\widehat{\yvet}$, or of the $i$-th variable $\widehat{\yvet}_i$, concatenated into vector $\widehat{\yvet}_{\text{bv}}$.

\noindent The difference between zero-constrained and structural approaches 
is more interesting from an interpretive point of view. Each approach has %in fact, 
some distinctive interesting features.
\begin{itemize}[nosep, leftmargin=*]
\item The structural approach is based on a linear model that directly incorporates the constraints on the final forecasts, whose solution is found by minimizing the trace of the error covariance matrix $\widetilde{\Wvet}_c$. For a similar result in the cross-sectional forecast reconciliation, see \cite{Wickra2019} and \cite{Ando2024}. The coherent combined forecasts, either $\widetilde{\yvet}^c = \Svet\Gvet\widehat{\yvet}$ or $\widetilde{\yvet}^c = \Svet\Gvet_{\text{bv}}\widehat{\yvet}_{\text{bv}}$, depend on the organization of the adopted base forecasts. In both cases, first the combination forecast for a set of `free' variables is computed, i.e., $\widetilde{\bvet}^c = \Gvet\widehat{\yvet} = \Gvet_{\text{bv}}\widehat{\yvet}_{\text{bv}}$, and then the %whole 
vector of coherent forecasts is obtained in a bottom-up fashion through pre-multiplication by the structural matrix $\Svet$: $\widetilde{\yvet}^c = \Svet\widetilde{\bvet}^c$.
\item As shown in \autoref{remark:twotransform}, the zero-constrained approach splits the coherent forecast combination methodology into two strictly linked steps. In the first, the linear minimum variance combination of multiple forecast vectors is computed, i.e., $\widehat{\yvet}^c$. In the second, this incoherent forecast vector is reconciled through an oblique projection onto a coherent subspace of $\mathbb{R}^n$, making use of $\Wvet_c$, the error covariance matrix of $\widehat{\yvet}^c$.
\end{itemize}
\end{rml}

Taken together, \autoref{thm:mmse} and Corollaries \ref{crl:unb} and \ref{crl:blue} characterize $\widetilde{\yvet}^c$ as the best linear unbiased coherent estimator of $\yvet$, under the commonly assumed hypotheses that the base forecasts are unbiased and that the joint dispersion matrix of the $p$ forecast experts is known.

\section{On the covariance matrix used by the coherent combination method}
\label{sec:on the covmat}
The choice of the forecast errors' covariance matrix $\Wvet$ greatly influences the combination and reconciliation processes, as its properties determine how the base forecasts are combined, and then the nature of the final coherent forecasts. In the following, we first consider two notable patterns of this matrix, where the base forecasts' error are assumed to be uncorrelated either across experts or across variables. Then, we deal with practical estimation issues, following an approach widely used in the context of forecast reconciliation, which makes use of the in-sample base forecasts' errors, when available, to estimate matrix $\Wvet$.

\subsection{Uncorrelated forecast errors: Block-diagonal $\mathbf{W}$ and $\mathbf{\Sigma}$}\label{sec:bdW}

If the forecast errors are uncorrelated across the $p$ experts, the covariance matrix $\Wvet$ has a block-diagonal form: $\Wvet = \text{Diag}\left(\Wvet_1,\ldots,\Wvet_j,\ldots,\Wvet_p\right)$. In this case, $\widetilde{\yvet}^c$ is found as the solution to the linearly constrained quadratic program
\begin{equation}
	\label{eq:lossWbd}
	\widetilde{\yvet}^c = \argmin_{\yvet} \sum_{j = 1}^p
	\left(\widehat{\yvet}^j - \yvet\right)^\top\Wvet_{j}^{-1}\left(\widehat{\yvet}^j - \yvet\right) \quad
	\text{s.t. } \Cvet\yvet = \Zerovet_{(n_u \times 1)},
\end{equation}
where the global loss function in Eq.~(\ref{eq:lcquadprog_be}) simplifies to the sum of $p$ quadratic %loss 
functions.

According to the traditional single-task forecast combination approach \citep{Bates1969,Thompson2024}, for each single variable $i$ the combination weights are obtained by exploiting the relationships between $\widehat{y}_i^{j}$ and $\widehat{y}_i^{l}$, with $\{j, l\} \in \big\{(j, l) \mid 1\leq j, l \leq p, \; j \neq l\big\}$. We show that this is a particular case of our general framework. To this end, it is convenient to adopt a by-variable organization of the base forecasts (see \autoref{sec:be_vs_bv}), that is, $\widehat{\yvet}_{\text{bv}} = \begin{bmatrix}[0.7] \widehat{\yvet}_1^\top & \ldots & \widehat{\yvet}_i^\top & \ldots & \widehat{\yvet}_n^\top \end{bmatrix}^\top \in \mathbb{R}^m$, where $\widehat{\yvet}_i \in \mathbb{R}^{p_i}$, $i=1,\ldots,n$, is the vector of the base forecasts available for the $i$-th variable.

Denote by $\Sigmavet = E\left[\epsvet_{\text{bv}}\epsvet_{\text{bv}}^\top\right]\in \mathbb{R}^{m \times m}$ the covariance matrix of the forecast errors organized by variable $\widehat{\epsvet}_{\text{bv}} = \begin{bmatrix}[0.7] \epsvet_1^\top & \ldots & \epsvet_i^\top & \ldots & \epsvet_n^\top \end{bmatrix}^\top \in \mathbb{R}^m$, and $\Jvet = \Pvet \Kvet \in \{0,1\}^{m \times n}$. Noting that $\Sigmavet = \Pvet\Wvet\Pvet^\top$, equivalent by-variable formulations of $\widetilde{\yvet}^c$ in Eq.~(\ref{eq:ytilde_occ}) are presented in \autoref{tab:equiv_solutions}. If we assume uncorrelated errors across the variables, the covariance matrix has a block-diagonal structure, i.e., $\Sigmavet = \mbox{Diag}\left(\Sigmavet_{1}, \dots, \Sigmavet_{i}, \dots, \Sigmavet_{n}\right)$, where $\Sigmavet_{i}$, $i=1,\ldots,n$, is the $(p_i \times p_i)$ forecast error covariance matrix for the $i$-th variable. In this case, after some simple calculations, we obtain $\widehat{y}^c_i = \displaystyle\frac{\Unovet_{p_i}^\top\Sigmavet^{-1}_{i}}{\Unovet_{p_i}^\top\Sigmavet^{-1}_{i}\Unovet_{p_i}}\widehat{\yvet}_{i}$, $i=1,\ldots,n$. In summary, if the forecast errors are uncorrelated across variables, each entry of the global combined forecast vector $\widehat{\yvet}^c$ corresponds to the single-variable optimally combined forecast: $\widehat{y}_{i}^c = \gammavet_i^\top\widehat{\yvet}_{i}$, where $\gammavet_i = \displaystyle\frac{\Sigmavet^{-1}_{i}\Unovet_{p_i}}{\Unovet_{p_i}^\top\Sigmavet^{-1}_{i}\Unovet_{p_i}}$, $i=1, \ldots, n$, a well-known result dating back to \cite{Newbold1974}. However, when $\widetilde{\yvet}^c$ is computed through the projection of the incoherent combined forecast vector $\widehat{\yvet}^c$  onto the linear space spanned by the constraints, i.e., when $\widehat{\yvet}^c$ is pre-multiplied by the projection matrix $\Mvet$, the coherent combined forecast for each \textbf{single} variable is computed by combining the base forecasts of \textbf{all} variables. In other terms, whereas $\widehat{y}_i^c$ is the result of a local forecast combination, its coherent counterpart $\widetilde{y}_i^c$ is obtained as a global forecast combination.
		
\subsection{Feasible estimates of the covariance matrix}
\label{sec:Wfeasible}
In the previous sections, we always considered the covariance matrix as known. However, in practice this rarely happens, and this matrix must be estimated somehow. For a forecast horizon $h=1$, it seems sensible to exploit (if available) the in-sample forecast errors (residuals), given by\footnote{Equivalently, one may use the by-variable version of the in-sample residuals, i.e., $\widehat{\epsvet}_{i,t} = y_{it}\Unovet_{p_i} - \widehat{\yvet}_{i,t} \in \mathbb{R}^{p_i}$.}
$
\widehat{\epsvet}_{t}^j = \Lvet_j\yvet_{t} - \widehat{\yvet}_{t}^j \in \mathbb{R}^{n_j}$, $j=1,\ldots,p,
$
where $t = 1, \ldots, T$ is the time index running on the training set used to estimate the base forecasts. This is a rather common practice in forecast reconciliation \citep{Hyndman2016, Wickra2019}, particularly when the available observations of the time series to be forecast are not long, and splitting the dataset %into training, validation and test 
may result in a validation set that is too short. As for the classical combination, the use of an in-sample fit to determine the weights of the combination is not uncommon \citep{Kapetanios2008-lc, Banbura2010-es}.%, Cobb2018}.

Denoting by $\widehat{\epsvet}_t = \begin{bmatrix}[0.7] \widehat{\epsvet}_{t}^1{^{\top}} & \ldots & \widehat{\epsvet}_{t}^j{^{\top}} & \ldots & \widehat{\epsvet}_{t}^p{^{\top}} \end{bmatrix}^\top$, $t=1, \ldots, T$, the $(m \times 1)$ vector containing the in-sample forecast errors, a natural estimator of $\Wvet$ (and $\Sigmavet$) is the sample forecast MSE matrix: $\widehat{\Wvet} = \displaystyle\frac{1}{T}\displaystyle\sum_{t=1}^{T}\widehat{\epsvet}_t\widehat{\epsvet}_t^\top$ (and $\widehat{\Sigmavet} = \Pvet\widehat{\Wvet}\Pvet^\top$). When dealing with $h>1$, $\widehat{\Wvet}_h$ is challenging to estimate (e.g., multi-step-ahead errors will not be independent) and we assume that $\widehat{\Wvet}_h \propto \widehat{\Wvet}$ as proposed by \cite{Wickra2019} and \cite{Ben_Taieb2021}. Moreover, when $m \gg T$, $\widehat{\Wvet}$ and $\widehat{\Sigmavet}$ are singular, and some regularization has to be adopted in order to recover a stable, non-singular covariance matrix. The natural choice is to consider the shrunk versions of the sample covariance matrices $\widehat{\Wvet}$ and $\widehat{\Sigmavet}$ towards their diagonal versions. That is,
\begin{equation}
	\label{eq:Wshr}
	\widehat{\Wvet}_{\text{shr}} = \widehat{\lambda}\left(\Ivet_m \odot \widehat{\Wvet}\right) + \left(1 - \widehat{\lambda}\right)\widehat{\Wvet} \; \rightarrow \;
	\widehat{\Sigmavet}_{\text{shr}} = \Pvet \widehat{\Wvet}_{\text{shr}} \Pvet^\top,
\end{equation}
where $\odot$ denotes the Hadamard product and $\widehat{\lambda}$ is an estimate of the coefficient of shrinkage intensity $\lambda$, $0\leq \lambda\leq 1$ \citep{Schafer2005}.

An alternative approach involves assuming uncorrelated errors across either experts or variables, which results in block-diagonal covariance matrices (see \autoref{sec:bdW}). In the former case, $\widehat{\Wvet}_{\text{bd}} = \mbox{Diag}\left(\widehat{\Wvet}_{1}, \dots, \widehat{\Wvet}_{j}, \dots, \widehat{\Wvet}_{p}\right)$, with $\widehat{\Wvet}_{j} = \displaystyle\frac{1}{T}\sum_{t=1}^{T}\widehat{\epsvet}_t^j\widehat{\epsvet}_t^j{^{\top}} \in \mathbb{R}^{n_j \times n_j}$, $j=1,\ldots,p$. In the latter case, $\widehat{\Sigmavet}_{\text{bd}} = \mbox{Diag}\left(\widehat{\Sigmavet}_{1}, \dots, \widehat{\Sigmavet}_{i}, \dots, \widehat{\Sigmavet}_{n}\right)$, with $\widehat{\Sigmavet}_{i} =  \displaystyle\frac{1}{T}\sum_{t=1}^{T}\widehat{\epsvet}_{i,t}\widehat{\epsvet}_{i,t}^\top \in \mathbb{R}^{p_i \times p_i}$, $i=1,\ldots,n$. When $n_j > T$ for some expert $j$, or $p_i > T$ for some variable $y_i$, the estimates $\widehat{\Wvet}_j$ and $\widehat{\Sigmavet}_i$ are singular. In this case, as usual, we can resort to the shrunk versions of these two matrices respectively given by
\begin{align*}
	\widehat{\Wvet}_{j,\text{shr}} & = \widehat{\lambda}_j\left(\Ivet_{n_j} \odot \widehat{\Wvet}_j \right) + \left(1 - \widehat{\lambda}_j\right)\widehat{\Wvet}_j, & j=1,\ldots,p\\
	\widehat{\Sigmavet}_{i,\text{shr}} & = \widehat{\nu}_i\left(\Ivet_{p_i} \odot \widehat{\Sigmavet}_i \right) + \left(1 - \widehat{\nu}_i\right)\widehat{\Sigmavet}_i, & i=1,\ldots,n
\end{align*}
where $\widehat{\lambda}_j$ and $\widehat{\nu}_i$ are estimates of the coefficients of shrinkage intensity $\lambda_j$ and $\nu_i$, respectively ($0\leq \lambda_j\leq 1$ and $0\leq \nu_i \leq 1$). Then, two other estimators can be considered: $\widehat{\Wvet}_{\text{bd-shr}} = \mbox{Diag}\left(\widehat{\Wvet}_{1,\text{shr}}, \dots, \widehat{\Wvet}_{p,\text{shr}}\right)$ and $\widehat{\Sigmavet}_{\text{bd-shr}} = \mbox{Diag}\left(\widehat{\Sigmavet}_{1,\text{shr}}, \dots, \widehat{\Sigmavet}_{n,\text{shr}}\right)$.

The expressions above define four feasible specifications, which we denote by occ$_{\text{shr}}$ for the full shrunk matrix, occ$_{\text{be}}$ and occ$_{\text{bv}}$ for the by-expert and by-variable block-diagonal shrunk matrices, respectively, and occ$_{\text{wls}}$ for the diagonal matrix $\Ivet_m \odot \widehat{\Wvet}$. Two of them are of limited interest in applications: as shown in \autoref{sec:bdW}, a by-variable block-diagonal covariance matrix makes the combination step local, so that occ$_{\text{bv}}$ and its special diagonal case occ$_{\text{wls}}$ reduce to a sequential scheme in which the base forecasts of each variable are first combined locally and the resulting vector is then reconciled. Of the two remaining specifications, the full matrix requires the estimation of the cross-expert blocks, whose number grows with $p^2$ and which are the least reliably estimated when $T$ is short. We therefore recommend the by-expert block-diagonal matrix as the default choice, as it retains the correlation among the forecast errors of different variables produced by the same expert, which is the dependence the coherence constraints exploit, at a moderate estimation cost. The results of Sections~\ref{sec:SimulationExperiment} and \ref{sec:energy}, reported in full in Appendices~\ref{app:sim} and \ref{app:fexp}, support this recommendation: occ$_{\text{be}}$ is the best or tied-best specification in every simulation setting and in the empirical application.

\section{Simulation experiment}
\label{sec:SimulationExperiment}
In order to assess the performance of the proposed approach, we ran a simulation experiment. The simulation framework, designed for a hierarchical multivariate time series, extends the univariate simulation experiment of \cite{Capistran2009} to a multivariate setting that accommodates for cross-sectional dependencies and linear constraints across series. The hierarchical structure of the series is the one on the left panel of \autoref{fig:example}, including a total of  $n = 7$  variables, with  $n_b = 4$  bottom series and with aggregation and structural matrices $\Avet$ and $\Svet$, respectively, given by Eq.~(\ref{eq:SimpleExampleMatrices}).

Denote by $t = 0, \dots, (T-1)$,  and $i = 1, \dots, n_b$, the indices running on time periods and bottom series, respectively. The observed bottom time series are simulated from a simple two-factor model\footnote{This formulation is equivalent to Eq.~(9) in \cite{Capistran2009}, where $\mu_y = 0$ and $\beta_{y1}=\beta_{y2}= 1$ in all settings. We omit these parameters to simplify the notation.}: $b_{i,t+1} = F_{1, i, t+1} + F_{2, i, t+1} + \eta_{i, t+1}$, $i= 1, \ldots, n_b$, where factors $F_{1,i}$ and $F_{2,i}$ come from a vector autoregressive model of order 1, $\Fvet_{i, t+1} = \Phivet_F\Fvet_{i, t} + \xivet_{i,t+1}$, characterized by a diagonal coefficient matrix $\Phivet_F$ and independent innovations $ \xivet_{i,t+1} \sim \mathcal{N}\left(\Zerovet_{(2 \times 1)}, \Ivet_2\right)$. In addition, $\etavet_{t+1} = \begin{bmatrix}[0.7] \eta_{1, t+1} & \dots & \eta_{n_b, t+1} \end{bmatrix}^\top \sim \mathcal{N}_{n_b}(\Zerovet_{n_b \times 1}, \sigma^2_{\eta}\Rvet)$, where $\sigma^2_{\eta} = 1$ and $\Rvet$ is the correlation matrix of the bottom series obtained as the closest positive definite matrix in the Frobenius norm to $\check{\Rvet}$, with elements $\rho_{i,k}$ distributed as $U(-1,1)$, where $U\left(-1,1\right)$ denotes a uniform distribution in $(-1,1)$, and $\check{\rho}_{i,i} = 1$, $i,k = 1, \ldots, n_b$. The $n$-variate linearly constrained time series $\yvet_{t+1}$ is given by $\yvet_{t+1} = \Svet \bvet_{t+1}$, where $\bvet_{t+1} = \begin{bmatrix}[0.7] b_{1,t+1} & \dots & b_{n_b,t+1} \end{bmatrix}^\top$.

To generate the base forecasts from different experts, the systematic part of the forecast from the $j$-th expert for the $i$-th bottom variable at time $t+1$ is obtained as $\check{b}^{j}_{i,t+1} = \mu_j + \beta_{j,1}F_{1,i,t+1} + \beta_{j,2} F_{2, i, t+1}$, $i=1, \ldots, n_b$, $j=1, \ldots, p$, where $\betavet_j = \begin{bmatrix}[0.7] \beta_{j,1} & \beta_{j,2} \end{bmatrix}^\top$ is the vector of factor loadings for the series components. The base forecast vector $\widehat{\yvet}^{j}_{t+1}$ produced by the $j$-th expert for the complete set of $n$ variables (i.e., the balanced case) is
\begin{equation}\label{eq:factory}
	\widehat{\yvet}^{j}_{t+1} = \Svet  \check{\bvet}^{j}_{t+1} + \bm{\varepsilon}^{j}_{t+1} ,
\end{equation}
where $\bm{\varepsilon}^{j}_{t+1} = \begin{bmatrix} \varepsilon_{1, t+1}^{j} & \dots & \varepsilon_{n, t+1}^{j} \end{bmatrix}^\top \sim \mathcal{N}_n(\Zerovet_{(n \times 1)}, \Gammavet_{\varepsilon})$ and $\Gammavet_{\varepsilon} = \Dvet^{1/2}\Thetavet\Dvet^{1/2}$, with variance proportional to the number of bottom variables involved in the aggregation defining the nodes at each  hierarchy level, i.e., $\Dvet = \text{Diag}(\sigma_j^2 \Svet\Unovet_{n_b})$.%, that in our setting is equal to 
%$$
%\Dvet = \sigma^2_j\begin{bmatrix}[0.9]
%4 & 0 & 0 & 0 & 0 & 0 & 0 \\
%0 & 2 & 0 & 0 & 0 & 0 & 0 \\
%0 & 0 & 2 & 0 & 0 & 0 & 0 \\
%0 & 0 & 0 & 1 & 0 & 0 & 0 \\
%0 & 0 & 0 & 0 & 1 & 0 & 0 \\
%0 & 0 & 0 & 0 & 0 & 1 & 0 \\
%0 & 0 & 0 & 0 & 0 & 0 & 1
%\end{bmatrix}, \quad j =1, \ldots, p,
%$$
There are two covariance structures: either uncorrelated errors\footnote{Due to space constraints, the results with uncorrelated errors are provided in \autoref{app:sim}. The main conclusions drawn do not differ substantially from those valid for the correlated case.}, i.e., $\Thetavet = \Ivet_n$, or  $\Thetavet$ equal to the closest positive definite matrix to $\check{\Thetavet}$ with entries $\check{\theta}_{i,k} = U(-1, 1)$, $i\neq k$, and $\check{\theta}_{i,i} = 1$, $i,k=1, \ldots, n$. 

\begin{table}[!tbp]
	\centering
		\caption{Each setting specifies values for the parameters in the multivariate forecast model, including mean terms $\mu_j$, factor loadings $\betavet_j$, error variance $\sigma_j^2$, autoregressive matrix $\Phivet_F$, and correlation $\rho_{i,k}$, $i \neq k$, for the bottom time series. Settings vary systematically to assess model performance under different assumptions, including uniform ($U$), beta ($B$), inverse gamma ($IG$), and normal ($\mathcal{N}$) distributions. Settings 1--5 replicate \cite{Capistran2009}, while Setting 6 includes a bias component $\mu_j \neq 0$ for each expert.}
	\label{tab:sett}
	\footnotesize
	%	\color{blue}
	\begin{tabular}{c|cccc}
		\toprule
		Setting & $\mu_j$ & $\betavet_j$& $\sigma_j^2$ & $\Phivet_F$ \\
		\midrule
		1 & 0 &$\Unovet_2$ & 1 & $\Zerovet_{2 \times 2}$ \\[0.5em]
		2 & 0 & $0.5 \cdot\Unovet_2$ & 1 & $\Zerovet_{2 \times 2}$ \\[0.5em]
		3 & 0 & $0.5 \cdot\Unovet_2$ & 1 & $0.9 \cdot\Ivet_2$ \\[0.2em]
		4 & 0 & $\begin{bmatrix} B(1, 1) \\ B(1, 1) \end{bmatrix}$ &  1 & $\Zerovet_{2 \times 2}$ \\[1em]
		5 & 0 & $0.5 \cdot\Unovet_2$ & $IG(5,5)$ & $\Zerovet_{2 \times 2}$ \\[0.5em]
		6 & $\mathcal{N}(0,1)$ & $0.5 \cdot \Unovet_2$ & 1 & $\Zerovet_{2 \times 2}$ \\
		\bottomrule
	\end{tabular}
	\vspace*{-\baselineskip}
\end{table}

Following \cite{Capistran2009}, we assume that $E[\varepsilon_{i, t+1}^{j}\varepsilon_{k, t+1}^{l}] = 0$ for $j \neq l$, and that $E[\varepsilon_{n_u + i, t+1}^{j}\eta_{i, t+1}] = 0$,  $j = 1, \dots, p$ and $i = 1, \dots, n_b$. This setup represents a scenario where, at time $t$, forecasters receive noisy signals that are imperfectly correlated with future realizations of the factors $F_{1,i, t+1}$ and $F_{2,i, t+1}$. For the calculation of covariance matrices in both reconciliation and combination processes, we utilize $N$ observations, with $T = N + 100$, where the last $100$ observations are reserved as a test set. In each experimental configuration (see \autoref{tab:sett}), we examine different values for the number of observations, $N \in \{50, 100, 200\}$, and the number of experts, $p \in \{4, 10, 20\}$, running 500 replications per setting. % to ensure robustness in results. 
In addition, we consider two different frameworks as \cite{Capistran2009}:
\begin{itemize}[nosep, leftmargin=0.5cm]
	\item \textbf{A balanced panel of forecasts}: no missing values, allowing for consistent forecasts across all variables and experts.
	\item \textbf{An unbalanced panel of forecasts}: we classify experts according to their participation frequency, distinguishing between frequent (40\%) and infrequent participants. Using the transition probabilities proposed by \cite{Capistran2009}, we generate a binary participation matrix for each variable,\footnote{In the optimal coherent combination approach utilizing a block-diagonal shrunk error covariance matrix, the different number of residuals across experts and variables is addressed through the covariance shrinkage parameters $\lambda_j$, $j = 1, \dots, p$, as detailed in \autoref{sec:on the covmat}.} that is applied to the fully populated  matrix of forecasts simulated by Eq.~(\ref{eq:factory}). This results in a realistic, unbalanced panel structure that allows us to examine the impact of varying participation frequencies on the accuracy and robustness of different coherent combination forecast approaches.
\end{itemize}

These frameworks provide a comprehensive environment for testing forecast performance under different data conditions, allowing us to assess robustness across scenarios. The forecast accuracy is evaluated using the average relative mean absolute error (AvgRelMAE)\footnote{Detailed results, including a forecast evaluation using the average relative mean squared error, are presented in \autoref{app:sim}. Notably, the use of the squared error metric does not change the main conclusions.}:
$$
\text{AvgRelMAE}^{app} = \left(\prod_{s = 1}^{500} \prod_{i = 1}^{n} \frac{\text{MAE}_{s,i}^{app}}{\text{MAE}_{s,i}^{\text{ew}}}\right)^{\frac{1}{500n}} \; \text{with} \;\; \text{MAE}_{s,i}^{app} = \displaystyle\frac{1}{100}\sum_{q = 1}^{100} \left|y_{i,s,q}-\check{y}_{i,s,q}^{\,app}\right|,
$$
where $i=1,\ldots,n$ denotes the variable, `app' is the approach used, $y_{i,s,q}$ is the observed value, and $\check{y}_{i,s,q}^{\,app}$ is the forecast value using the `app' approach (either coherent or incoherent; see \autoref{tab:approaches}). 

\begin{table}[!tbp]
	\centering
		\caption{Summary of forecasting approaches used in the simulation (\autoref{sec:SimulationExperiment}) and in the forecasting experiment on Australian daily electricity generation time series (\autoref{sec:energy}).
		For single-expert reconciliation and coherent combination sequential approaches, the shrunk in-sample MSE matrix is used for the reconciliation: in scr it is computed from the in-sample residuals combined with the same weights applied to the base forecasts, while in src, each expert is reconciled separately, using its own in-sample residuals.}
	\label{tab:approaches}
	\footnotesize
	%\vspace*{-0.25\baselineskip}
	\setlength{\tabcolsep}{3pt}
	\resizebox{\linewidth}{!}{\begin{tabular}{c|c}
		\toprule
		\begin{minipage}[t]{0.48\linewidth}
			\textbf{Approach \& description}
		\end{minipage} &  \begin{minipage}[t]{0.52\linewidth}
			\textbf{Approach \& description}
		\end{minipage} \\
		\midrule
		\begin{minipage}[t]{0.48\linewidth}
			\textit{Base (incoherent forecasts)}
			\begin{itemize}[leftmargin=!, nosep, align = left, labelwidth=1.1cm]
				\item[base$^\ast$] Best base forecasts (\autoref{sec:SimulationExperiment})
				\item[tbats] Exponential smoothing state space model with Box--Cox transformation, ARMA errors, Trend and Seasonal components (\autoref{sec:energy})
			\end{itemize}
			\vspace{1ex}
			\hrule
			\vspace{1ex}
			\textit{Single-expert reconciliation}
			\begin{itemize}[leftmargin=!, nosep, align = left, labelwidth=1.1cm]
				\item[base$^\ast_{\text{shr}}$] Cross-sectional reconciliation of the best base forecasts (\autoref{sec:SimulationExperiment})
				\item[base$_{\text{shr}}$] Best cross-sectional reconciliation of base forecasts (\autoref{sec:SimulationExperiment})
				\item[tbats$_{\text{shr}}$] Cross-sectional reconciliation of tbats base forecasts (\autoref{sec:energy})
			\end{itemize}
		\end{minipage} &  \begin{minipage}[t]{0.52\linewidth}
			\textit{Single-task combination} (\textit{incoherent forecasts}, Sections \ref{sec:SimulationExperiment}/\ref{sec:energy})
			\begin{itemize}[leftmargin=!, nosep, align = left, labelwidth=1.1cm]
				\item[ew] Equal-weighted average
				\item[ow$_{\text{var}}$] Weighted average, optimal weights inversely proportional to the MSE's
				\item[ow$_{\text{cov}}$] Weighted average, optimal weights in the unit simplex \citep{Conflitti2015} computed using the MSE matrix 
			\end{itemize}
			\vspace{1ex}
			\hrule
			\vspace{1ex}
			\textit{Coherent combination} (Sections \ref{sec:SimulationExperiment}/\ref{sec:energy})
			\begin{itemize}[leftmargin=!, nosep, align = left, labelwidth=1.1cm]
				\item[src] Sequential reconciliation-then-combination with ew (balanced panel of forecasts)
				\item[scr$_{\text{ew}}$] Sequential ew combination-then-reconciliation
				\item[scr$_{\text{var}}$] Sequential ow$_{\text{var}}$ combination-then-reconciliation
				\item[scr$_{\text{cov}}$] Sequential ow$_{\text{cov}}$ combination-then-reconciliation
				\item[occ$_{\text{be}}$] Optimal coherent combination using a \textit{by-expert} block-diagonal shrunk error covariance matrix
			\end{itemize}
		\end{minipage} \\
		\bottomrule
\end{tabular}}
	%\vspace*{-1.25\baselineskip}
	\vspace*{-\baselineskip}
\end{table}

\begin{table}[!tbp]
	\centering
	\footnotesize
	\def\arraystretch{1}
	\setlength{\tabcolsep}{2pt}
	\caption{AvgRelMAE for the simulation experiment. Benchmark approach: equal-weighted average (ew). Bold entries identify the best-performing approaches, italic entries identify the second-best, and red color denotes forecasts worse than the benchmark. The column `Sett.' identifies the parameter setting defined in \autoref{tab:sett}.}\label{tab:sim}
\resizebox*{0.925\linewidth}{!}{
	
\begin{tabular}[t]{cc>{}c|ccccccccc>{}c|cccccc}
\toprule
\multicolumn{3}{c}{\textbf{ }} & \multicolumn{10}{c}{\textbf{Balanced panel of forecasts}} & \multicolumn{6}{c}{\textbf{Unbalanced panel of forecasts}} \\
\textbf{Sett.} & \textbf{p} & \textbf{N} & \rotatebox{90}{base$^{\ast}$} & \rotatebox{90}{base$^{\ast}_{\text{shr}}$} & \rotatebox{90}{base$_{\text{shr}}$} & \rotatebox{90}{ow$_{\text{var}}$} & \rotatebox{90}{ow$_{\text{cov}}$} & \rotatebox{90}{src} & \rotatebox{90}{scr$_{\text{ew}}$} & \rotatebox{90}{scr$_{\text{var}}$} & \rotatebox{90}{scr$_{\text{cov}}$} & \rotatebox{90}{occ$_{\text{be}}$} & \rotatebox{90}{ow$_{\text{var}}$} & \rotatebox{90}{ow$_{\text{cov}}$} & \rotatebox{90}{scr$_{\text{ew}}$} & \rotatebox{90}{scr$_{\text{var}}$} & \rotatebox{90}{scr$_{\text{cov}}$} & \rotatebox{90}{occ$_{\text{be}}$}\\
\midrule
1 & 4 & 50 & \textcolor{red}{1.308} & \textcolor{red}{1.063} & \textcolor{red}{1.058} & \textcolor{red}{1.002} & \textcolor{red}{1.027} & \em{0.921} & 0.949 & 0.951 & 0.969 & \textbf{0.905} & \textcolor{red}{1.031} & \textcolor{red}{1.069} & \em{0.905} & \textcolor{red}{1.054} & \textcolor{red}{1.092} & \textbf{0.892}\\
 &  & 100 & \textcolor{red}{1.307} & \textcolor{red}{1.038} & \textcolor{red}{1.037} & \textcolor{red}{1.000} & \textcolor{red}{1.014} & \em{0.915} & 0.946 & 0.946 & 0.955 & \textbf{0.894} & \textcolor{red}{1.024} & \textcolor{red}{1.070} & \em{0.901} & \textcolor{red}{1.033} & \textcolor{red}{1.087} & \textbf{0.885}\\
 &  & 200 & \textcolor{red}{1.305} & \textcolor{red}{1.026} & \textcolor{red}{1.025} & 1.000 & \textcolor{red}{1.007} & \em{0.912} & 0.943 & 0.942 & 0.946 & \textbf{0.890} & \textcolor{red}{1.010} & \textcolor{red}{1.062} & \em{0.897} & \textcolor{red}{1.003} & \textcolor{red}{1.072} & \textbf{0.880}\\
\addlinespace[0.25em]
 & 10 & 50 & \textcolor{red}{1.417} & \textcolor{red}{1.152} & \textcolor{red}{1.146} & \textcolor{red}{1.001} & \textcolor{red}{1.053} & \em{0.962} & 0.980 & 0.981 & \textcolor{red}{1.020} & \textbf{0.953} & \textcolor{red}{1.220} & \textcolor{red}{1.286} & \em{0.945} & \textcolor{red}{1.276} & \textcolor{red}{1.316} & \textbf{0.941}\\
 &  & 100 & \textcolor{red}{1.417} & \textcolor{red}{1.125} & \textcolor{red}{1.123} & \textcolor{red}{1.000} & \textcolor{red}{1.033} & \em{0.959} & 0.979 & 0.979 & \textcolor{red}{1.003} & \textbf{0.947} & \textcolor{red}{1.216} & \textcolor{red}{1.314} & \em{0.940} & \textcolor{red}{1.250} & \textcolor{red}{1.333} & \textbf{0.928}\\
 &  & 200 & \textcolor{red}{1.412} & \textcolor{red}{1.111} & \textcolor{red}{1.110} & 1.000 & \textcolor{red}{1.019} & \em{0.958} & 0.976 & 0.976 & 0.989 & \textbf{0.946} & \textcolor{red}{1.183} & \textcolor{red}{1.335} & \em{0.937} & \textcolor{red}{1.190} & \textcolor{red}{1.344} & \textbf{0.923}\\
\addlinespace[0.25em]
 & 20 & 50 & \textcolor{red}{1.461} & \textcolor{red}{1.188} & \textcolor{red}{1.181} & \textcolor{red}{1.001} & \textcolor{red}{1.063} & \em{0.979} & 0.990 & 0.991 & \textcolor{red}{1.040} & \textbf{0.974} & \textcolor{red}{1.351} & \textcolor{red}{1.424} & \textbf{0.972} & \textcolor{red}{1.426} & \textcolor{red}{1.456} & \em{0.981}\\
 &  & 100 & \textcolor{red}{1.463} & \textcolor{red}{1.162} & \textcolor{red}{1.159} & \textcolor{red}{1.000} & \textcolor{red}{1.044} & \em{0.978} & 0.990 & 0.990 & \textcolor{red}{1.023} & \textbf{0.972} & \textcolor{red}{1.335} & \textcolor{red}{1.450} & \em{0.969} & \textcolor{red}{1.385} & \textcolor{red}{1.471} & \textbf{0.968}\\
 &  & 200 & \textcolor{red}{1.455} & \textcolor{red}{1.147} & \textcolor{red}{1.145} & 1.000 & \textcolor{red}{1.028} & \em{0.978} & 0.989 & 0.988 & \textcolor{red}{1.008} & \textbf{0.971} & \textcolor{red}{1.287} & \textcolor{red}{1.465} & \em{0.966} & \textcolor{red}{1.302} & \textcolor{red}{1.475} & \textbf{0.958}\\
\midrule
2 & 4 & 50 & \textcolor{red}{1.226} & \textcolor{red}{1.054} & \textcolor{red}{1.052} & \textcolor{red}{1.001} & \textcolor{red}{1.026} & \em{0.949} & 0.967 & 0.968 & 0.986 & \textbf{0.942} & \textcolor{red}{1.022} & \textcolor{red}{1.045} & \em{0.930} & \textcolor{red}{1.032} & \textcolor{red}{1.055} & \textbf{0.919}\\
 &  & 100 & \textcolor{red}{1.226} & \textcolor{red}{1.037} & \textcolor{red}{1.036} & \textcolor{red}{1.000} & \textcolor{red}{1.014} & \em{0.945} & 0.965 & 0.966 & 0.975 & \textbf{0.936} & \textcolor{red}{1.020} & \textcolor{red}{1.047} & \em{0.927} & \textcolor{red}{1.023} & \textcolor{red}{1.055} & \textbf{0.914}\\
 &  & 200 & \textcolor{red}{1.225} & \textcolor{red}{1.026} & \textcolor{red}{1.026} & 1.000 & \textcolor{red}{1.007} & \em{0.943} & 0.963 & 0.963 & 0.967 & \textbf{0.932} & \textcolor{red}{1.013} & \textcolor{red}{1.046} & \em{0.923} & \textcolor{red}{1.008} & \textcolor{red}{1.052} & \textbf{0.910}\\
\addlinespace[0.25em]
 & 10 & 50 & \textcolor{red}{1.294} & \textcolor{red}{1.113} & \textcolor{red}{1.110} & \textcolor{red}{1.001} & \textcolor{red}{1.045} & \em{0.977} & 0.988 & 0.989 & \textcolor{red}{1.023} & \textbf{0.974} & \textcolor{red}{1.147} & \textcolor{red}{1.186} & \em{0.963} & \textcolor{red}{1.173} & \textcolor{red}{1.198} & \textbf{0.960}\\
 &  & 100 & \textcolor{red}{1.295} & \textcolor{red}{1.095} & \textcolor{red}{1.093} & \textcolor{red}{1.000} & \textcolor{red}{1.030} & \em{0.976} & 0.988 & 0.988 & \textcolor{red}{1.009} & \textbf{0.972} & \textcolor{red}{1.145} & \textcolor{red}{1.199} & \em{0.960} & \textcolor{red}{1.161} & \textcolor{red}{1.205} & \textbf{0.952}\\
 &  & 200 & \textcolor{red}{1.293} & \textcolor{red}{1.083} & \textcolor{red}{1.083} & 1.000 & \textcolor{red}{1.018} & \em{0.975} & 0.987 & 0.987 & 0.999 & \textbf{0.970} & \textcolor{red}{1.129} & \textcolor{red}{1.207} & \em{0.958} & \textcolor{red}{1.133} & \textcolor{red}{1.208} & \textbf{0.947}\\
\addlinespace[0.25em]
 & 20 & 50 & \textcolor{red}{1.319} & \textcolor{red}{1.135} & \textcolor{red}{1.132} & \textcolor{red}{1.000} & \textcolor{red}{1.050} & \em{0.988} & 0.994 & 0.994 & \textcolor{red}{1.032} & \textbf{0.986} & \textcolor{red}{1.226} & \textcolor{red}{1.269} & \textbf{0.982} & \textcolor{red}{1.259} & \textcolor{red}{1.280} & \em{0.989}\\
 &  & 100 & \textcolor{red}{1.319} & \textcolor{red}{1.117} & \textcolor{red}{1.114} & \textcolor{red}{1.000} & \textcolor{red}{1.036} & \em{0.988} & 0.994 & 0.994 & \textcolor{red}{1.021} & \textbf{0.986} & \textcolor{red}{1.221} & \textcolor{red}{1.283} & \em{0.981} & \textcolor{red}{1.245} & \textcolor{red}{1.289} & \textbf{0.981}\\
 &  & 200 & \textcolor{red}{1.318} & \textcolor{red}{1.105} & \textcolor{red}{1.103} & 1.000 & \textcolor{red}{1.023} & \em{0.987} & 0.994 & 0.994 & \textcolor{red}{1.009} & \textbf{0.984} & \textcolor{red}{1.199} & \textcolor{red}{1.286} & \em{0.980} & \textcolor{red}{1.206} & \textcolor{red}{1.287} & \textbf{0.974}\\
\midrule
3 & 4 & 50 & \textcolor{red}{1.110} & \textcolor{red}{1.034} & \textcolor{red}{1.034} & \textcolor{red}{1.000} & \textcolor{red}{1.020} & \em{0.979} & 0.988 & 0.988 & \textcolor{red}{1.003} & \textbf{0.978} & \textcolor{red}{1.321} & \textcolor{red}{1.372} & \em{0.969} & \textcolor{red}{1.379} & \textcolor{red}{1.437} & \textbf{0.961}\\
 &  & 100 & \textcolor{red}{1.112} & \textcolor{red}{1.025} & \textcolor{red}{1.025} & \textcolor{red}{1.000} & \textcolor{red}{1.013} & \em{0.977} & 0.986 & 0.987 & 0.995 & \textbf{0.975} & \textcolor{red}{1.325} & \textcolor{red}{1.400} & \em{0.966} & \textcolor{red}{1.362} & \textcolor{red}{1.447} & \textbf{0.958}\\
 &  & 200 & \textcolor{red}{1.111} & \textcolor{red}{1.016} & \textcolor{red}{1.016} & 1.000 & \textcolor{red}{1.007} & \em{0.974} & 0.984 & 0.984 & 0.989 & \textbf{0.971} & \textcolor{red}{1.313} & \textcolor{red}{1.418} & \em{0.963} & \textcolor{red}{1.329} & \textcolor{red}{1.453} & \textbf{0.957}\\
\addlinespace[0.25em]
 & 10 & 50 & \textcolor{red}{1.137} & \textcolor{red}{1.059} & \textcolor{red}{1.059} & \textcolor{red}{1.000} & \textcolor{red}{1.029} & \em{0.992} & 0.995 & 0.995 & \textcolor{red}{1.017} & \textbf{0.991} & \textcolor{red}{1.462} & \textcolor{red}{1.556} & \em{0.985} & \textcolor{red}{1.523} & \textcolor{red}{1.585} & \textbf{0.984}\\
 &  & 100 & \textcolor{red}{1.139} & \textcolor{red}{1.050} & \textcolor{red}{1.049} & \textcolor{red}{1.000} & \textcolor{red}{1.021} & \em{0.990} & 0.995 & 0.995 & \textcolor{red}{1.011} & \textbf{0.989} & \textcolor{red}{1.462} & \textcolor{red}{1.586} & \em{0.983} & \textcolor{red}{1.503} & \textcolor{red}{1.599} & \textbf{0.979}\\
 &  & 200 & \textcolor{red}{1.138} & \textcolor{red}{1.044} & \textcolor{red}{1.041} & 1.000 & \textcolor{red}{1.013} & \em{0.989} & 0.995 & 0.995 & \textcolor{red}{1.004} & \textbf{0.988} & \textcolor{red}{1.437} & \textcolor{red}{1.596} & \em{0.982} & \textcolor{red}{1.455} & \textcolor{red}{1.600} & \textbf{0.977}\\
\addlinespace[0.25em]
 & 20 & 50 & \textcolor{red}{1.147} & \textcolor{red}{1.069} & \textcolor{red}{1.067} & \textcolor{red}{1.000} & \textcolor{red}{1.031} & \em{0.996} & 0.997 & 0.998 & \textcolor{red}{1.021} & \textbf{0.995} & \textcolor{red}{1.527} & \textcolor{red}{1.630} & \textbf{0.993} & \textcolor{red}{1.594} & \textcolor{red}{1.648} & \em{0.997}\\
 &  & 100 & \textcolor{red}{1.148} & \textcolor{red}{1.058} & \textcolor{red}{1.057} & \textcolor{red}{1.000} & \textcolor{red}{1.023} & \em{0.995} & 0.997 & 0.997 & \textcolor{red}{1.015} & \textbf{0.994} & \textcolor{red}{1.514} & \textcolor{red}{1.644} & \em{0.992} & \textcolor{red}{1.563} & \textcolor{red}{1.653} & \textbf{0.992}\\
 &  & 200 & \textcolor{red}{1.148} & \textcolor{red}{1.053} & \textcolor{red}{1.051} & 1.000 & \textcolor{red}{1.016} & \em{0.995} & 0.998 & 0.998 & \textcolor{red}{1.010} & \textbf{0.994} & \textcolor{red}{1.482} & \textcolor{red}{1.652} & \em{0.992} & \textcolor{red}{1.504} & \textcolor{red}{1.655} & \textbf{0.990}\\
\midrule
4 & 4 & 50 & \textcolor{red}{1.243} & \textcolor{red}{1.074} & \textcolor{red}{1.074} & 0.991 & \textcolor{red}{1.001} & \em{0.950} & 0.968 & 0.957 & 0.950 & \textbf{0.922} & 0.999 & \textcolor{red}{1.017} & \em{0.925} & \textcolor{red}{1.008} & \textcolor{red}{1.027} & \textbf{0.909}\\
 &  & 100 & \textcolor{red}{1.241} & \textcolor{red}{1.056} & \textcolor{red}{1.056} & 0.989 & 0.988 & 0.946 & 0.966 & 0.953 & \em{0.936} & \textbf{0.915} & \textcolor{red}{1.001} & \textcolor{red}{1.023} & \em{0.923} & \textcolor{red}{1.004} & \textcolor{red}{1.031} & \textbf{0.904}\\
 &  & 200 & \textcolor{red}{1.242} & \textcolor{red}{1.042} & \textcolor{red}{1.042} & 0.989 & 0.982 & 0.942 & 0.963 & 0.950 & \em{0.929} & \textbf{0.909} & 0.993 & \textcolor{red}{1.017} & \em{0.919} & 0.989 & \textcolor{red}{1.022} & \textbf{0.900}\\
\addlinespace[0.25em]
 & 10 & 50 & \textcolor{red}{1.317} & \textcolor{red}{1.134} & \textcolor{red}{1.134} & 0.986 & 0.987 & 0.978 & 0.988 & 0.973 & \em{0.955} & \textbf{0.948} & \textcolor{red}{1.129} & \textcolor{red}{1.162} & \em{0.959} & \textcolor{red}{1.156} & \textcolor{red}{1.175} & \textbf{0.947}\\
 &  & 100 & \textcolor{red}{1.313} & \textcolor{red}{1.117} & \textcolor{red}{1.117} & 0.985 & 0.968 & 0.976 & 0.987 & 0.972 & \textbf{0.936} & \em{0.942} & \textcolor{red}{1.126} & \textcolor{red}{1.173} & \em{0.957} & \textcolor{red}{1.141} & \textcolor{red}{1.181} & \textbf{0.938}\\
 &  & 200 & \textcolor{red}{1.317} & \textcolor{red}{1.107} & \textcolor{red}{1.107} & 0.985 & 0.959 & 0.975 & 0.986 & 0.971 & \textbf{0.926} & \em{0.939} & \textcolor{red}{1.114} & \textcolor{red}{1.183} & \em{0.953} & \textcolor{red}{1.119} & \textcolor{red}{1.187} & \textbf{0.933}\\
\addlinespace[0.25em]
 & 20 & 50 & \textcolor{red}{1.336} & \textcolor{red}{1.150} & \textcolor{red}{1.150} & 0.983 & 0.973 & 0.989 & 0.994 & 0.977 & \textbf{0.948} & \em{0.958} & \textcolor{red}{1.214} & \textcolor{red}{1.250} & \em{0.979} & \textcolor{red}{1.248} & \textcolor{red}{1.264} & \textbf{0.975}\\
 &  & 100 & \textcolor{red}{1.339} & \textcolor{red}{1.138} & \textcolor{red}{1.138} & 0.983 & \em{0.952} & 0.988 & 0.994 & 0.977 & \textbf{0.928} & 0.954 & \textcolor{red}{1.206} & \textcolor{red}{1.260} & \em{0.978} & \textcolor{red}{1.230} & \textcolor{red}{1.268} & \textbf{0.964}\\
 &  & 200 & \textcolor{red}{1.334} & \textcolor{red}{1.123} & \textcolor{red}{1.123} & 0.983 & \em{0.937} & 0.987 & 0.994 & 0.976 & \textbf{0.913} & 0.950 & \textcolor{red}{1.186} & \textcolor{red}{1.265} & \em{0.976} & \textcolor{red}{1.193} & \textcolor{red}{1.270} & \textbf{0.957}\\
\midrule
5 & 4 & 50 & \textcolor{red}{1.241} & \textcolor{red}{1.049} & \textcolor{red}{1.049} & 0.989 & \textcolor{red}{1.010} & \em{0.939} & 0.959 & 0.955 & 0.969 & \textbf{0.926} & 0.989 & \textcolor{red}{1.010} & \em{0.919} & 0.999 & \textcolor{red}{1.019} & \textbf{0.901}\\
 &  & 100 & \textcolor{red}{1.245} & \textcolor{red}{1.038} & \textcolor{red}{1.034} & 0.989 & 0.998 & \em{0.934} & 0.957 & 0.952 & 0.958 & \textbf{0.919} & 0.986 & \textcolor{red}{1.011} & \em{0.916} & 0.989 & \textcolor{red}{1.019} & \textbf{0.895}\\
 &  & 200 & \textcolor{red}{1.243} & \textcolor{red}{1.023} & \textcolor{red}{1.023} & 0.988 & 0.991 & \em{0.930} & 0.954 & 0.949 & 0.951 & \textbf{0.915} & 0.982 & \textcolor{red}{1.011} & \em{0.912} & 0.977 & \textcolor{red}{1.017} & \textbf{0.892}\\
\addlinespace[0.25em]
 & 10 & 50 & \textcolor{red}{1.322} & \textcolor{red}{1.122} & \textcolor{red}{1.122} & 0.994 & \textcolor{red}{1.036} & \em{0.972} & 0.985 & 0.981 & \textcolor{red}{1.012} & \textbf{0.966} & \textcolor{red}{1.120} & \textcolor{red}{1.157} & \em{0.955} & \textcolor{red}{1.145} & \textcolor{red}{1.168} & \textbf{0.945}\\
 &  & 100 & \textcolor{red}{1.320} & \textcolor{red}{1.103} & \textcolor{red}{1.103} & 0.994 & \textcolor{red}{1.021} & \em{0.971} & 0.985 & 0.981 & 1.000 & \textbf{0.964} & \textcolor{red}{1.118} & \textcolor{red}{1.171} & \em{0.952} & \textcolor{red}{1.134} & \textcolor{red}{1.177} & \textbf{0.936}\\
 &  & 200 & \textcolor{red}{1.320} & \textcolor{red}{1.087} & \textcolor{red}{1.087} & 0.993 & \textcolor{red}{1.009} & \em{0.968} & 0.983 & 0.979 & 0.990 & \textbf{0.960} & \textcolor{red}{1.105} & \textcolor{red}{1.179} & \em{0.950} & \textcolor{red}{1.109} & \textcolor{red}{1.180} & \textbf{0.930}\\
\addlinespace[0.25em]
 & 20 & 50 & \textcolor{red}{1.353} & \textcolor{red}{1.151} & \textcolor{red}{1.149} & 0.997 & \textcolor{red}{1.049} & \em{0.986} & 0.993 & 0.991 & \textcolor{red}{1.030} & \textbf{0.983} & \textcolor{red}{1.212} & \textcolor{red}{1.254} & \textbf{0.978} & \textcolor{red}{1.246} & \textcolor{red}{1.266} & \em{0.980}\\
 &  & 100 & \textcolor{red}{1.352} & \textcolor{red}{1.132} & \textcolor{red}{1.128} & 0.996 & \textcolor{red}{1.031} & \em{0.984} & 0.992 & 0.990 & \textcolor{red}{1.015} & \textbf{0.980} & \textcolor{red}{1.204} & \textcolor{red}{1.263} & \em{0.977} & \textcolor{red}{1.228} & \textcolor{red}{1.269} & \textbf{0.969}\\
 &  & 200 & \textcolor{red}{1.356} & \textcolor{red}{1.116} & \textcolor{red}{1.113} & 0.996 & \textcolor{red}{1.020} & \em{0.984} & 0.992 & 0.990 & \textcolor{red}{1.007} & \textbf{0.979} & \textcolor{red}{1.186} & \textcolor{red}{1.272} & \em{0.974} & \textcolor{red}{1.193} & \textcolor{red}{1.273} & \textbf{0.963}\\
\midrule
6 & 4 & 50 & \textcolor{red}{1.370} & \textcolor{red}{1.240} & \textcolor{red}{1.240} & 0.966 & 0.962 & 0.961 & 0.975 & 0.928 & \em{0.919} & \textbf{0.898} & 0.871 & 0.887 & 0.874 & \em{0.870} & 0.882 & \textbf{0.840}\\
 &  & 100 & \textcolor{red}{1.367} & \textcolor{red}{1.224} & \textcolor{red}{1.224} & 0.964 & 0.950 & 0.956 & 0.972 & 0.923 & \em{0.906} & \textbf{0.891} & 0.869 & 0.886 & 0.874 & \em{0.865} & 0.881 & \textbf{0.830}\\
 &  & 200 & \textcolor{red}{1.366} & \textcolor{red}{1.212} & \textcolor{red}{1.212} & 0.964 & 0.944 & 0.952 & 0.970 & 0.921 & \em{0.899} & \textbf{0.886} & 0.863 & 0.880 & 0.869 & \em{0.853} & 0.873 & \textbf{0.826}\\
\addlinespace[0.25em]
 & 10 & 50 & \textcolor{red}{1.508} & \textcolor{red}{1.363} & \textcolor{red}{1.363} & 0.980 & \textcolor{red}{1.003} & 0.981 & 0.990 & \em{0.964} & 0.980 & \textbf{0.946} & \textcolor{red}{1.016} & \textcolor{red}{1.047} & \em{0.913} & \textcolor{red}{1.036} & \textcolor{red}{1.049} & \textbf{0.891}\\
 &  & 100 & \textcolor{red}{1.507} & \textcolor{red}{1.349} & \textcolor{red}{1.349} & 0.978 & 0.986 & 0.979 & 0.990 & \em{0.962} & 0.965 & \textbf{0.942} & \textcolor{red}{1.013} & \textcolor{red}{1.055} & \em{0.910} & \textcolor{red}{1.025} & \textcolor{red}{1.056} & \textbf{0.878}\\
 &  & 200 & \textcolor{red}{1.499} & \textcolor{red}{1.331} & \textcolor{red}{1.331} & 0.979 & 0.975 & 0.977 & 0.988 & 0.961 & \em{0.956} & \textbf{0.940} & \textcolor{red}{1.005} & \textcolor{red}{1.061} & \em{0.911} & \textcolor{red}{1.007} & \textcolor{red}{1.060} & \textbf{0.875}\\
\addlinespace[0.25em]
 & 20 & 50 & \textcolor{red}{1.554} & \textcolor{red}{1.410} & \textcolor{red}{1.408} & 0.990 & \textcolor{red}{1.031} & 0.989 & 0.994 & \em{0.982} & \textcolor{red}{1.014} & \textbf{0.973} & \textcolor{red}{1.143} & \textcolor{red}{1.178} & \em{0.954} & \textcolor{red}{1.175} & \textcolor{red}{1.186} & \textbf{0.952}\\
 &  & 100 & \textcolor{red}{1.552} & \textcolor{red}{1.389} & \textcolor{red}{1.389} & 0.988 & \textcolor{red}{1.012} & 0.989 & 0.994 & \em{0.980} & 0.998 & \textbf{0.969} & \textcolor{red}{1.135} & \textcolor{red}{1.183} & \em{0.951} & \textcolor{red}{1.155} & \textcolor{red}{1.187} & \textbf{0.936}\\
 &  & 200 & \textcolor{red}{1.553} & \textcolor{red}{1.379} & \textcolor{red}{1.379} & 0.987 & 0.998 & 0.988 & 0.995 & \em{0.979} & 0.986 & \textbf{0.966} & \textcolor{red}{1.118} & \textcolor{red}{1.190} & \em{0.947} & \textcolor{red}{1.125} & \textcolor{red}{1.192} & \textbf{0.927}\\
\bottomrule
\end{tabular}
}
	\label{tab:MAE_simresults_corr}
	\vspace*{-\baselineskip}
\end{table}

\autoref{tab:MAE_simresults_corr} presents several critical aspects that influence the performance of forecasting methods under different conditions. It clearly appears that as the number of residuals used to estimate the error covariance matrices grows, the quality of forecast improves in terms of accuracy. This finding highlights the critical role of robust statistical techniques to estimate the covariances, particularly in high-dimensional settings where estimation errors can propagate and affect the results negatively. In relation to this problem, we observe that AvgRelMAE decreases as the number of experts $p$ increases. Developing criteria for identifying the most relevant experts for combining and reconciling forecasts is an important aspect that lies beyond the scope of this paper.

Looking at the different forecasting procedures, the coherent approaches src (for balanced cases), scr$_{\text{ew}}$, scr$_{\text{var}}$, and occ$_{\text{be}}$ consistently outperform the benchmark equal-weighted (ew) average, as indicated by indices uniformly below 1. Among these, occ$_{\text{be}}$ is the most effective overall, ranking first in the majority of cases for both balanced panels (51 out of 54 cases) and unbalanced panels (50 out of 54 cases) for any setting. 
The scr$_{\text{ew}}$ method also demonstrates strong performance, particularly in unbalanced panel, where it frequently ranks among the top-performing methods (it ranks first four times, and second 47 times). Similarly, the src method shows excellent results in balanced cases, ranking second in 37 out of 54 cases. However, a major drawback of this approach is that it cannot be used in the unbalanced case, which limits its applicability in real-world scenarios where data availability is often irregular.
In conclusion, among the evaluated approaches, occ$_{\text{be}}$ stands out as the most effective and accurate, highlighting the advantages of simultaneously combining and reconciling forecasts for different variables and from multiple experts. 

\section{Forecasting Australian daily electricity generation} \label{sec:energy}

To illustrate the effectiveness of the coherent forecast combination methodology for real-life data, we performed a forecasting experiment on daily electricity generation from various energy sources in Australia \citep{Panagiotelis2023}. Daily time series data were obtained from \url{https://opennem.org.au}, which compiles publicly available data from the Australian Energy Market Operator (AEMO). Accurate day-ahead forecasts are crucial for operational planning and ensuring the efficiency and stability of the power network, particularly as the growth of intermittent renewable energy sources, such as wind and solar, introduces significant variability and complexity into the power system.

\begin{figure}[tb]
	\centering
      \includegraphics[width=0.85\linewidth]{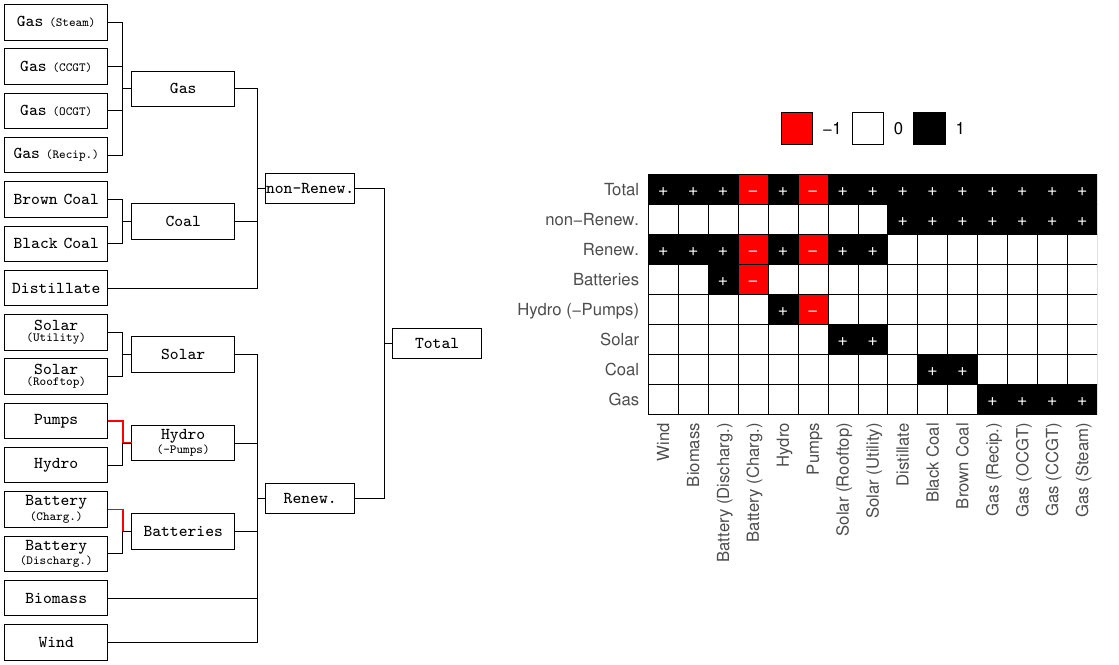}
      \caption{Australian daily electricity generation hierarchy (left) and the corresponding linear combination matrix (right), mapping 15 bottom variables into eight upper variables. Red color denotes bottom variables that, when aggregated, are subtracted instead of added.\label{fig:energy}}
      \vspace*{-\baselineskip}
\end{figure}

\autoref{fig:energy} shows the source generation hierarchy and the corresponding linear combination matrix $\Avet \in \{-1, 0, 1\}^{8 \times 15}$ for the whole dataset of 23 time series, with 15 of these being bottom-level series, representing the specific sources of generation. Detailed descriptions of the aggregation levels, including the components of each source, are available in \cite{Panagiotelis2023}\footnote{The original dataset, available at \href{https://github.com/PuwasalaG/Probabilistic-Forecast-Reconciliation}{\texttt{https://github.com/PuwasalaG/Probabilistic-Forecast-Reconciliation}}, underwent a cleaning process. Specifically, two out of the 15 original bottom time series (Distillate and Biomass) contained some negative values (from $-0.06$ to $-0.01$), which affected 26\% and 0.9\% of their data, respectively. To address this issue, these negative values were replaced by zero.}

%\begin{figure}[tb]
%	\centering
%	\includegraphics[width = 0.9\linewidth]{fig/ts_energy/plot_all_ts_energy.pdf}
%	\caption{Australian daily electricity generation time series.}
%	\label{fig:plot_ts_eg}
%\end{figure}

In order to assess the forecast accuracy of the various approaches, we performed a rolling forecast experiment with an expanding window. The first training window consists of 140 days (20 weeks) of data. One$\,$- to seven$\,$-$\,$step$\,$-$\,$ahead forecasts were generated leading to $Q_1 = 226$ one$\,$-, ..., and $Q_7 = 220$ seven$\,$-$\,$step$\,$-$\,$ahead daily forecasts for evaluation. Each series was independently modeled using three different approaches: stlf (seasonal and trend decomposition using loess; \citealp{Cleveland1990-ux}), arima (autoregressive integrated moving average; \citealp{Box1976-ud}), and tbats (exponential smoothing state space model with Box--Cox transformation, ARMA errors, and trend and seasonal components; \citealp{De-Livera2011-hx}). Starting from the base forecasts produced by these models, we considered three single-task combination procedures, resulting in incoherent forecasts, three single model reconciliation approaches, and five coherent combination procedures (see \autoref{tab:approaches})\footnote{The base forecasting models were implemented using the automatic forecasting procedure of the \textsf{R} package \texttt{forecast} \citep{Hyndman2008-qw, Hyndman2023-gd}. The combined and reconciled forecasts were computed using the \Rpackage. A complete set of results is available in the \textsf{GitHub} repository \githuburl.}:
\begin{enumerate}[nosep, leftmargin=*]
	\item \textit{Single-task combination}: Equal weights (ew) and two classical optimal weighting schemes proposed by \cite{Bates1969} and \cite{Newbold1974}---respectively called ow$_{\text{var}}$, based on the diagonal elements of the in-sample MSE matrix; and ow$_{\text{cov}}$, based on the whole MSE matrix, with optimized weights constrained to be nonnegative and sum to unity \citep{Conflitti2015}.
	\item \textit{Single-expert reconciliation}: For each base forecasting model, the reconciled forecasts are obtained through the MinT approach by \cite{Wickra2019}, with a shrunk in-sample error covariance matrix  (stlf$_{\text{shr}}$, arima$_{\text{shr}}$, tbats$_{\text{shr}}$).
	\item \textit{Coherent combination}: We consider four sequential approaches, namely (i) sequential reconciliation-then-equal-weighted-combination (src), (ii)--(iv) sequential combination-then-reconciliation: (scr$_{\text{ew}}$, scr$_{\text{var}}$, scr$_{\text{cov}}$), and occ$_{\text{be}}$, the optimal coherent combination approach with by-expert block-diagonal shrunk %in-sample 
	covariance matrix (see \autoref{sec:Wfeasible}).
\end{enumerate}
The forecast accuracy is evaluated using the average relative mean absolute (AvgRelMAE) and squared (AvgRelMSE) error \citep{Fleming1986, Davydenko2013} computed as follows:
$$
\text{AvgRelMAE}^{app} = \left(\prod_{h = 1}^H \text{AvgRelMAE}_h^{app}\right)^{\frac{1}{H}} \; \text{and} \; \text{AvgRelMSE}^{app} = \left(\prod_{h = 1}^H \text{AvgRelMSE}_h^{app}\right)^{\frac{1}{H}}.
$$
The AvgRelMAE and AvgRelMSE for a fixed forecast horizon $h$ are defined as
$$
\text{AvgRelMAE}_h^{app} = \left(\prod_{i = 1}^n \frac{\text{MAE}_{h,i}^{app}}{\text{MAE}_{h,i}^{\text{ew}}}\right)^{\frac{1}{n}} \; \text{and} \;\; \text{AvgRelMSE}_h^{app} = \left(\prod_{i = 1}^n \frac{\text{MSE}_{h,i}^{app}}{\text{MSE}_{h,i}^{\text{ew}}}\right)^{\frac{1}{n}},
$$
with the mean squared error (MSE) and mean absolute error (MAE) for each forecast horizon and variable computed as
$$
\text{MSE}_{h,i}^{app} = \displaystyle\frac{1}{Q_h}\sum_{q = 1}^{Q_h} \left(y_{i,h,q}-\check{y}_{i,h,q}^{\,app}\right)^2 \; \text{and} \;\; \text{MAE}_{h,i}^{app} = \displaystyle\frac{1}{Q_h}\sum_{q = 1}^{Q_h} \left|y_{i,h,q}-\check{y}_{i,h,q}^{\,app}\right|,
$$
where $h = 1,\ldots,H$ is the forecast horizon, $i=1,\ldots,n$ denotes the variable, $Q_h$ is the dimension of the test set, `app' is the approach used, $y_{i,h,q}$ is the observed value, and $\check{y}_{i,h,q}^{\,app}$ is the forecast value using the `app' approach (coherent or incoherent). 

Finally, we evaluated the forecasting performances of different approaches by first applying the pairwise DM test \citep{Diebold1995-xs} to investigate the null hypothesis of equal predictive accuracy (EPA) across models and, then, we used the model confidence set (MCS) developed by \cite{Hansen2011}, which identifies approaches with statistically superior performance. In addition, in \autoref{app:fexp}, we present the results for the non-parametric Friedman test along with the post hoc multiple-comparison-with-the-best (MCB) Nemenyi test \citep{Koning2005-up, Kourentzes2019-dj, Makridakis2022-we}, both of which are well-established tools for evaluating multiple forecast approaches.

\subsection{Results}

Forecast accuracy indices $AvgRelMAE$ and $AvgRelMSE$ for all 23 time series are reported in \autoref{tab:Energy_accuracy}, while \autoref{fig:energy_dm} and \autoref{tab:energy_mcs} show the pairwise DM tests and the MCS approach, respectively. A detailed analysis, with separate results for upper and bottom time series and a wider set of coherent combination procedures, is available in \autoref{app:fexp}. The forecasting approaches are evaluated for each forecast horizon (one through seven days), and across all horizons (denoted by `1:7').

\begin{table}[!tb]
	\centering
	\footnotesize
	\def\arraystretch{0.85}
	\caption{AvgRelMAE (top panel) and AvgRelMSE (bottom panel) of daily forecasts for the Australian electricity generation dataset. Benchmark approach: equal-weighted average (ew). Bold entries identify the best-performing approaches, italicized entries identify the second-best, and red color denotes forecasts worse than the benchmark.}\label{tab:Energy_accuracy}
	
\begin{tabular}[t]{>{}l|cccccccc}
\toprule
\multicolumn{1}{c}{\textbf{}} & \multicolumn{8}{c}{\textbf{Forecast horizon}} \\
\cmidrule(l{0pt}r{0pt}){2-9}
\multicolumn{1}{l|}{\textbf{Approach}} & 1 & 2 & 3 & 4 & 5 & 6 & 7 & 1:7\\
\midrule
\addlinespace[0.3em]
\multicolumn{9}{c}{\textbf{AvgRelMAE} - All 23 time series}\\
\addlinespace[0em]
\multicolumn{9}{l}{\textit{Base (incoherent forecasts) and single model reconciliation}}\\
tbats & \textcolor{red}{1.0447} & \textcolor{red}{1.0515} & \textcolor{red}{1.0348} & \textcolor{red}{1.0266} & \textcolor{red}{1.0305} & \textcolor{red}{1.0288} & \textcolor{red}{1.0201} & \textcolor{red}{1.0331}\\
tbats$_{\text{shr}}$ & \textcolor{red}{1.0320} & \textcolor{red}{1.0413} & \textcolor{red}{1.0231} & \textcolor{red}{1.0134} & \textcolor{red}{1.0212} & \textcolor{red}{1.0208} & \textcolor{red}{1.0188} & \textcolor{red}{1.0235}\\
\addlinespace[0em]
\multicolumn{9}{l}{\textit{Combination (incoherent forecasts)}}\\
ow$_{\text{var}}$ & 0.9927 & 0.9921 & 0.9982 & 0.9983 & 0.9967 & 0.9967 & 0.9990 & 0.9965\\
ow$_{\text{cov}}$ & \textcolor{red}{1.0216} & \textcolor{red}{1.0208} & \textcolor{red}{1.0390} & \textcolor{red}{1.0423} & \textcolor{red}{1.0307} & \textcolor{red}{1.0250} & \textcolor{red}{1.0325} & \textcolor{red}{1.0309}\\
\addlinespace[0em]
\multicolumn{9}{l}{\textit{Coherent combination}}\\
src & 0.9939 & 0.9941 & 0.9919 & 0.9895 & 0.9887 & \em{0.9908} & \em{0.9933} & 0.9915\\
scr$_{\text{ew}}$ & 0.9952 & 0.9959 & 0.9911 & 0.9908 & 0.9908 & 0.9932 & 0.9961 & 0.9930\\
scr$_{\text{var}}$ & \em{0.9819} & \em{0.9803} & \em{0.9869} & \em{0.9895} & \em{0.9887} & 0.9913 & 0.9972 & \em{0.9882}\\
scr$_{\text{cov}}$ & \textcolor{red}{1.0081} & \textcolor{red}{1.0081} & \textcolor{red}{1.0270} & \textcolor{red}{1.0327} & \textcolor{red}{1.0245} & \textcolor{red}{1.0197} & \textcolor{red}{1.0250} & \textcolor{red}{1.0215}\\
occ$_{\text{be}}$ & \textbf{0.9779} & \textbf{0.9745} & \textbf{0.9843} & \textbf{0.9852} & \textbf{0.9851} & \textbf{0.9880} & \textbf{0.9926} & \textbf{0.9843}\\
\midrule
\addlinespace[0.3em]
\multicolumn{9}{c}{\textbf{AvgRelMSE} - All 23 time series}\\
\addlinespace[0em]
\multicolumn{9}{l}{\textit{Base (incoherent forecasts) and single model reconciliation}}\\
tbats & \textcolor{red}{1.0796} & \textcolor{red}{1.0780} & \textcolor{red}{1.0445} & \textcolor{red}{1.0270} & \textcolor{red}{1.0322} & \textcolor{red}{1.0288} & \textcolor{red}{1.0142} & \textcolor{red}{1.0393}\\
tbats$_{\text{shr}}$ & \textcolor{red}{1.0478} & \textcolor{red}{1.0577} & \textcolor{red}{1.0304} & \textcolor{red}{1.0108} & \textcolor{red}{1.0219} & \textcolor{red}{1.0213} & \textcolor{red}{1.0116} & \textcolor{red}{1.0257}\\
\addlinespace[0em]
\multicolumn{9}{l}{\textit{Combination (incoherent forecasts)}}\\
ow$_{\text{var}}$ & 0.9840 & 0.9881 & 0.9995 & \textcolor{red}{1.0032} & \textcolor{red}{1.0020} & \textcolor{red}{1.0028} & \textcolor{red}{1.0054} & 0.9995\\
ow$_{\text{cov}}$ & \textcolor{red}{1.0279} & \textcolor{red}{1.0494} & \textcolor{red}{1.0972} & \textcolor{red}{1.1103} & \textcolor{red}{1.1009} & \textcolor{red}{1.0993} & \textcolor{red}{1.1055} & \textcolor{red}{1.0908}\\
\addlinespace[0em]
\multicolumn{9}{l}{\textit{Coherent combination}}\\
src & 0.9827 & 0.9855 & 0.9863 & \em{0.9833} & \textbf{0.9852} & \textbf{0.9873} & \textbf{0.9911} & \em{0.9859}\\
scr$_{\text{ew}}$ & 0.9875 & 0.9898 & 0.9859 & 0.9859 & \em{0.9885} & \em{0.9905} & \em{0.9962} & 0.9890\\
scr$_{\text{var}}$ & \em{0.9586} & \em{0.9683} & \em{0.9838} & 0.9942 & 0.9982 & \textcolor{red}{1.0017} & \textcolor{red}{1.0114} & 0.9910\\
scr$_{\text{cov}}$ & \textcolor{red}{1.0026} & \textcolor{red}{1.0287} & \textcolor{red}{1.0795} & \textcolor{red}{1.0972} & \textcolor{red}{1.0942} & \textcolor{red}{1.0913} & \textcolor{red}{1.0981} & \textcolor{red}{1.0773}\\
occ$_{\text{be}}$ & \textbf{0.9481} & \textbf{0.9560} & \textbf{0.9754} & \textbf{0.9831} & 0.9891 & 0.9939 & 0.9993 & \textbf{0.9808}\\
\bottomrule
\end{tabular}

	%\vspace*{-\baselineskip}
\end{table}

\begin{figure}[!tb]
	\centering
	\includegraphics[width = 0.81\linewidth]{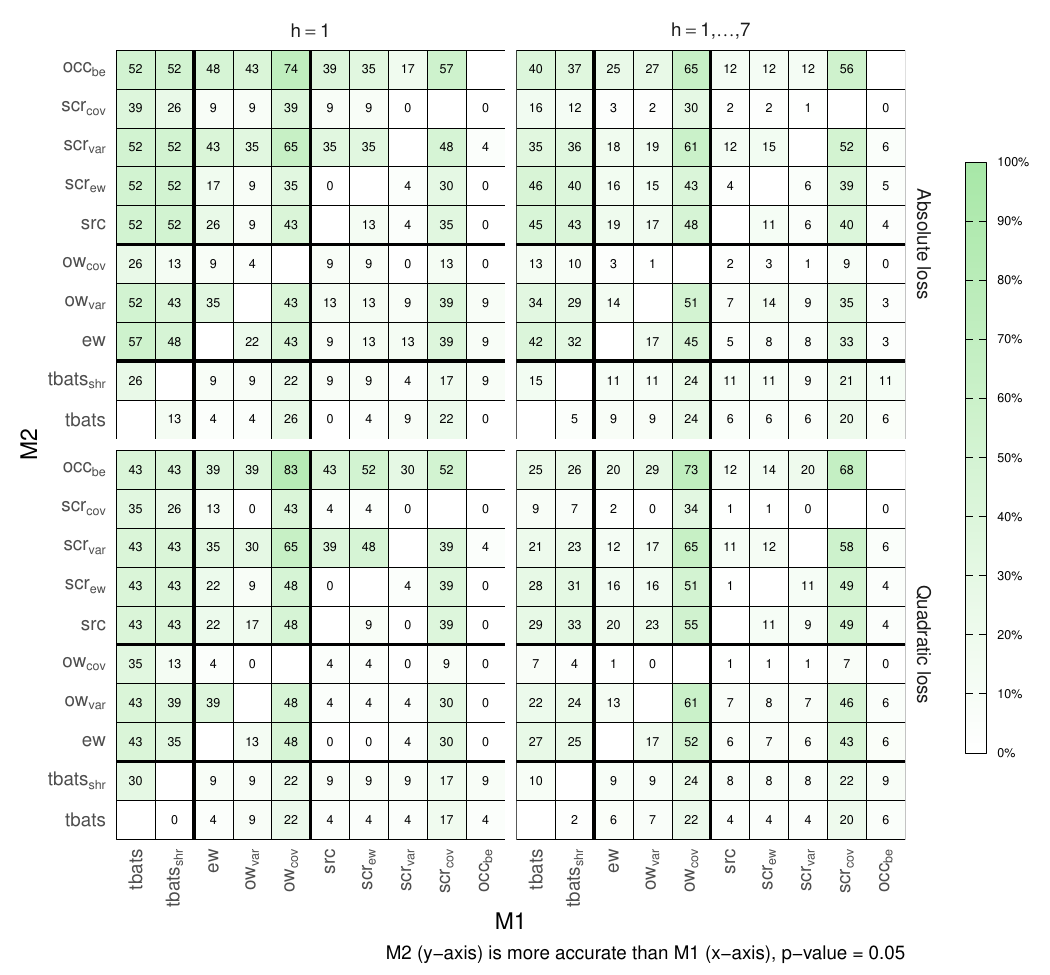}
	\caption{Pairwise DM-test results for the Australian electricity generation dataset, evaluated using absolute loss (top panels) and quadratic loss (bottom panels) across different forecast horizons. The left panels correspond to forecast horizon $h = 1$, while the right panels are for $h = 1, \dots, 7$. Each cell reports the percentage of series for which the $p$-value of the DM-test is below $0.05$. For example, the value 52 in the top-left cell means that occ$_{\text{be}}$ was more accurate ($p$-value $<0.05$) than tbats for $h = 1$ in terms of the absolute error in 52\% of the series.}\label{fig:energy_dm}
	%\vspace*{-\baselineskip}
\end{figure}

\begin{table}[!tb]
	\centering
	\footnotesize
	\def\arraystretch{0.9}
	\caption{Model confidence set results ($10^4$ bootstrap samples) for the Australian electricity generation dataset, evaluated using absolute loss (top panels) and quadratic loss (bottom panel) across different forecast horizons ($h = 1$ and $h = 1, \dots, 7$). Each cell reports the percentage of series for which that approach is in the model confidence set across different thresholds ($\delta \in \{95\%, 90\%, 80\%\}$).}\label{tab:energy_mcs}
	
\begin{tabular}[t]{>{}l|cc>{}c|ccc}
\toprule
\multicolumn{1}{c}{\textbf{ }} & \multicolumn{3}{c}{\textbf{$h =1$}} & \multicolumn{3}{c}{\textbf{$h =1:7$}} \\
\cmidrule(l{0pt}r{0pt}){2-4} \cmidrule(l{0pt}r{0pt}){5-7}
\multicolumn{1}{l|}{\textbf{Approach}} & $\delta = 95\%$ & $\delta = 90\%$ & $\delta = 80\%$ & $\delta = 95\%$ & $\delta = 90\%$ & $\delta = 80\%$\\
\midrule
\addlinespace[0.3em]
\multicolumn{7}{c}{Absolute loss - All 23 time series}\\
\addlinespace[0em]
\multicolumn{7}{l}{\textit{Base (incoherent forecasts) and single model reconciliation}}\\
tbats & 56.5 & 56.5 & 52.2 & 78.3 & 69.6 & 56.5\\
tbats$_{\text{shr}}$ & 78.3 & 73.9 & 60.9 & 87.0 & 82.6 & 69.6\\
\addlinespace[0em]
\multicolumn{7}{l}{\textit{Combination (incoherent forecasts)}}\\
ew & 87.0 & 87.0 & 78.3 & \textbf{95.7} & \em{91.3} & 78.3\\
ow$_{\text{var}}$ & \em{95.7} & \em{95.7} & 82.6 & \textbf{95.7} & \em{91.3} & \em{82.6}\\
ow$_{\text{cov}}$ & 73.9 & 69.6 & 60.9 & 73.9 & 65.2 & 43.5\\
\addlinespace[0em]
\multicolumn{7}{l}{\textit{Coherent combination}}\\
src & 91.3 & 91.3 & 87.0 & \textbf{95.7} & \textbf{95.7} & \textbf{87.0}\\
scr$_{\text{ew}}$ & 91.3 & 91.3 & 87.0 & \textbf{95.7} & \em{91.3} & 78.3\\
scr$_{\text{var}}$ & \textbf{100.0} & \textbf{100.0} & \em{91.3} & \em{91.3} & \em{91.3} & \textbf{87.0}\\
scr$_{\text{cov}}$ & 82.6 & 78.3 & 73.9 & 78.3 & 69.6 & 65.2\\
occ$_{\text{be}}$ & \textbf{100.0} & \textbf{100.0} & \textbf{95.7} & \textbf{95.7} & \textbf{95.7} & \textbf{87.0}\\
\midrule
\addlinespace[0.3em]
\multicolumn{7}{c}{Quadratic loss - All 23 time series}\\
\addlinespace[0em]
\multicolumn{7}{l}{\textit{Base (incoherent forecasts) and single model reconciliation}}\\
tbats & 65.2 & 65.2 & 60.9 & \em{91.3} & 73.9 & 73.9\\
tbats$_{\text{shr}}$ & 73.9 & 69.6 & 65.2 & \textbf{95.7} & 82.6 & 69.6\\
\addlinespace[0em]
\multicolumn{7}{l}{\textit{Combination (incoherent forecasts)}}\\
ew & 87.0 & 78.3 & 60.9 & \textbf{95.7} & \em{87.0} & \em{82.6}\\
ow$_{\text{var}}$ & \textbf{100.0} & 82.6 & 78.3 & \em{91.3} & \em{87.0} & \em{82.6}\\
ow$_{\text{cov}}$ & 78.3 & 69.6 & 52.2 & 78.3 & 60.9 & 34.8\\
\addlinespace[0em]
\multicolumn{7}{l}{\textit{Coherent combination}}\\
src & \em{95.7} & \em{91.3} & \em{82.6} & \em{91.3} & \textbf{91.3} & \textbf{91.3}\\
scr$_{\text{ew}}$ & \em{95.7} & 87.0 & 73.9 & \em{91.3} & \textbf{91.3} & \textbf{91.3}\\
scr$_{\text{var}}$ & \em{95.7} & \textbf{95.7} & \textbf{95.7} & \em{91.3} & \textbf{91.3} & \textbf{91.3}\\
scr$_{\text{cov}}$ & 82.6 & 82.6 & 73.9 & 73.9 & 60.9 & 56.5\\
occ$_{\text{be}}$ & \textbf{100.0} & \textbf{95.7} & \textbf{95.7} & \em{91.3} & \textbf{91.3} & \textbf{91.3}\\
\bottomrule
\end{tabular}

	\vspace*{-\baselineskip}
\end{table}

We begin by examining the incoherent forecasting approaches reported in \autoref{tab:Energy_accuracy}. The best base forecast on the present dataset is tbats, but its $AvgRelMAE$ and $AvgRelMSE$ for the full set of $23$ time series are above the equal-weighted (ew) benchmark at every forecast horizon and under both losses. The two classical single-task combinations behave differently from one another: the variance-weighted combination $\text{ow}_{\text{var}}$ is essentially equivalent to ew, while the covariance-weighted combination $\text{ow}_{\text{cov}}$ is consistently worse than the benchmark. The aggregate Diebold--Mariano and MCS tests in \autoref{fig:energy_dm} and \autoref{tab:energy_mcs} indicate that the differences among the incoherent approaches are not statistically significant, in line with the well-known forecast combination puzzle \citep{Smith2009, Claeskens2016, Qian2019, Frazier2023, Liu2024}.

The single-expert MinT reconciliation reported in the main table, $\text{tbats}_{\text{shr}}$, improves on the corresponding base model under both losses, but the improvement is not sufficient to bring the indices below one: $\text{tbats}_{\text{shr}}$ remains above the equal-weighted benchmark at every forecast horizon. The same pattern holds for the other two base models reported in \autoref{app:fexp}. On this dataset, both classical single-task combination across experts and MinT reconciliation across variables do not outperform the equal-weighted average.

The coherent combination procedures behave differently. With the exception of $\text{scr}_{\text{cov}}$, which is consistently worse than ew under both losses, the sequential schemes src, $\text{scr}_{\text{ew}}$, and $\text{scr}_{\text{var}}$, together with the proposed optimal coherent combination $\text{occ}_{\text{be}}$, improve on the base forecasts, on the single-task combinations and on the single-expert MinT reconciliation. Among the sequential schemes, the variance-weighted procedure $\text{scr}_{\text{var}}$ is the most accurate. The proposed approach $\text{occ}_{\text{be}}$ obtains the best $AvgRelMAE$ and $AvgRelMSE$ at every forecast horizon and is also the best performer at the upper and bottom levels of the hierarchy reported in \autoref{app:fexp}. The sequential alternative src is the closest competitor at long horizons, but is outperformed at short horizons by $\text{occ}_{\text{be}}$ and is not applicable when the panel of base forecasts is unbalanced, that is, when different experts forecast different subsets of the variables.

Overall, the results indicate that neither classical single-task combination across experts nor single-expert MinT reconciliation across variables is sufficient to improve on the equal-weighted benchmark on this dataset. The combination of the two within a single procedure changes the picture, and the proposed optimal coherent combination $\text{occ}_{\text{be}}$ emerges as the best performer across forecast horizons and across hierarchical levels, while remaining applicable in the unbalanced case, in which the most competitive sequential alternative is not. These findings support the empirical motivation of the framework: in settings with several experts forecasting linearly constrained variables, combination and reconciliation are most effective when addressed jointly rather than separately or sequentially.

\section{Conclusions, limitations, and future work}\label{sec:conclusion}

In this paper, we introduced several coherent multi-task forecast combination procedures designed for linearly constrained time series, with a focus on an approach involving an optimization-based technique to combine multiple unbiased, possibly incoherent base forecasts, while ensuring coherence across the variables.
Our solution, which combines \cite{Bates1969} and \cite{Stone1942-fa}, has nice theoretical properties, including unbiasedness and minimum forecast error variance, making it a powerful tool that balances forecast accuracy and coherence in different scenarios.

However, the theoretical solution relies on two key assumptions: the base forecasts are unbiased, and the error covariance matrix is known. Regarding the former issue, the task of removing any bias from the base forecasts is left to the user. Addressing the second challenge, we outlined practical approaches for the estimation of the covariance matrix, such as using in-sample residuals, which consider techniques from the literature on both forecast combination and reconciliation. These approaches provide a feasible path forward for practitioners dealing with multiple base forecasts in real-world settings, where there is not perfect knowledge of the covariance matrix.

From a practical perspective, the new methodology was empirically validated through a simulation experiment and in a real-world application, demonstrating its effectiveness and robustness. The results showed that the coherent multi-task forecast combination consistently outperformed individual base forecasts, single-task combinations, forecasts reconciled by single experts, and sequential coherent combination approaches. This performance highlighted the method's ability to integrate multiple sources of information while maintaining coherence.

Future research presents several opportunities for the coherent multi-task forecast combination methodology. A particularly important area involves the development of more advanced and flexible approaches for estimating the error covariance matrix. While the practical solutions we proposed relied on in-sample residuals, the formulation of more sophisticated techniques holds significant potential for improving covariance estimation, especially in high-dimensional settings where data limitations are prevalent. Another critical direction is extending the framework from point to probabilistic forecasting, a shift that is essential for informed decision making in domains such as energy, finance, and economics. 
It should be emphasized that \autoref{thm:mmse} shares several similarities with modern stacking and cross-learning forecasting approaches \citep{WOLPERT1992241, Breiman1996, Semenoglou2021, Oliveira2023, Lovaglio2024, Chen2025}.
These parallels offer a foundation for
further exploration.
In our ongoing research, we aim to investigate these directions, in order to improve both the accuracy and flexibility of the proposed methodology.

\section*{Data Availability Statement}

The dataset, \textsf{R} source codes, and the associated package \texttt{FoCo2} used in this study are publicly available on \textsf{GitHub}. The codes and project files can be accessed at \href{https://github.com/danigiro/cfc-project}{\url{github.com/danigiro/cfc-project}}, and the \texttt{FoCo2} package is available at \href{https://github.com/danigiro/FoCo2}{\url{github.com/danigiro/FoCo2}}.

\section*{Acknowledgments}

The authors are grateful to the Editor, the Associate Editor, and the anonymous reviewers for their valuable comments and suggestions, which have substantially improved the paper. The authors acknowledge financial support from project PRIN2022 `PRICE: A New Paradigm for High-Frequency Finance' 2022C799SX.

\section*{Declaration of Generative AI and AI-assisted technologies in the writing process}

During the preparation of this work, the authors used ChatGPT and Claude to improve the English language and wording of some sentences. After using these tools, the authors reviewed and edited the content as needed and take full responsibility for the content of the publication.

\clearpage
\spacingset{1.2}
\appendix
\renewcommand{\thesection}{\Alph{section}}
\renewcommand{\thesubsection}{\Alph{section}.\arabic{subsection}}
\setcounter{section}{0}
\counterwithin{figure}{section}
\counterwithin{table}{section}
\counterwithin{equation}{section}

\section{An example of a general linearly constrained time series}
\label{app:lcts}

\noindent \autoref{sec:general_lcts} shows that, for a genuine hierarchical time series, the zero-constraints matrix $\Cvet$ already has the form $\begin{bmatrix}[0.7] \Ivet_{n_u} & -\Avet \end{bmatrix}$, so that the linear combination matrix $\Avet$ and the structural-like matrix $\Svet$ can be read off directly from the constraints. This appendix works out an example in which this is not the case, namely the general linearly constrained time series shown in the right panel of \autoref{fig:example}, consisting of two simple hierarchies that share the same top-level series, but with different bottom series.

Denoting $\yvet_t = \begin{bmatrix}[0.7] \texttt{X}_t & \texttt{A}_t & \texttt{AA}_t & \texttt{AB}_t & \texttt{B}_t & \texttt{C}_t & \texttt{D}_t \end{bmatrix}^\top$, we may write
\begin{equation}\label{eq:lsr_app}
	\addtolength{\jot}{-0.5em}
\left\{\rule{0cm}{0.8cm}\right. \begin{aligned}
	\texttt{X} &= \texttt{C} + \texttt{D} \\
	\texttt{X} &= \texttt{AA} + \texttt{AB} + \texttt{B} \\
	\texttt{A} &= \texttt{AA} + \texttt{AB} 
\end{aligned}\quad\mbox{and}\quad \begin{bmatrix}
	1 & 0 & 0 & 0 & 0 & -1 & -1 \\[-1ex]
	1 & 0 & -1 & -1 & -1 & 0 & 0 \\[-1ex]
	0 & 1 & -1 & -1 & 0 & 0 & 0 \\
\end{bmatrix} \yvet_t = \Zerovet_{(3 \times 1)},
\end{equation}
where the zero-constraints matrix has not the requested $\begin{bmatrix}[0.7] \Ivet_{n_u} & -\Avet \end{bmatrix}$ form. However, $\Avet$ (and thus $\Cvet$ and $\Svet$) may be obtained after simple operations on (\ref{eq:lsr_app}), shown in \cite{Girolimetto2024}:
\begin{equation}\label{eq:lsr_new_app}
	\addtolength{\jot}{-0.5em}
	\left\{\rule{0cm}{0.8cm}\right.\begin{aligned}
		\texttt{X} &= \texttt{C} + \texttt{D} \\
		\texttt{A} &=  -\texttt{B} + \texttt{AA} + \texttt{AB}  \\
		\texttt{AA} &= -\texttt{AB} - \texttt{B} + \texttt{C} + \texttt{D}
	\end{aligned} \quad\mbox{and}\quad 
	\Avet = \begin{bmatrix}
	0 & 0 & 1 & 1 \\[-1ex]
	0 & -1 & 1 & 1 \\[-1ex]
	-1 & -1 & 1 & 1 \\
	\end{bmatrix},
\end{equation}
with $n_u=3$, $n_b=4$, $\uvet_t = \begin{bmatrix}[0.7] \texttt{X}_t & \texttt{A}_t & \texttt{AA}_t \end{bmatrix}^\top$ and $\bvet_t = \begin{bmatrix}[0.7] \texttt{AB}_t & \texttt{B}_t & \texttt{C}_t & \texttt{D}_t \end{bmatrix}^\top$. Note that the upper and bottom variables are not those suggested by the graphical representation: \texttt{AA} is a free variable in the original formulation but becomes a constrained one after the row operations, which shows that the partition between $\uvet_t$ and $\bvet_t$ is a property of the constraints and not of the way the series are drawn.

\clearpage

\section{Balanced case: $n_j=n$ $\forall j$ and $p_i = p$ $\forall i$}
%\label{sec: Balanced}
\label{app:rect}
\noindent In the balanced case, expression (\ref{eq:Stonemodel}) may be grouped as $p$ linear models relating the (unbiased) base forecast $\widehat{\yvet}^j$ to the target vector $\yvet$:
\begin{equation}
	\label{eq:hatyreg}
	\widehat{\yvet}^j = \yvet + \epsvet^j, \quad j=1,\ldots,p,
\end{equation}
where the base forecast errors $\epsvet^j$, $j=1,\ldots,p$, are $(n \times 1)$ zero-mean random vectors, with $(n \times n)$ covariance matrices $\Wvet_{jl} = E\left[\epsvet^j(\epsvet^l)^\top\right]$, $j,l = 1, \ldots p$. The $p$ relationships in expression (\ref{eq:hatyreg}) can be written in compact form as a simplified multivariate multiple linear model \citep{Seber1984}:
\begin{equation}
	\label{mod:multivariate}
	\begin{aligned}
		\widehat{\Yvet} & = \Big[
			\;\overbrace{\yvet^{\phantom{A}} \, \ldots^{\phantom{A}} \, \yvet}^{p \text{ times}} \;
		\Big] + \Evet \\
		& = \yvet\Unovet_p^\top + \Evet,
	\end{aligned}
\end{equation}
where $\widehat{\Yvet}$ is the ($n \times p$) matrix defined in (\ref{Yhat}), containing the base forecasts of the target vector $\yvet$ produced by $p$ different experts, and the ($n \times p$) matrix $\Evet = \left[\epsvet^1 \; \ldots \; \epsvet^j \; \ldots \; \epsvet^p\right]$ contains the base forecast errors $\epsvet^j = \widehat{\yvet}^j - \yvet$, $j=1, \ldots, p$, across all models and variables. In both matrices $\widehat{\Yvet}$ and $\Evet$, each column denotes a single expert $j = 1,\dots,p$, and each row denotes an individual variable $i = 1, \dots, n$. Applying the $\text{vec}\left(\cdot\right)$ operator to expression (\ref{mod:multivariate}), noting that $\text{vec}\left(\yvet\Unovet_p^\top\right) = \left(\Unovet_p \otimes \Ivet_n\right)\yvet$, yields the multiple regression model
\begin{equation}\label{mod:stone_unconstrained_be}
	\widehat{\yvet} = \Kvet\yvet + \epsvet ,
\end{equation}
where $\widehat{\yvet} = \text{vec}\big(\,\widehat{\Yvet}\,\big)$ is a $(np \times 1)$ vector containing the base forecasts stacked by-expert, $\Kvet = \Unovet_p \otimes \Ivet_n$ is a $(np \times n)$ matrix, and $\epsvet = \text{vec}\left(\Evet\right)$ is a ($np \times 1$) vector of errors with zero average and $(np\times np)$, known and positive definite  covariance matrix $\Wvet$:
\begin{equation}
\Wvet = \begin{bmatrix}[0.9]
	\Wvet_{1} & \cdots & \Wvet_{1j} & \cdots & \Wvet_{1p} \\
	\vdots & \ddots & \vdots & \ddots & \vdots \\
	\Wvet_{j1} & \cdots & \Wvet_{j} & \cdots & \Wvet_{jp} \\
	\vdots & \ddots & \vdots & \ddots & \vdots \\
	\Wvet_{p1} & \cdots & \Wvet_{pj} & \cdots & \Wvet_{p}
\end{bmatrix}
\end{equation}
where $\Wvet_j \equiv \Wvet_{jj}$, $j=1,\ldots,p$.				% Appendix B

\clearpage
\section{A simple numerical example}\label{app:simpex}

%\subsubsection*{A simple numerical example}
%\clearpage
%\vspace{.5cm}
%\noindent\textit{A simple numerical example}

\noindent Consider the case $n=3$, i.e., $\yvet = \left[y_1 \; y_2 \; y_3\right]^\top$, $p_1 = 2, p_2 = 1$ and $p_3 = 4$, with
$$
\widehat{\Yvet} = \begin{bmatrix}[0.9]
	\widehat{y}_1^1 & \bullet & \widehat{y}_1^3 & \bullet\\[.2cm]
	\bullet         & \widehat{y}_2^2 & \bullet & \bullet\\[.2cm]
	\widehat{y}_3^1 & \widehat{y}_3^2 & \widehat{y}_3^3 & \widehat{y}_3^4\\
\end{bmatrix},
$$
where the symbol $\bullet$ denotes an empty entry of $\widehat{\Yvet}$. In this case, it is $p=4$ and $m = 7$. The vector containing all the available base forecasts is given by
$$
\widehat{\yvet} = \begin{bmatrix}[0.9]
	\widehat{\yvet}^1 \\ \widehat{\yvet}^2 \\ \widehat{\yvet}^3 \\ \widehat{\yvet}^4
\end{bmatrix} =
\begin{bmatrix}
	\widehat{y}_{1}^1 &
	\widehat{y}_{3}^1 &
	\widehat{y}_{2}^2 &
	\widehat{y}_{3}^2 &
	\widehat{y}_{1}^3 &
	\widehat{y}_{3}^3 &
	\widehat{y}_{3}^4
\end{bmatrix}^\top,
$$
and the selection matrices 
$$
\Lvet_1 = \Lvet_3 =
\begin{bmatrix}[0.9] 1 & 0 & 0 \\ 0 & 0 & 1 \end{bmatrix}, \quad
\Lvet_2 =
\begin{bmatrix}[0.9] 0 & 1 & 0 \\ 0 & 0 & 1 \end{bmatrix}, \quad
\Lvet_4 = \begin{bmatrix} 0 & 0 & 1 \end{bmatrix} ,
$$
are such that
$$
\Lvet_1\yvet = \Lvet_3\yvet = \begin{bmatrix}[0.9]
	y_1 \\ y_3
\end{bmatrix}, \quad \Lvet_2\yvet = \begin{bmatrix}[0.9]
	y_2 \\ y_3
\end{bmatrix}, \quad \Lvet_4\yvet = y_3 .
$$
Finally, it is
$$
\Lvet = \text{Diag}\left(\Lvet_1,\Lvet_2,\Lvet_3,\Lvet_4\right) = \begin{bmatrix}[0.9]
	1 & 0 & 0 & 0 & 0 & 0 & 0 & 0 & 0 & 0 & 0 & 0\\
	0 & 0 & 1 & 0 & 0 & 0 & 0 & 0 & 0 & 0 & 0 & 0\\
	0 & 0 & 0 & 0 & 0 & 1 & 0 & 0 & 0 & 0 & 0 & 0\\
	0 & 0 & 0 & 0 & 0 & 0 & 0 & 0 & 1 & 0 & 0 & 0\\
	0 & 0 & 0 & 0 & 0 & 0 & 0 & 0 & 0 & 1 & 0 & 0\\
	0 & 0 & 0 & 0 & 0 & 0 & 0 & 0 & 0 & 0 & 1 & 0\\
	0 & 0 & 0 & 0 & 0 & 0 & 0 & 0 & 0 & 0 & 0 & 1\\
\end{bmatrix}.
$$

              % Appendix C

\clearpage
%\section{Proofs of Theorem \ref{thm:mmse} and Corollary \ref{crl:unb}}\label{app:proof}
\section{Proofs of the results in Section~\ref{sec:OCC}}\label{app:proof}
\begin{proof}[Proof of \autoref{thm:mmse}: Optimal linear coherent forecast combination $\widetilde{\yvet}^c$] \;\\
\noindent Consider the augmented (lagrangean) objective function 
	\begin{align*}
		\mathcal{L}\left(\yvet,\lambdavet\right) & = \left(\widehat{\yvet} - \Kvet\yvet\right)^\top\Wvet^{-1}\left(\widehat{\yvet} - \Kvet\yvet\right) + 2\lambdavet^\top \Cvet\yvet \\
		& = \yvet^\top\Kvet^\top\Wvet^{-1}\Kvet\yvet - 2\widehat{\yvet}^\top\Wvet^{-1}\Kvet\yvet + \widehat{\yvet}^\top\Wvet^{-1}\widehat{\yvet} + 2\lambdavet^\top \Cvet\yvet .
	\end{align*}
	The first order condition, obtained by equating to zero the partial derivatives of $\mathcal{L}\left(\yvet,\lambdavet\right)$ wrt $\yvet$ and $\lambdavet$, respectively, results in:
	\begin{equation}
		\label{eq:derivate_new}
		\begin{bmatrix}
			\Kvet^\top\Wvet^{-1}\Kvet & \Cvet^\top \\
			\Cvet & \Zerovet_{(n_u \times 1)}
		\end{bmatrix}
		\begin{bmatrix}
			\;\yvet\;\; \\ \lambdavet
		\end{bmatrix} = 
		\begin{bmatrix}
			\Kvet^\top\Wvet^{-1}\widehat{\yvet} \\ \Zerovet_{(n_u \times 1)}
		\end{bmatrix}.
	\end{equation}
The optimal coherent combination forecast vector $\widetilde{\yvet}^c$ is thus obtained by solving equation (\ref{eq:derivate_new}):
	\begin{equation*}
		\label{eq:occ}
		\begin{bmatrix}
			\;\widetilde{\yvet}^c\;\; \\ \widetilde{\lambdavet}
		\end{bmatrix} = 
		\begin{bmatrix}
			\Kvet^\top\Wvet^{-1}\Kvet & \Cvet^\top \\[.15cm]
			\Cvet & \Zerovet_{(n_u \times 1)}
		\end{bmatrix}^{-1}
		\begin{bmatrix}
			\Kvet^\top\Wvet^{-1}\widehat{\yvet} \\[.15cm] \Zerovet_{(n_u \times 1)}
		\end{bmatrix}.
	\end{equation*}
	As $\Kvet^\top\Wvet^{-1}\Kvet$ has full rank, define $\Wvet_c = \left(\Kvet^\top\Wvet^{-1}\Kvet\right)^{-1}$. Then
	\begingroup 
\addtolength\jot{-2pt}
	\begin{align*}%\label{eq:ytildeocc}
			\widetilde{\yvet}^c &= \left[\Wvet_c-\Wvet_c\Cvet^\top\left(\Cvet\Wvet_c\Cvet^\top\right)^{-1}\Cvet \Wvet_c\right]\Kvet^\top\Wvet^{-1}\widehat{\yvet}\\
			&= \left[\Ivet_n-\Wvet_c\Cvet^\top\left(\Cvet\Wvet_c\Cvet^\top\right)^{-1}\Cvet\right]\Wvet_c\Kvet^\top\Wvet^{-1}\widehat{\yvet}\\
			&= \left[\Ivet_n-\Wvet_c\Cvet^\top\left(\Cvet\Wvet_c\Cvet^\top\right)^{-1}\Cvet\right]\Omegavet^\top\widehat{\yvet}\\
			&= \Mvet\Omegavet^\top\widehat{\yvet} = \Psivet^\top\widehat{\yvet},
		\end{align*}
		\endgroup
	where $\Psivet^\top = \Mvet\Omegavet^\top$, $\Omegavet = \Wvet^{-1}\Kvet\Wvet_c$, and 
	$\Mvet = \left[\Ivet_n-\Wvet_c\Cvet^\top\left(\Cvet\Wvet_c\Cvet^\top\right)^{-1}\Cvet\right]$.
\end{proof}

\begin{proof}[Proof of \autoref{crl:unb}: Unbiasedness of $\widetilde{\yvet}^c$ and an important property of its error covariance matrix]
The unbiasedness of the $p$ base forecasts $\widehat{\yvet}^j$ can be expressed in compact form as
\begin{equation*}
\label{unbiasedness_yhat}
E\left(\widehat{\yvet}\right) = \Kvet E\left(\yvet\right)= \Kvet\muvet .
\end{equation*}
Then, noting that
$\Omegavet^\top\Kvet = \Wvet_c\Kvet^\top\Wvet^{-1}\Kvet = \Ivet_n$ and $\Cvet\muvet = \Zerovet_{(n_u \times 1)}$,
it is
\begingroup 
\addtolength\jot{-5pt}
\begin{align*}
	E\left(\widetilde{\yvet}^c\right) &= \Mvet\Omegavet^\top E\left(\widehat{\yvet}\right) 
	 = \Mvet\Omegavet^\top\Kvet\muvet
	 = \Mvet\muvet \\ 
	&= \muvet - \Wvet_c\Cvet^\top\left(\Cvet\Wvet_c\Cvet^\top\right)^{-1}\Cvet\muvet = \muvet .
\end{align*}
\endgroup

\noindent To compute $\widetilde{\Wvet}_c$, we start by noting that $\widetilde{\yvet}^c - \yvet = \Mvet\Omegavet^\top\widehat{\yvet} - \yvet = \Mvet\Omegavet^\top\Kvet\yvet + \Mvet\Omegavet^\top\epsvet - \yvet = \Mvet\Omegavet^\top\epsvet$. The error covariance matrix of $\widetilde{\yvet}^c$ is thus given by
\begin{equation*}
\label{eq:Wc_long}
\widetilde{\Wvet}_c = E\left(\Mvet\Omegavet^\top\epsvet\epsvet^\top\Omegavet\Mvet^\top\right) =
\Mvet\Omegavet^\top\Wvet\Omegavet\Mvet^\top = \Mvet\Wvet_c\Mvet^\top ,
\end{equation*}
where the last expression is obtained by noting that
$$
\Omegavet^\top\Wvet\Omegavet = \Wvet_c\underbrace{\Kvet^\top\Wvet^{-1}\Wvet\Wvet^{-1}\Kvet}_{\Wvet_c^{-1}}\Wvet_c = \Wvet_c .
$$
Proceeding further in the calculations gives:
\begingroup 
\addtolength\jot{-5pt}
\begin{align*}
\widetilde{\Wvet}_c & =	\Mvet \Wvet_c \Mvet^\top\\
					& = \Wvet_c - \Wvet_c\Cvet^\top\left(\Cvet\Wvet_c\Cvet^\top\right)^{-1}\Cvet\Wvet_c - \Wvet_c\Cvet^\top\left(\Cvet\Wvet_c\Cvet^\top\right)^{-1}\Cvet\Wvet_c \\
	& \phantom{=} \qquad\qquad\qquad\qquad\qquad\qquad\qquad +\Wvet_c\Cvet^\top\left(\Cvet\Wvet_c\Cvet^\top\right)^{-1}\overbrace{\Cvet\Wvet_c\Cvet^\top\left(\Cvet\Wvet_c\Cvet^\top\right)^{-1}}^{\Ivet_{n_u}}\Cvet\Wvet_c\\
	& = \Wvet_c - 
	\Wvet_c\Cvet^\top\left(\Cvet\Wvet_c\Cvet^\top\right)^{-1}\Cvet\Wvet_c = \Mvet\Wvet_c .
\end{align*}
\endgroup

\noindent Finally, in order to derive (\ref{eq:Wtildebest}), we prove\footnote{The proof extends a result by \cite{Sun2004} about the covariance matrix of the multi-task combined forecast vector $\widehat{\yvet}^c$ in the balanced case.} first the inequality
\begin{equation}
\label{eq:Wc_vs_W}
\Lvet_j\Wvet_c\Lvet_j^\top \preceq \Wvet_{j} .
\end{equation}

\noindent Let $\Kvet_j = \left[\Zerovet_{(n_1 \times n)}\, \ldots\, \Lvet_j\, \ldots\, \Zerovet_{(n_p \times n)}\right]^\top \in \{0,1\}^{m \times n}$ be a matrix whose $j$-th block is $\Lvet_{j}$ and other blocks are $(n_l \times n)$ zero matrices, $l=1,\ldots, p$, with $l \ne j$. Observing that $$\Kvet_j\Lvet_j^\top = \left[\Zerovet_{(n_1 \times n_j)}\, \ldots\, \Ivet_{n_j}\, \ldots\, \Zerovet_{(n_p \times n_j)}\right]^\top$$ is a $(m \times n_j)$ matrix such that $(\Lvet_j\Kvet_j^\top)(\Kvet_j\Lvet_j^\top) = \Ivet_{n_j}$, and $\Kvet^\top \Kvet_j\Lvet_j^\top = \Lvet_j^\top$, it is easy to check that $\Lvet_j\Kvet_j^\top\Wvet\Kvet_j\Lvet_j^\top = \Wvet_j$. Therefore, applying Schwarz matrix inequality\footnote{Corollary (2.4) in \cite{Rao2000} (p. 311) says that, for matrices $\Avet$ and $\Bvet$ having the same number of columns, it is:
\begin{equation*}
	\label{eq:Rao_corollary}
	\Avet\Avet^\top \succeq \Avet\Bvet^\top\left(\Bvet\Bvet^\top\right)^{-}\Bvet\Avet^\top,
\end{equation*}
where the symbol $^{-}$ denotes the generalized inverse. In this case, it is: \label{footnote_Rao2000}
$$
\overbrace{\underbrace{\Lvet_j\Kvet_j^\top\Wvet^{\frac{1}{2}}}_{\Avet}
	\underbrace{\Wvet^{-\frac{1}{2}}\Kvet}_{\Bvet^\top}}^{\Lvet_j}
\overbrace{\left[\underbrace{\Kvet^\top\Wvet^{-\frac{1}{2}}}_{\Bvet}
	\underbrace{\Wvet^{-\frac{1}{2}}\Kvet}_{\Bvet^\top}
	\right]^{-1}}^{\Wvet_c}
\overbrace{\underbrace{\Kvet^\top\Wvet^{-\frac{1}{2}}}_{\Bvet}
	\underbrace{\Wvet^{\frac{1}{2}}\Kvet_j\Lvet_j^\top}_{\Avet^\top}}^{\Lvet_j^\top}.
$$}, we have:
\begin{align}
	\Lvet_j\Wvet_c\Lvet_j^\top & = \Lvet_j\left(\Kvet^\top\Wvet^{-1}\Kvet\right)^{-1}\Lvet_j^\top \nonumber\\
			& = 
			\left[\left(\Wvet^{-\frac{1}{2}}\Kvet\right)^\top\left(\Wvet^{\frac{1}{2}}\Kvet_j\Lvet_j^\top \right)\right]^\top
			\left[\left(\Wvet^{-\frac{1}{2}}\Kvet\right)^\top\left(\Wvet^{-\frac{1}{2}}\Kvet\right)\right]^{-1} \nonumber 
			 \left[\left(\Wvet^{-\frac{1}{2}}\Kvet\right)^\top\left(\Wvet^{\frac{1}{2}}\Kvet_j\Lvet_j^\top \right)\right] \nonumber\\
			& \preceq  \left(\Wvet^{\frac{1}{2}}\Kvet_j\Lvet_j^\top \right)^\top\left(\Wvet^{\frac{1}{2}}\Kvet_j\Lvet_j^\top \right) = \Lvet_j\Kvet_j^\top\Wvet\Kvet_j\Lvet_j^\top = \Wvet_{j} ,\label{eq:ineqL}
\end{align}
which proves $\Lvet_j\Wvet_c\Lvet_j^\top \preceq \Wvet_{j}$.

\noindent Consider now the inequality $\Lvet_j\widetilde{\Wvet}_c\Lvet_j^\top \preceq \Lvet_j\Wvet_c\Lvet_j^\top$. In Remark 3 we have shown that $\widetilde{\yvet}^c$ may be interpreted as a single-expert optimal combination reconciled forecast of the multi-task forecast combination vector $\widehat{\yvet}^c$. It follows that $\widetilde{\Wvet}_c \preceq \Wvet_c$ \citep{Panagiotelis2021}. In addition, since $\widetilde{\Wvet}_c = \Mvet\Wvet_c$, we can write 
\begin{equation}
	\label{eq:WtildevsWc}
\widetilde{\Wvet}_c = %\Mvet\Wvet_c = 
\Wvet_c - \Deltavet,
\end{equation}
where $\Deltavet = \Wvet_c\Cvet^\top\left(\Cvet\Wvet_c\Cvet^\top\right)^{-1}\Cvet\Wvet_c$ is a symmetric positive semidefinite matrix. Since pre-multiplying and post-multiplying a positive semidefinite matrix by the same matrix always gives a positive semidefinite matrix, pre- and post- multiplying expression (\ref{eq:WtildevsWc}) by $\Lvet_j$ gives:
\begin{equation}
	\Lvet_j\widetilde{\Wvet}_c\Lvet_j^\top = %\Lvet_j\Wvet_c\Lvet_j^\top - \Lvet_j\Deltavet\Lvet_j^\top = 
	\Lvet_j\Wvet_c\Lvet_j^\top - \Deltavet_{j}  \preceq \Lvet_j\Wvet_c\Lvet_j^\top ,\label{eq:Wtc}
\end{equation}
where $\Deltavet_{j} = \Lvet_j\Deltavet\Lvet_j^\top$ is a positive semidefinite matrix. Finally, from (\ref{eq:Wc_vs_W}) and (\ref{eq:Wtc}) it results:
\begin{equation*}
		\Lvet_j\widetilde{\Wvet}_c\Lvet_j^\top \preceq \Lvet_j\Wvet_c\Lvet_j^\top \preceq  \Wvet_{j} . 
		\vspace*{-1\baselineskip}
\end{equation*}
\end{proof}

\begin{proof}[Proof: Unbiasedness of $\widehat{\yvet}^c$]\;\\
The unbiasedness of the $p$ base forecasts can be expressed in compact form as
\begin{equation*}
E\left(\widehat{\yvet}\right)=\Kvet E\left(\yvet\right)=\Kvet\muvet .
\end{equation*}
Since $\widehat{\yvet}^c=\Omegavet^\top\widehat{\yvet}$ and
$$
\Omegavet^\top\Kvet = \Wvet_c\Kvet^\top\Wvet^{-1}\Kvet = \left(\Kvet^\top\Wvet^{-1}\Kvet\right)^{-1} \Kvet^\top\Wvet^{-1}\Kvet = \Ivet_n,
$$
it follows that
$$
E\left(\widehat{\yvet}^c\right)=\Omegavet^\top E\left(\widehat{\yvet}\right) =\Omegavet^\top\Kvet\muvet =\muvet,
$$
which proves that $\widehat{\yvet}^c$ is unbiased.
\end{proof}

%\section{A constrained Gauss--Markov result for coherent forecast combination}
%\label{app:GM}
%
%%\begin{prp}[Constrained Gauss--Markov result]
%%\label{prp:GM_structural}
%\noindent Among all linear unbiased coherent estimators of $\yvet$ obtainable from the available base forecasts, 
%the estimator $\widetilde{\yvet}^c$ in expression (\ref{eq:optcoherent_struc_be}) has minimum forecast 
%error covariance matrix in the nonnegative definite ordering.
%%\end{prp}

\begin{proof}[Proof of \autoref{crl:blue}: Best linear unbiased coherent estimator] \textit{A constrained Gauss--Markov result for coherent forecast combination}\;\\
For this proof, we consider the structural by-expert formulation shown in \autoref{tab:equiv_solutions}:
%Since $\yvet$ is coherent, it can be written as $\yvet=\Svet\bvet$ and 
\begin{align*}
\widehat{\yvet} &= \Kvet\Svet\bvet + \boldsymbol{\varepsilon}, \qquad
E(\epsvet)=\Zerovet, \qquad
E(\epsvet\epsvet^\top)=\Wvet , \\
\widetilde{\yvet}^c &= \Svet\Gvet\widehat{\yvet}, \qquad
			\Gvet = \left(\Svet^\top\Wvet^{-1}_c\Svet\right)^{-1}\Svet^\top\Kvet^\top\Wvet^{-1}, \qquad
			\Wvet_c = \left(\Kvet^\top\Wvet^{-1}\Kvet\right)^{-1}.
\end{align*}
Consider $\check{\yvet}=\Svet\Hvet\widehat{\yvet}$, an arbitrary unbiased linear coherent estimator of $\yvet$ obtainable from the base forecasts, where $\Hvet \in \mathbb{R}^{n_b \times m}$. Due to the unbiasedness condition $E(\check{\yvet})=%\Svet\Hvet\Kvet\Svet E(\bvet) =
\Svet E(\bvet) $, it is
 $\Hvet\Kvet\Svet=\Ivet_{n_b}$.

% E\left[\left(\check{\yvet} - \yvet\right)\left(\check{\yvet} - \yvet\right)^\top\right] - E\left[\left(\widetilde{\yvet}^c - \yvet\right)\left(\widetilde{\yvet}^c - \yvet\right)^\top\right]
\noindent Define $\Dvet=\Hvet-\Gvet$.
Since both $\Hvet$ and $\Gvet$ satisfy the unbiasedness condition, it follows that $\Dvet\Kvet\Svet=\Zerovet$.
Replacing $\Hvet=\Gvet+\Dvet$ into the forecast error covariance matrix of $\check{\yvet}$ 
$$
E\left[\left(\check{\yvet} - \yvet\right)\left(\check{\yvet} - \yvet\right)^\top\right] =\Svet\Hvet\Wvet\Hvet^\top\Svet^\top ,
$$
gives
\begin{align}
E\left[\left(\check{\yvet} - \yvet\right)\left(\check{\yvet} - \yvet\right)^\top\right]
&=
\Svet(\Gvet+\Dvet)\Wvet(\Gvet+\Dvet)^\top\Svet^\top \nonumber \\
&=
\Svet\Gvet\Wvet\Gvet^\top\Svet^\top
+
\Svet\Dvet\Wvet\Dvet^\top\Svet^\top +
\Svet\Gvet\Wvet\Dvet^\top\Svet^\top
+
\Svet\Dvet\Wvet\Gvet^\top\Svet^\top ,
\label{eq:GM_decomposition}
\end{align}
where the cross terms are null:
\begin{align}
\Svet\Gvet\Wvet\Dvet^\top\Svet^\top
&=
\Svet
\left(\Svet^\top\Wvet_c^{-1}\Svet\right)^{-1}
\Svet^\top
\Kvet^\top
\Wvet^{-1}
\Wvet
\Dvet^\top
\Svet^\top \nonumber \\
&=
\Svet
\left(\Svet^\top\Wvet_c^{-1}\Svet\right)^{-1}
\Svet^\top
\Kvet^\top
\Dvet^\top
\Svet^\top \nonumber \\
&=
\Svet
\left(\Svet^\top\Wvet_c^{-1}\Svet\right)^{-1}
\left(\Dvet\Kvet\Svet\right)^\top
\Svet^\top \nonumber =
\Zerovet ,
\end{align}
and
$$
\Svet\Dvet\Wvet\Gvet^\top\Svet^\top=\Zerovet .
$$
Therefore, expression (\ref{eq:GM_decomposition}) reduces to
$$
E\left[\left(\check{\yvet} - \yvet\right)\left(\check{\yvet} - \yvet\right)^\top\right]
=
\underbrace{\Svet\Gvet\Wvet\Gvet^\top\Svet^\top}_{E\left[\left(\widetilde{\yvet}^c - \yvet\right)\left(\widetilde{\yvet}^c - \yvet\right)^\top\right]}
+
\Svet\Dvet\Wvet\Dvet^\top\Svet^\top .
$$
Since $\Wvet$ is positive definite, 
$$
E\left[\left(\check{\yvet} - \yvet\right)\left(\check{\yvet} - \yvet\right)^\top\right] - E\left[\left(\widetilde{\yvet}^c - \yvet\right)\left(\widetilde{\yvet}^c - \yvet\right)^\top\right] = \Svet\Dvet\Wvet\Dvet^\top\Svet^\top  \succeq \Zerovet_{n\times n} ,
$$
which proves that $\widetilde{\yvet}^c$ has minimum forecast error covariance matrix within the class 
of all linear unbiased coherent estimators.
\end{proof}				% Appendix A

\clearpage
\section{Equivalence between occ and the sequential approaches: a special case}\label{app:case}
If we assume uncorrelated forecast errors across the $p$ models and that the covariance matrices of the $p$ models are equal ($\Wvet_{jj} = \overline{\Wvet}$ $\forall j = 1, \dots, p$), it is%, then
$$
\Wvet = \begin{bmatrix}
	\overline{\Wvet} & \cdots & \Zerovet_{(n \times n)} & \cdots & \Zerovet_{(n \times n)} \\
	\vdots & \ddots & \vdots & \ddots & \vdots \\
	\Zerovet_{(n \times n)} & \cdots & \overline{\Wvet} & \cdots & \Zerovet_{(n \times n)} \\
	\vdots & \ddots & \vdots & \ddots & \vdots \\
	\Zerovet_{(n \times n)} & \cdots &\Zerovet_{(n \times n)} & \cdots & \overline{\Wvet} \\
\end{bmatrix} = \Ivet_p \otimes \overline{\Wvet}.
$$
%with $\Wvet_{jj} = \overline{\Wvet}$ $\forall j = 1, \dots, p$, 
Then
$$
\Wvet_c^{-1} = \Kvet^\top\Wvet^{-1}\Kvet = \sum_{j=1}^p \Wvet_{jj}^{-1} = p\overline{\Wvet}^{\;-1} \qquad \text{and} \qquad \Wvet_c\Kvet^\top\Wvet^{-1}\widehat{\yvet} = 
%\sum_{j=1}^p \Wvet_{jj}^{-1}\widehat{\yvet}_j = p \overline{\Wvet}^{\;-1}  
\widehat{\yvet}^{\text{sa}},
$$
where $\widehat{\yvet}^{\text{sa}} = \displaystyle \frac{1}{p}\sum_{j=1}^p \widehat{\yvet}^j$
is the vector containing the simple average of the base forecasts.
If we substitute these quantities in \autoref{thm:mmse}, we get
$$
\begin{aligned}
	\widetilde{\yvet}^c
	&= \left[\Ivet_n-\frac{1}{p}\overline{\Wvet} \Cvet^\top\left(\Cvet\frac{1}{p}\overline{\Wvet}\Cvet^\top\right)^{-1}\Cvet\right]\frac{1}{p}\overline{\Wvet} \; \overline{\Wvet}^{\;-1} \sum_{j=1}^p \widehat{\yvet}_j \\
	&= \left[\Ivet_n-\overline{\Wvet} \Cvet^\top\left(\Cvet\overline{\Wvet}\Cvet^\top\right)^{-1}\Cvet\right]\frac{1}{p}\sum_{j=1}^p \widehat{\yvet}_j\\
	&= \left[\Ivet_n-\overline{\Wvet} \Cvet^\top\left(\Cvet\overline{\Wvet}\Cvet^\top\right)^{-1}\Cvet\right]\widehat{\yvet}^{\text{sa}} =
	\overline{\Mvet} \widehat{\yvet}^{\text{sa}}.
\end{aligned}
$$
where $\overline{\Mvet} = \left[\Ivet_n-\overline{\Wvet} \Cvet^\top\left(\Cvet\overline{\Wvet}\Cvet^\top\right)^{-1}\Cvet\right]$.
Thus the coherent combined forecast vector consists in the projection according to the metric matrix $\overline{\Wvet}$ onto the coherent linear subspace of the simple average of the base forecasts produced by the $p$ models.
In addition, it is worth noting that using the equal weight combination and the reconciliation through the covariance matrix $\overline{\Wvet}$, we obtain that%in this case, denoting with src and scr, respectively, the (sequential) procedures of applying forecast reconciliation using $\overline{\Wvet}$ first, and then equal weight combination, and equal weight combination first, and then forecast combination using $\overline{\Wvet}$, it is:
$$
\begin{array}{ccc}
\toprule
	\text{Procedure} & 1^{st} \text{step} & 2^{nd} \text{step} \\
	\midrule
	\text{src}       & \widetilde{\yvet}^j = \overline{\Mvet}\widehat{\yvet}^j & \widetilde{\yvet}^{src} = \displaystyle\frac{1}{p}\sum_{j=1}^{p}\widetilde{\yvet}^j = \overline{\Mvet}\widehat{\yvet}^{\text{sa}}\\
	\text{scr}       & \widehat{\yvet}^{\text{sa}} = \displaystyle\frac{1}{p}\sum_{j=1}^{p}\widehat{\yvet}^j & \widetilde{\yvet}^{scr} = \overline{\Mvet} \widehat{\yvet}^{\text{sa}} \\
	\bottomrule
\end{array}
$$
that is:
$$
\widetilde{\yvet}^{src} = \widetilde{\yvet}^{scr} = \widetilde{\yvet}^c.
$$	% Appendix D

\clearpage
\section{Equivalent derivations of the optimal coherent linear forecast combination}
\label{app:alternative_proofs}

\noindent Let us first prove that the solution found with the zero-constrained approach is the same as the one obtained with the structural approach. At this end, consider the model and the linear constraints of the zero-constrained approach according to the `by-expert' representation\footnote{An analogous reasoning may be followed if the `by-variable' representation is adopted.}:
\begin{equation}
\label{eq:zc_be}
\widehat{\yvet} = \Kvet\yvet + \epsvet, \text{ s.t. } \Cvet\yvet=\Zerovet.
\end{equation}
The constraint $\Cvet\yvet = \Zerovet$ is equivalent to\footnote{$\Cvet\yvet = \Zerovet$ means $\uvet - \Avet\bvet = \Zerovet$, i.e., $\uvet = \Avet\bvet$, while $\yvet = \Svet\bvet$ corresponds to $\uvet = \Avet \bvet$ and $\bvet = \bvet$, this last expression being trivially true.} $\yvet = \Svet\bvet$, that can be explicitly incorporated into expression (\ref{eq:zc_be}). Since it is $\Cvet\Svet = \Zerovet_{(n_u \times n_b)}$, this operation results in an unrestricted linear model:
\begin{equation}
\label{eq:st_be}
\widehat{\yvet} = \Kvet\Svet\bvet + \epsvet,  \text{ s.t. } \Cvet\Svet\bvet=\Zerovet \quad \Rightarrow \quad \widehat{\yvet} = \Kvet\Svet\bvet + \epsvet ,
\end{equation}
and the rhs of expression (\ref{eq:st_be}) corresponds to the linear model of the structural approach according to the `by-expert' base forecasts' representation. Then, computing the coherent forecast combination of $\yvet$ through either model (\ref{eq:zc_be}) or (\ref{eq:st_be}) gives the same result, the only difference being the closed-form expressions of $\widehat{\yvet}^c$ and its error covariance matrix $\widetilde{\Wvet}_c$ according to either approaches (see Table \ref{tab:equiv_solutions}). This establishes the equivalence between zero-constrained and structural approaches in the optimal coherent linear forecast combination method\footnote{In passing, this way of operating also establishes the equivalence between zero-constrained and structural approaches when $p=1$, which corresponds to the optimal linear cross-sectional forecast reconciliation, complementing in a sense the result shown by \cite{Wickra2019} in their online appendix, where this simple step is missing.}.

\subsection{Zero-constrained approach and by-expert base forecasts' organization}
\noindent See the proofs of \autoref{thm:mmse} and \autoref{crl:unb} in \autoref{app:proof}.

\subsection{Zero-constrained approach and by-variable base forecasts' organization}

\noindent Pre-multiplication of model (\ref{mod:stone_unconstrained_selection}) by the permutation matrix $\Pvet \in \{0,1\}^{m \times m}$ (see section \ref{sec:be_vs_bv}) gives the equivalent, re-parameterized model according to a by-variable data organization:
\begin{equation}
\label{mod:stone_unconstrained_bv}
\widehat{\yvet}_{\text{bv}} = \Jvet\yvet + \epsvet_{\text{bv}} ,
\end{equation}
where $\widehat{\yvet}_{\text{bv}} = \Pvet\widehat{\yvet}$, $\epsvet_{\text{bv}} = \Pvet\epsvet$,
\begin{equation}
\label{eq:Jmatrix}
\Jvet = \Pvet\Kvet =
\begin{bmatrix}[0.7]
\Unovet_{p_1} & \cdots & \Zerovet_{(p_1 \times 1)} & \cdots & \Zerovet_{(p_1 \times 1)} \\
\vdots & \ddots & \vdots & \ddots & \vdots \\
\Zerovet_{(p_i \times 1)} & \cdots & \Unovet_{p_i} & \cdots & \Zerovet_{(p_i \times 1)} \\
\vdots & \ddots & \vdots & \ddots & \vdots \\
\Zerovet_{(p_n \times 1)} & \cdots & \Zerovet_{p_n \times 1} & \cdots & \Unovet_{p_n}
\end{bmatrix}
\in \{0,1\}^{m \times n},
\end{equation}
and $\epsvet_{\text{bv}} \in \mathbb{R}^m$ is a zero-mean random vector, with covariance matrix $\Sigmavet = E\left(\epsvet_{\text{bv}}\epsvet_{\text{bv}}^\top\right) = \Pvet\Wvet\Pvet^\top$. In this case, the solution to the linearly constrained quadratic program
\begin{equation}\label{mod:stone_new_bv}
	\widetilde{\yvet}^c = \argmin_{\yvet} 
	\left(\widehat{\yvet}_{\text{bv}} - \Jvet\yvet\right)^\top\Sigmavet^{-1}\left(\widehat{\yvet}_{\text{bv}} - \Jvet\yvet\right) \quad
	\text{s.t. } \Cvet\yvet = \Zerovet_{(n_u \times 1)}
\end{equation}
can be obtained as for the case of by-expert data organization (see \ref{app:proof}):
\begin{equation}
	\label{eq:ytilde_occ_bv}
	\widetilde{\yvet}^c = \Phivet^\top\widehat{\yvet}_{\text{bv}} = \Mvet\Gammavet^\top\widehat{\yvet}_{\text{bv}},
\end{equation}
with weight matrix $\Phi^\top = \Mvet\Gammavet^\top \in \mathbb{R}^{n \times m}$, where
\begingroup 
%\addtolength\jot{-5pt}
\begin{align}
	\Mvet & = \left[\Ivet_n - \Sigmavet_c\Cvet^\top\left(\Cvet\Sigmavet_c\Cvet^\top\right)^{-1}\Cvet\right],\label{eq:Mvet_bv}\\
	\Gammavet & = \Sigmavet^{-1}\Jvet\Sigmavet_c,\label{eq:Gammavet} \\
	\Sigmavet_c & = \left(\Jvet^\top\Sigmavet^{-1}\Jvet\right)^{-1}.\label{eq:Wc_bv}
\end{align}
\endgroup
In order to establish the equivalence between expressions (\ref{eq:ytilde_occ_bv}) and (\ref{eq:ytilde_occ}), it will suffice to show that $\Wvet_c = \Sigmavet_c$. For, noting that $\Kvet = \Pvet^\top \Jvet$ and $\Wvet = \Pvet^\top\Sigmavet\Pvet$, we have:
$$
\Wvet_c = \left(\Kvet^\top\Wvet^{-1}\Kvet\right)^{-1} =  \left(\Jvet^\top\Pvet\Pvet^\top\Sigmavet^{-1}\Pvet\Pvet^\top\Jvet\right)^{-1} =
\left(\Jvet^\top\Sigmavet^{-1}\Jvet\right)^{-1} =
\Sigmavet_c .
$$
It follows that:
\begin{itemize}
	\item $ \left[\Ivet_n - \Wvet_c\Cvet^\top\left(\Cvet\Wvet_c\Cvet^\top\right)^{-1}\Cvet\right] = \left[\Ivet_n - \Sigmavet_c\Cvet^\top\left(\Cvet\Sigmavet_c\Cvet^\top\right)^{-1}\Cvet\right] = \Mvet$;
	\item $\Omegavet = \Pvet^\top\Gammavet$ and $\Gammavet = \Pvet\Omegavet$;
	\item $\widetilde{\yvet}^c = \Mvet\Omegavet^\top\widehat{\yvet} = \Mvet\Gammavet^\top\widehat{\yvet}_{\text{bv}}$.
\end{itemize}
In addition, denoting $\Gammavet = \begin{bmatrix} \Gammavet_1 \; \cdots \; \Gammavet_i \; \cdots \; \Gammavet_n \end{bmatrix}^\top$ and $\Phivet = \begin{bmatrix} \Phivet_1 \; \cdots \; \Phivet_i \; \cdots \; \Phivet_n \end{bmatrix}^\top$, we can express the multi-task combination forecast $\widehat{\yvet}^c$ and the multi-task coherent combination forecast  $\widetilde{\yvet}^c$ as, respectively,
$$
\widehat{\yvet}^c = \Gammavet^\top\widehat{\yvet}_{\text{bv}} =
\displaystyle\sum_{i=1}^{n}\Gammavet_i\widehat{\yvet}_i
\quad\text{ and }
\quad
\widetilde{\yvet}^c = \Phivet^\top\widehat{\yvet}_{\text{bv}} =
\displaystyle\sum_{i=1}^{n}\Phivet_i\widehat{\yvet}_i ,
$$
with $\Phivet_i = \Mvet\Gammavet_i$, $i=1, \ldots, n$.

\subsection{Structural approach and by-expert base forecasts' organization}

\noindent Deriving the multi-task coherent combination forecast according to the structural approach can be seen as a multi-task extension of the optimal combination forecast reconciliation approach proposed by \cite{Athanasopoulos2009} and \cite{Hyndman2011}. This approach, originally developed to reconcile the incoherent base forecasts of a genuine hierarchical/grouped time series, has been extended by \cite{Girolimetto2024} to deal with also general linearly constrained multiple series not having a natural hierarchical/grouped structure. In particular, \cite{Girolimetto2024} have shown how a general linearly constrained time series can be represented in a structural-like form, as in expression (\ref{eq:structural-like-formulation}), i.e.,  $\yvet = \Svet\bvet$. Building on this, in the following we derive the same results as \autoref{thm:mmse} and \autoref{crl:unb}.
 
Let us consider the $p$ linear models
\begin{equation}
\label{eq:structural_be}
\widehat{\yvet}^j = \Lvet_j\Svet\bvet + \epsvet^j, \quad j=1, \ldots, p,
\end{equation}
where $\Lvet_j$ is the selection matrix defined in Section~\ref{sec:CFC}. Expression (\ref{eq:structural_be}) can be written in compact form as
\begin{equation}
	\label{eq:matrix_structural_be}
	\widehat{\yvet} = \Kvet\Svet\bvet + \epsvet,
\end{equation}
with $E\left(\epsvet\epsvet^\top\right) = \Wvet$. Coherent multi-task forecast combination can be achieved through model (\ref{eq:matrix_structural_be}), by extending to the case of $p \ge 2$ vectors of base forecasts the minimum trace (MinT) solution proposed by \cite{Wickra2019} as for cross-sectional forecast reconciliation. This result is shown by the following Theorem.

\begin{thm}[Optimal linear coherent forecast combination $\widetilde{\yvet}^c$ (structural approach)] 
\label{thm:mmse_structural}
Let $\widehat{y}_i^j$, $1 \le n_j \le n$, $1 \le j \le p_i$, be the unbiased forecasts of the scalar variable $y_i$, which is part of an $n$-dimensional target forecast vector $\yvet$, and denote $\widehat{\yvet}^j \in \mathbb{R}^{n_j}$ the vector of the base forecasts produced by the $j$-th expert. Let the zero-mean base forecast errors be $\epsvet^j = \widehat{\yvet}^j - \yvet$, $j=1, \ldots, p$, and assume that $\epsvet^{j}$ and $\epsvet^{l}$ ($j \ne l$) are correlated, with variance and cross-covariance matrices denoted by $\Wvet_{j} \in \mathbb{R}^{n_j \times n_j}$ and $\Wvet_{jl} \in \mathbb{R}^{n_j \times n_l}$, respectively. Denoting $\Wvet \in \mathbb{R}^{m \times m}$ the p.d. error covariance matrix as in (\ref{eq:Wmatrix}), the coherent linear forecast combination is given by:
\begin{equation}
	\label{eq:optcoherent_struc_be}
	\widetilde{\yvet}^c =
 \Svet\Gvet\widehat{\yvet},
\end{equation}
where 
\begin{equation}
	\label{eq:Gmatrix_be}
	\Gvet = \left(\Svet^\top\Wvet^{-1}_c\Svet\right)^{-1}\Svet^\top\Kvet^\top \Wvet^{-1}.
\end{equation}
\end{thm}

\begin{proof}
From the proof of Corollary \ref{crl:unb}, we know that the unbiasedness of the base forecasts may be expressed in compact form as $E\left(\widehat{\yvet}\right) = \Kvet E(\yvet)$. In addition, it is worth noting that $E(\yvet) = \Svet E(\bvet)$, and thus $E\left(\widetilde{\yvet}^c\right) = \Svet\Gvet E\left(\widehat{\yvet}\right) = \Svet\Gvet\Kvet E(\yvet) = \Svet\Gvet\Kvet \Svet E(\bvet)$. Then, for $\widetilde{\yvet}^c$ to be unbiased, i.e., $E\left(\widetilde{\yvet}^c\right) = E(\yvet) = \Svet E(\bvet)$, it must be $\Svet\Gvet\Kvet\Svet = \Svet$, that is:
\begin{equation}
	\label{eq:unbiasedness_condition}
	\Gvet\Kvet\Svet = \Ivet_{n_b}.
\end{equation}
The $(n \times 1)$ vector of the coherent combined forecast error $\widetilde{\epsvet} = \widetilde{\yvet}^c - \yvet$, may be expressed as
$$
\widetilde{\epsvet} = \Svet\Gvet\widehat{\yvet} - \yvet =
\Svet\Gvet\left(\widehat{\yvet} - \Kvet\yvet\right).
$$
Noting that $\Wvet = E\big[\left(\widehat{\yvet} - \Kvet\yvet\right)\;\left(\widehat{\yvet} - \Kvet\yvet\right)^\top\big]$, the covariance matrix of $\widetilde{\yvet}^c$ is given by $\widetilde{\Wvet}_c = E\big[\widetilde{\epsvet}\;\widetilde{\epsvet}^\top\big] = \Svet\Gvet\Wvet\Gvet^\top \Svet^\top$. We want to find the optimal combination weight matrix $\Gvet$ under restriction (\ref{eq:unbiasedness_condition}) that minimizes the trace of $\widetilde{\Wvet}_c$, i.e., $\text{tr}\left(\Svet\Gvet\Wvet\Gvet^\top \Svet^\top\right)$. At this end, consider the lagrangean function
$$
\mathcal{L}\left(\Gvet,\Lambdavet\right) = \text{tr}\left(\Svet\Gvet\Wvet\Gvet^\top \Svet^\top\right) - 2\text{tr}\left[\Lambdavet^\top \left(\Gvet\Kvet\Svet - \Ivet_{n_b}\right) \right] ,
$$
where $\Lambdavet \in \mathbb{R}^{n_b \times n_b}$ is a matrix of Lagrange multipliers. Solving the first-order conditions system
\begin{align*}
\Wvet\Gvet^\top\Svet^\top\Svet - \Kvet\Svet \Lambdavet & = \Zerovet_{m \times n_b} \\
\Gvet\Svet - \Ivet_{n_b} & = \Zerovet_{(n_b \times n_b)}
\end{align*}
results in the following expressions:
\begin{align*}
	\Kvet\Svet \Lambdavet  &= \Wvet\Gvet^\top\Svet^\top\Svet \\
	\Wvet^{-1}\Kvet\Svet \Lambdavet &= \Gvet^\top\Svet^\top\Svet &  \mbox{{\footnotesize pre-multiplying by} }\Wvet^{-1}\mbox{{\footnotesize on both sides}}\\
	\Svet^\top\Kvet^\top\Wvet^{-1}\Kvet\Svet \Lambdavet &= \Svet^\top\Kvet^\top\Gvet^\top\Svet^\top\Svet &  \mbox{{\footnotesize pre-multiplying by} }\Svet^\top\Kvet^\top \mbox{ {\footnotesize on both sides}}\\
	\Svet^\top\Wvet^{-1}_c\Svet \Lambdavet &= \Svet^\top\Kvet^\top\Gvet^\top\Svet^\top\Svet &  \Wvet^{-1}_c = \Kvet^\top\Wvet^{-1}\Kvet\\
	\Svet^\top\Wvet^{-1}_c\Svet \Lambdavet &= \Svet^\top\Svet & \mbox{{\footnotesize according to (\ref{eq:unbiasedness_condition}),} } \Svet^\top\Kvet^\top\Gvet^\top = \Ivet_{n_b}\\
	\Lambdavet &= \left(\Svet^\top\Wvet^{-1}_c \Svet\right)^{-1}\Svet^\top\Svet.
\end{align*}
Then, 
\begingroup 
%\addtolength\jot{-5pt}
\begin{align*}
	\Wvet\Gvet^\top\Svet^\top\Svet - \Kvet\Svet \Lambdavet &= \Zerovet \\
	\Wvet\Gvet^\top\Svet^\top\Svet &=  \Kvet\Svet \Lambdavet\\
	\Wvet\Gvet^\top\Svet^\top\Svet &= \Kvet\Svet \left(\Svet^\top\Wvet^{-1}_c\Svet\right)^{-1}\Svet^\top\Svet \\
	\Wvet\Gvet^\top &= \Kvet\Svet \left(\Svet^\top\Wvet^{-1}_c\Svet\right)^{-1} & \quad \mbox{{\footnotesize post-multiplying by}} \left(\Svet^\top\Svet\right)^{-1} \mbox{{\footnotesize on both sides}}\\
	\Gvet^\top &=  \Wvet^{-1}\Kvet\Svet\left(\Svet^\top\Wvet^{-1}_c\Svet\right)^{-1} & \quad \mbox{{\footnotesize pre-multiplying by} } \Wvet^{-1} \mbox{{\footnotesize on both sides}}\\
	\Gvet &= \left(\Svet^\top\Wvet^{-1}_c\Svet\right)^{-1}\Svet^\top\Kvet^\top \Wvet^{-1} .
\end{align*}
\endgroup
\vspace*{-2\baselineskip} \;\\
\end{proof}

\begin{rml}
Since $\Cvet\Svet = \Zerovet_{(n_u \times n_b)}$, the combined forecast vector (\ref{eq:optcoherent_struc_be}) is also coherent, i.e., $\Cvet\widetilde{\yvet}^c = \Zerovet_{(n_u \times 1)}$. The unbiasedness property may be checked by noting that $E\left(\Svet\Gvet\widehat{\yvet}\right) = \Svet\Gvet E\left(\widehat{\yvet}\right) = \Svet\Gvet\Kvet E(\yvet)$. Then,
$$
	E\left(\Svet\Gvet\widehat{\yvet}\right) = \Svet\overbrace{\Gvet\Kvet\Svet}^{\Ivet_{n_b}} E(\bvet) = \Svet E(\bvet) = E(\yvet).
$$
\end{rml}

\begin{rml}
A straightforward interpretation of expressions (\ref{eq:optcoherent_struc_be}) and (\ref{eq:Gmatrix_be}) is obtained by considering the GLS estimator of $\bvet$ in model (\ref{eq:matrix_structural_be}):
$$
	\widetilde{\bvet}^c = \left(\Svet^\top\Wvet^{-1}_c\Svet\right)^{-1}\Svet^\top\Kvet^\top\Wvet^{-1}\widehat{\yvet} = \Gvet\widehat{\yvet},
$$
where $\Gvet$ is the $(n_b \times m)$ optimal combination weight matrix mapping all the $m$ available forecasts into the $n_b$ reconciled forecasts of the free (bottom) variables of series $\yvet_t$. This result is then used to compute the complete vector of coherent combined forecasts according to the structural representation linking the target vectors $\bvet$ and $\yvet$, that is, $\widetilde{\yvet}^c = \Svet\widetilde{\bvet}^c = \Svet\Gvet\widehat{\yvet}$. In addition, since $E\left[\left(\widetilde{\bvet}^c - \bvet\right)\left(\widetilde{\bvet}^c - \bvet\right)^\top\right] = \left(\Svet^\top\Wvet^{-1}_c\Svet\right)^{-1}$, it follows $E\left[\left(\widetilde{\yvet}^c - \yvet\right)\left(\widetilde{\yvet}^c - \yvet\right)^\top\right] = \Svet\left(\Svet^\top\Wvet^{-1}_c\Svet\right)^{-1}\Svet^\top = \widetilde{\Wvet}_c$ (see Corollary \ref{crl:Wtilde_structural}).
\end{rml}

\begin{crl}[Property of the error covariance matrix (structural approach)]\label{crl:Wtilde_structural}
The error covariance matrix of $\widetilde{\yvet}^c$ in expression (\ref{eq:optcoherent_struc_be}) is given by 
	\begin{equation}
		\label{eq:coherentWc_be}
		\widetilde{\Wvet}_c = \Svet\left(\Svet^\top\Wvet^{-1}_c\Svet\right)^{-1}\Svet^\top.
	\end{equation}
Furthermore, it is
$$
\Lvet_j\widetilde{\Wvet}_c\Lvet_j^\top \preceq \Wvet_{j}, \quad j=1, \ldots, p.
$$
\end{crl}

\begin{proof}
In the proof of Theorem \ref{thm:mmse_structural} it was noted that the error covariance matrix of $\widetilde{\yvet}^c$ is equal to $\widetilde{\Wvet}_c = \Svet\Gvet\Wvet\Gvet^\top \Svet^\top$. Replacing $\Gvet$ with expression (\ref{eq:Gmatrix_be}) gives:
\begingroup 
%\addtolength\jot{-5pt}
\begin{align}
\widetilde{\Wvet}_c & = \Svet\left(\Svet^\top\Wvet^{-1}_c\Svet\right)^{-1} \underbrace{\Svet^\top\Kvet^\top\Wvet^{-1} \Wvet\Wvet^{-1} \Kvet\Svet\left(\Svet^\top\Wvet^{-1}_c\Svet\right)^{-1}}_{\Ivet_{n_b}} \Svet^\top \nonumber \\
 & = \Svet\left(\Svet^\top\Wvet^{-1}_c\Svet\right)^{-1} \Svet^\top . \label{eq:Wtilde_struc_final}
\end{align}
\endgroup

\noindent As for the derivation of $\Lvet_j\widetilde{\Wvet}_c\Lvet_j^\top \preceq \Wvet_{j}$, $j=1, \ldots, p$, by applying Schwarz matrix inequality, we have\footnote{It is worth noting that $\left(\Wvet^{-1/2}\Kvet\Svet\right)^\top\left(\Wvet^{1/2}\Kvet_j\right) = \Svet^\top\Kvet^\top\Kvet_j =  \Svet^\top\Lvet_j^\top$.}$^,$\footnote{According to footnote \ref{footnote_Rao2000}, in this case it is:
	$$
	\overbrace{\underbrace{\Kvet_j^\top\Wvet^{1/2}}_{\Avet}
		\underbrace{\Wvet^{-1/2}\Kvet\Svet}_{\Bvet^\top}}^{\Svet}
	\overbrace{\Bigg[\underbrace{\Svet^\top\Kvet^\top\Wvet^{-1/2}}_{\Bvet}
		\underbrace{\Wvet^{-1/2}\Kvet\Svet}_{\Bvet^\top}
		\Bigg]^{-1}}^{\left[\Svet^\top\Kvet^\top\Wvet^{-1}\Kvet\Svet\right]^{-1}}
	\overbrace{\underbrace{\Svet^\top\Kvet^\top\Wvet^{-1/2}}_{\Bvet}
		\underbrace{\Wvet^{1/2}\Kvet_j}_{\Avet^\top}}^{\Svet^\top} .
	$$
}
\begingroup 
%\addtolength\jot{-5pt}
	\begin{align} 
		\Lvet_j\widetilde{\Wvet}_c\Lvet_j^\top & = \Lvet_j\Svet\left(\Svet^\top\Wvet^{-1}_c\Svet\right)^{-1}\Svet^\top\Lvet_j^\top  = \Lvet_j\Svet\left(\Svet^\top\Kvet^\top\Wvet^{-1}\Kvet\Svet\right)^{-1}\Svet^\top\Lvet_j^\top \nonumber \\
		& = \left[\left(\Wvet^{-1/2}\Kvet\Svet\right)^\top\left(\Wvet^{1/2}\Kvet_j\Lvet_j^\top\right)\right]^\top \left[\left(\Wvet^{-1/2}\Kvet\Svet\right)^\top\left(\Wvet^{-1/2}\Kvet\Svet\right)\right]^{-1} \nonumber\\
		& \qquad\qquad\qquad\qquad\qquad\qquad\qquad\qquad\qquad\qquad \left[\left(\Wvet^{-1/2}\Kvet\Svet\right)^\top\left(\Wvet^{1/2}\Kvet_j\Lvet_j^\top\right)\right] \nonumber \\
		& \preceq  \left(\Wvet^{1/2}\Kvet_j\Lvet_j^\top\right)^\top\left(\Wvet^{1/2}\Kvet_j\Lvet_j^\top\right) = \Wvet_{j} \label{eq:SchwarzCoheren},
	\end{align}
\endgroup
where $\Kvet_j = \left[\Zerovet \ldots \Lvet_j \ldots \Zerovet\right]^\top$ is a $(m \times n)$ matrix whose $j$-th block place is the selection matrix $\Lvet_j$, and other blocks are $(n_l \times n)$ zero matrices, $l=1,\ldots,p, \; l \ne j$.
\end{proof}

\subsection{Structural approach and by-variable base forecasts' organization}

\noindent Pre-multiplication of model (\ref{eq:matrix_structural_be}) by the permutation matrix $\Pvet$ gives the equivalent, re-parameterized structural model according to a by-variable data organization:
\begin{equation}
\label{eq:matrix_structural_bb}
\widehat{\yvet}_{\text{bv}} = \Jvet\Svet\bvet + \epsvet_{\text{bv}},
\end{equation}
where $\Jvet = \Pvet\Kvet \in \{0,1\}^{m \times n}$ is defined by (\ref{eq:Jmatrix}). In this case, the coherent combination is given by:
\begin{equation}
	\label{eq:optcoherent_struc_bv}
	\widetilde{\yvet}^c = \Svet\Gvet_{\text{bv}}\widehat{\yvet}_{\text{bv}},
\end{equation}
where $\Gvet_{\text{bv}} = \left(\Svet^\top\Jvet^\top\Sigmavet^{-1}\Jvet\Svet\right)^{-1}\Svet^\top\Jvet^\top \Sigmavet^{-1}$.
Since $\Jvet^\top\Sigmavet^{-1}\Jvet = \Kvet^\top\Pvet^\top\Sigmavet^{-1}\Pvet = \Kvet^\top\Wvet^{-1}\Kvet = \Wvet_c^{-1}$, it is $\Gvet_{\text{bv}} = \left(\Svet^\top\Wvet^{-1}_c\Svet\right)^{-1}\Svet^\top\Kvet^\top \Wvet^{-1} = \Gvet$, and  $\left(\Svet^\top\Wvet^{-1}_c\Svet\right) = \left(\Svet^\top\Jvet^\top\Pvet\Pvet^\top\Sigmavet^{-1}\Pvet\Pvet^\top\Jvet\Svet\right)= \left(\Svet^\top\Sigmavet^{-1}_c\Svet\right)$. In addition, $\Gvet_{\text{bv}} = \Gvet\Pvet^\top$, $\Gvet = \Gvet_{\text{bv}}\Pvet$, and $\widetilde{\yvet}_c = \Svet\Gvet\widehat{\yvet} = \Svet\Gvet_{\text{bv}}\widehat{\yvet}_{\text{bv}}$.

% Appendix D

\clearpage
\clearpage

\section{Simulation: extended figures and tables}\label{app:sim}

\noindent The forecast accuracy is evaluated also using the Average Relative Mean Squared Error ($AvgRelMSE$), given by
$$
AvgRelMSE^{app} = \left(\prod_{s = 1}^{500} \prod_{i = 1}^{n} \frac{MSE_{s,i}^{app}}{MSE_{s,i}^{\text{ew}}}\right)^{\frac{1}{500n}} \; \text{with} \;\; MSE_{s,i}^{app} = \displaystyle\frac{1}{100}\sum_{q = 1}^{100} \left(y_{i,s,q}-\check{y}_{i,s,q}^{\,app}\right)^2,
$$
where $i=1,...,n$ denotes the variable, $app$ is the approach used, $y_{i,s,q}$ is the observed value and $\check{y}_{i,s,q}^{\,app}$ is the forecast value using the $app$ approach (either coherent or incoherent, \autoref{tab:approaches_app}). 

\vskip2cm

\begin{table}[!htb]
	\centering
		\caption{Summary of forecasting approaches used in the simulation and in the forecasting experiment on the Australian daily electricity generation time series.
		For single-expert reconciliation and coherent combination sequential approaches, the shrunk in-sample MSE %error covariance 
		matrix is used for the reconciliation.}
	\label{tab:approaches_app}
	\footnotesize
	\setlength{\tabcolsep}{3pt}
	\resizebox{\linewidth}{!}{
\begin{tabular}{c|c}
		\toprule
		\begin{minipage}[t]{0.5\linewidth}
			\textbf{Approach \& description}
		\end{minipage} &  \begin{minipage}[t]{0.5\linewidth}
			\textbf{Approach \& description}
		\end{minipage} \\
\bottomrule
\end{tabular}}
\end{table}

\begin{table}[H]
	\centering
	\scriptsize
	\setlength{\tabcolsep}{2pt}
	\def\arraystretch{1}
	\caption{AvgRelMSE for the simulation experiment. Benchmark approach: ew. Bold entries identify the best performing approaches, italic entries identify the second best and red color denotes forecasts worse than the benchmark.}
	\resizebox*{0.85\linewidth}{!}{
		
\begin{tabular}[t]{cc>{}c|ccccccccc>{}c|cccccc}
\toprule
\multicolumn{3}{c}{\textbf{ }} & \multicolumn{10}{c}{\textbf{Balanced panel of forecasts}} & \multicolumn{6}{c}{\textbf{Unbalanced panel of forecasts}} \\
\textbf{Sett.} & \textbf{p} & \textbf{N} & \rotatebox{90}{base$^{\ast}$} & \rotatebox{90}{base$^{\ast}_{\text{shr}}$} & \rotatebox{90}{base$_{\text{shr}}$} & \rotatebox{90}{ow$_{\text{var}}$} & \rotatebox{90}{ow$_{\text{cov}}$} & \rotatebox{90}{src} & \rotatebox{90}{scr$_{\text{ew}}$} & \rotatebox{90}{scr$_{\text{var}}$} & \rotatebox{90}{scr$_{\text{cov}}$} & \rotatebox{90}{occ$_{\text{be}}$} & \rotatebox{90}{ow$_{\text{var}}$} & \rotatebox{90}{ow$_{\text{cov}}$} & \rotatebox{90}{scr$_{\text{ew}}$} & \rotatebox{90}{scr$_{\text{var}}$} & \rotatebox{90}{scr$_{\text{cov}}$} & \rotatebox{90}{occ$_{\text{be}}$}\\
\midrule
1 & 4 & 50 & \textcolor{red}{1.712} & \textcolor{red}{1.130} & \textcolor{red}{1.120} & \textcolor{red}{1.004} & \textcolor{red}{1.056} & \em{0.847} & 0.901 & 0.905 & 0.939 & \textbf{0.818} & \textcolor{red}{1.063} & \textcolor{red}{1.145} & \em{0.815} & \textcolor{red}{1.111} & \textcolor{red}{1.193} & \textbf{0.792}\\
 &  & 100 & \textcolor{red}{1.708} & \textcolor{red}{1.078} & \textcolor{red}{1.074} & \textcolor{red}{1.001} & \textcolor{red}{1.029} & \em{0.837} & 0.895 & 0.896 & 0.913 & \textbf{0.801} & \textcolor{red}{1.046} & \textcolor{red}{1.145} & \em{0.807} & \textcolor{red}{1.062} & \textcolor{red}{1.178} & \textbf{0.779}\\
 &  & 200 & \textcolor{red}{1.702} & \textcolor{red}{1.052} & \textcolor{red}{1.049} & 1.000 & \textcolor{red}{1.013} & \em{0.832} & 0.888 & 0.888 & 0.895 & \textbf{0.791} & \textcolor{red}{1.020} & \textcolor{red}{1.129} & \em{0.801} & \textcolor{red}{1.002} & \textcolor{red}{1.148} & \textbf{0.771}\\
\bottomrule
\end{tabular}

	}
\end{table}

\subsection{Simulation with uncorrelated forecast errors}
\begin{table}[H]
	\centering
	\scriptsize
	\setlength{\tabcolsep}{2pt}
	\def\arraystretch{1}
	\caption{AvgRelMAE for the simulation experiment (balanced panel of forecasts). Benchmark approach: ew. Bold entries identify the best performing approaches, italic entries identify the second best and red color denotes forecasts worse than the benchmark.}
	\resizebox*{0.725\linewidth}{!}{
		
\begin{tabular}[t]{cc>{}c|ccccccccccccc}
\toprule
\multicolumn{1}{c}{\textbf{}} & \multicolumn{1}{c}{\textbf{}} & \multicolumn{1}{c}{\textbf{}} & \multicolumn{13}{c}{\textbf{\textbf{Approach}}} \\
\cmidrule(l{0pt}r{0pt}){4-16}
\textbf{Sett.} & \textbf{p} & \textbf{N} & \rotatebox{90}{base$^{\ast}$} & \rotatebox{90}{base$^{\ast}_{\text{shr}}$} & \rotatebox{90}{base$_{\text{shr}}$} & \rotatebox{90}{ow$_{\text{var}}$} & \rotatebox{90}{ow$_{\text{cov}}$} & \rotatebox{90}{src} & \rotatebox{90}{scr$_{\text{ew}}$} & \rotatebox{90}{scr$_{\text{var}}$} & \rotatebox{90}{scr$_{\text{cov}}$} & \rotatebox{90}{occ$_{\text{wls}}$} & \rotatebox{90}{occ$_{\text{bv}}$} & \rotatebox{90}{occ$_{\text{shr}}$} & \rotatebox{90}{occ$_{\text{be}}$}\\
\midrule
1 & 4 & 50 & \textcolor{red}{1.269} & \textcolor{red}{1.132} & \textcolor{red}{1.132} & \textcolor{red}{1.003} & \textcolor{red}{1.028} & \em{0.955} & 0.962 & 0.965 & 0.982 & \textbf{0.953} & 0.969 & \textcolor{red}{1.053} & 0.959\\
 &  & 100 & \textcolor{red}{1.267} & \textcolor{red}{1.127} & \textcolor{red}{1.126} & \textcolor{red}{1.001} & \textcolor{red}{1.015} & \em{0.953} & 0.959 & 0.959 & 0.968 & \textbf{0.951} & 0.962 & \textcolor{red}{1.031} & 0.955\\
 &  & 200 & \textcolor{red}{1.265} & \textcolor{red}{1.121} & \textcolor{red}{1.121} & \textcolor{red}{1.001} & \textcolor{red}{1.008} & \em{0.953} & 0.956 & 0.957 & 0.961 & \textbf{0.952} & 0.959 & 1.000 & 0.954\\
\addlinespace[0.25em]
 & 10 & 50 & \textcolor{red}{1.358} & \textcolor{red}{1.213} & \textcolor{red}{1.213} & \textcolor{red}{1.002} & \textcolor{red}{1.050} & \em{0.979} & 0.984 & 0.986 & \textcolor{red}{1.021} & \textbf{0.978} & \textcolor{red}{1.019} & \textcolor{red}{1.165} & 0.981\\
 &  & 100 & \textcolor{red}{1.358} & \textcolor{red}{1.208} & \textcolor{red}{1.207} & \textcolor{red}{1.001} & \textcolor{red}{1.032} & \em{0.978} & 0.982 & 0.983 & \textcolor{red}{1.005} & \textbf{0.977} & \textcolor{red}{1.002} & \textcolor{red}{1.197} & 0.979\\
\bottomrule
\end{tabular}

		}
		\vspace{-1cm}
\end{table}

\begin{table}[H]
	\centering
	\scriptsize
	\setlength{\tabcolsep}{2pt}
	\def\arraystretch{1}
	\caption{AvgRelMSE for the simulation experiment (balanced panel of forecasts). Benchmark approach: ew. Bold entries identify the best performing approaches, italic entries identify the second best and red color denotes forecasts worse than the benchmark.}
	\resizebox*{0.725\linewidth}{!}{
		
\begin{tabular}[t]{cc>{}c|ccccccccccccc}
\toprule
\multicolumn{1}{c}{\textbf{}} & \multicolumn{1}{c}{\textbf{}} & \multicolumn{1}{c}{\textbf{}} & \multicolumn{13}{c}{\textbf{\textbf{Approach}}} \\
\cmidrule(l{0pt}r{0pt}){4-16}
\textbf{Sett.} & \textbf{p} & \textbf{N} & \rotatebox{90}{base$^{\ast}$} & \rotatebox{90}{base$^{\ast}_{\text{shr}}$} & \rotatebox{90}{base$_{\text{shr}}$} & \rotatebox{90}{ow$_{\text{var}}$} & \rotatebox{90}{ow$_{\text{cov}}$} & \rotatebox{90}{src} & \rotatebox{90}{scr$_{\text{ew}}$} & \rotatebox{90}{scr$_{\text{var}}$} & \rotatebox{90}{scr$_{\text{cov}}$} & \rotatebox{90}{occ$_{\text{wls}}$} & \rotatebox{90}{occ$_{\text{bv}}$} & \rotatebox{90}{occ$_{\text{shr}}$} & \rotatebox{90}{occ$_{\text{be}}$}\\
\midrule
1 & 4 & 50 & \textcolor{red}{1.611} & \textcolor{red}{1.283} & \textcolor{red}{1.283} & \textcolor{red}{1.005} & \textcolor{red}{1.058} & \em{0.912} & 0.926 & 0.930 & 0.965 & \textbf{0.908} & 0.939 & \textcolor{red}{1.108} & 0.919\\
 &  & 100 & \textcolor{red}{1.606} & \textcolor{red}{1.269} & \textcolor{red}{1.269} & \textcolor{red}{1.002} & \textcolor{red}{1.031} & \em{0.909} & 0.919 & 0.921 & 0.939 & \textbf{0.906} & 0.926 & \textcolor{red}{1.063} & 0.913\\
 &  & 200 & \textcolor{red}{1.603} & \textcolor{red}{1.257} & \textcolor{red}{1.257} & \textcolor{red}{1.001} & \textcolor{red}{1.015} & \em{0.908} & 0.914 & 0.915 & 0.924 & \textbf{0.907} & 0.919 & 0.999 & 0.910\\
\addlinespace[0.25em]
 & 10 & 50 & \textcolor{red}{1.847} & \textcolor{red}{1.474} & \textcolor{red}{1.472} & \textcolor{red}{1.004} & \textcolor{red}{1.104} & \em{0.957} & 0.968 & 0.971 & \textcolor{red}{1.043} & \textbf{0.956} & \textcolor{red}{1.039} & \textcolor{red}{1.355} & 0.962\\
 &  & 100 & \textcolor{red}{1.843} & \textcolor{red}{1.457} & \textcolor{red}{1.457} & \textcolor{red}{1.002} & \textcolor{red}{1.065} & \em{0.956} & 0.965 & 0.966 & \textcolor{red}{1.010} & \textbf{0.954} & \textcolor{red}{1.004} & \textcolor{red}{1.432} & 0.958\\
\bottomrule
\end{tabular}

		}
\end{table}

\begin{table}[H]
	\centering
	\scriptsize
	\setlength{\tabcolsep}{2pt}
	\def\arraystretch{1}
	\caption{AvgRelMAE for the simulation experiment (unbalanced panel of forecasts). Benchmark approach: ew. Bold entries identify the best performing approaches, italic entries identify the second best and red color denotes forecasts worse than the benchmark.}
	%\resizebox*{0.9\linewidth}{!}{
		
\begin{tabular}[t]{cc>{}c|ccccccccc}
\toprule
\multicolumn{1}{c}{\textbf{}} & \multicolumn{1}{c}{\textbf{}} & \multicolumn{1}{c}{\textbf{}} & \multicolumn{9}{c}{\textbf{\textbf{Approach}}} \\
\cmidrule(l{0pt}r{0pt}){4-12}
\textbf{Sett.} & \textbf{p} & \textbf{N} & \rotatebox{90}{ow$_{\text{var}}$} & \rotatebox{90}{ow$_{\text{cov}}$} & \rotatebox{90}{scr$_{\text{ew}}$} & \rotatebox{90}{scr$_{\text{var}}$} & \rotatebox{90}{scr$_{\text{cov}}$} & \rotatebox{90}{occ$_{\text{wls}}$} & \rotatebox{90}{occ$_{\text{bv}}$} & \rotatebox{90}{occ$_{\text{shr}}$} & \rotatebox{90}{occ$_{\text{be}}$}\\
\midrule
1 & 4 & 50 & \textcolor{red}{1.067} & \textcolor{red}{1.109} & 0.919 & \textcolor{red}{1.097} & \textcolor{red}{1.141} & \textbf{0.908} & 0.911 & 0.911 & \em{0.911}\\
 &  & 100 & \textcolor{red}{1.065} & \textcolor{red}{1.121} & 0.914 & \textcolor{red}{1.079} & \textcolor{red}{1.145} & \textbf{0.903} & \em{0.906} & 0.906 & 0.906\\
 &  & 200 & \textcolor{red}{1.048} & \textcolor{red}{1.112} & 0.911 & \textcolor{red}{1.044} & \textcolor{red}{1.128} & \textbf{0.902} & \em{0.904} & 0.904 & 0.905\\
\addlinespace[0.25em]
 & 10 & 50 & \textcolor{red}{1.250} & \textcolor{red}{1.325} & 0.953 & \textcolor{red}{1.312} & \textcolor{red}{1.358} & \em{\textbf{0.950}} & 0.955 & \em{\textbf{0.950}} & 0.951\\
 &  & 100 & \textcolor{red}{1.249} & \textcolor{red}{1.359} & 0.948 & \textcolor{red}{1.285} & \textcolor{red}{1.379} & \textbf{0.942} & 0.947 & \em{0.942} & 0.944\\
 &  & 200 & \textcolor{red}{1.216} & \textcolor{red}{1.382} & 0.946 & \textcolor{red}{1.224} & \textcolor{red}{1.391} & \em{0.938} & 0.943 & \textbf{0.938} & 0.940\\
\bottomrule
\end{tabular}

	%	}
\end{table}

\begin{table}[H]
	\centering
	\scriptsize
	\setlength{\tabcolsep}{2pt}
	\def\arraystretch{1}
	\caption{AvgRelMSE for the simulation experiment (unbalanced panel of forecasts). Benchmark approach: ew. Bold entries identify the best performing approaches, italic entries identify the second best and red color denotes forecasts worse than the benchmark.}
	%\resizebox*{0.9\linewidth}{!}{
		
\begin{tabular}[t]{cc>{}c|ccccccccc}
\toprule
\multicolumn{1}{c}{\textbf{}} & \multicolumn{1}{c}{\textbf{}} & \multicolumn{1}{c}{\textbf{}} & \multicolumn{9}{c}{\textbf{\textbf{Approach}}} \\
\cmidrule(l{0pt}r{0pt}){4-12}
\textbf{Sett.} & \textbf{p} & \textbf{N} & \rotatebox{90}{ow$_{\text{var}}$} & \rotatebox{90}{ow$_{\text{cov}}$} & \rotatebox{90}{scr$_{\text{ew}}$} & \rotatebox{90}{scr$_{\text{var}}$} & \rotatebox{90}{scr$_{\text{cov}}$} & \rotatebox{90}{occ$_{\text{wls}}$} & \rotatebox{90}{occ$_{\text{bv}}$} & \rotatebox{90}{occ$_{\text{shr}}$} & \rotatebox{90}{occ$_{\text{be}}$}\\
\midrule
1 & 4 & 50 & \textcolor{red}{1.142} & \textcolor{red}{1.237} & 0.841 & \textcolor{red}{1.207} & \textcolor{red}{1.305} & \textbf{0.821} & 0.826 & 0.827 & \em{0.825}\\
 &  & 100 & \textcolor{red}{1.135} & \textcolor{red}{1.259} & 0.831 & \textcolor{red}{1.160} & \textcolor{red}{1.309} & \textbf{0.811} & \em{0.816} & 0.818 & 0.816\\
 &  & 200 & \textcolor{red}{1.100} & \textcolor{red}{1.242} & 0.825 & \textcolor{red}{1.087} & \textcolor{red}{1.272} & \textbf{0.808} & \em{0.812} & 0.812 & 0.814\\
\addlinespace[0.25em]
 & 10 & 50 & \textcolor{red}{1.574} & \textcolor{red}{1.768} & 0.905 & \textcolor{red}{1.727} & \textcolor{red}{1.849} & \em{\textbf{0.901}} & 0.911 & \em{\textbf{0.901}} & 0.904\\
 &  & 100 & \textcolor{red}{1.561} & \textcolor{red}{1.851} & 0.897 & \textcolor{red}{1.649} & \textcolor{red}{1.898} & \textbf{0.885} & 0.895 & \em{0.885} & 0.888\\
 &  & 200 & \textcolor{red}{1.483} & \textcolor{red}{1.917} & 0.892 & \textcolor{red}{1.501} & \textcolor{red}{1.936} & \em{0.878} & 0.886 & \textbf{0.878} & 0.881\\
\bottomrule
\end{tabular}

	%	}
\end{table}

\subsection{Simulation with correlated forecast errors}
\begin{table}[H]
	\centering
	\scriptsize
	\setlength{\tabcolsep}{2pt}
	\def\arraystretch{1}
	\caption{AvgRelMAE for the simulation experiment (balanced panel of forecasts). Benchmark approach: ew. Bold entries identify the best performing approaches, italic entries identify the second best and red color denotes forecasts worse than the benchmark.}
	\resizebox*{0.7\linewidth}{!}{
		
\begin{tabular}[t]{cc>{}c|ccccccccccccc}
\toprule
\multicolumn{1}{c}{\textbf{}} & \multicolumn{1}{c}{\textbf{}} & \multicolumn{1}{c}{\textbf{}} & \multicolumn{13}{c}{\textbf{\textbf{Approach}}} \\
\cmidrule(l{0pt}r{0pt}){4-16}
\textbf{Sett.} & \textbf{p} & \textbf{N} & \rotatebox{90}{base$^{\ast}$} & \rotatebox{90}{base$^{\ast}_{\text{shr}}$} & \rotatebox{90}{base$_{\text{shr}}$} & \rotatebox{90}{ow$_{\text{var}}$} & \rotatebox{90}{ow$_{\text{cov}}$} & \rotatebox{90}{src} & \rotatebox{90}{scr$_{\text{ew}}$} & \rotatebox{90}{scr$_{\text{var}}$} & \rotatebox{90}{scr$_{\text{cov}}$} & \rotatebox{90}{occ$_{\text{wls}}$} & \rotatebox{90}{occ$_{\text{bv}}$} & \rotatebox{90}{occ$_{\text{shr}}$} & \rotatebox{90}{occ$_{\text{be}}$}\\
\midrule
1 & 4 & 50 & \textcolor{red}{1.308} & \textcolor{red}{1.063} & \textcolor{red}{1.058} & \textcolor{red}{1.002} & \textcolor{red}{1.027} & \em{0.921} & 0.949 & 0.951 & 0.969 & 0.945 & 0.961 & 0.976 & \textbf{0.905}\\
 &  & 100 & \textcolor{red}{1.307} & \textcolor{red}{1.038} & \textcolor{red}{1.037} & \textcolor{red}{1.000} & \textcolor{red}{1.014} & \em{0.915} & 0.946 & 0.946 & 0.955 & 0.944 & 0.955 & 0.939 & \textbf{0.894}\\
 &  & 200 & \textcolor{red}{1.305} & \textcolor{red}{1.026} & \textcolor{red}{1.025} & 1.000 & \textcolor{red}{1.007} & 0.912 & 0.943 & 0.942 & 0.946 & 0.944 & 0.951 & \em{0.908} & \textbf{0.890}\\
\addlinespace[0.25em]
 & 10 & 50 & \textcolor{red}{1.417} & \textcolor{red}{1.152} & \textcolor{red}{1.146} & \textcolor{red}{1.001} & \textcolor{red}{1.053} & \em{0.962} & 0.980 & 0.981 & \textcolor{red}{1.020} & 0.974 & \textcolor{red}{1.016} & \textcolor{red}{1.148} & \textbf{0.953}\\
 &  & 100 & \textcolor{red}{1.417} & \textcolor{red}{1.125} & \textcolor{red}{1.123} & \textcolor{red}{1.000} & \textcolor{red}{1.033} & \em{0.959} & 0.979 & 0.979 & \textcolor{red}{1.003} & 0.973 & 1.000 & \textcolor{red}{1.109} & \textbf{0.947}\\
\bottomrule
\end{tabular}

		}
\end{table}

\begin{table}[H]
	\centering
	\scriptsize
	\setlength{\tabcolsep}{2pt}
	\def\arraystretch{1}
	\caption{AvgRelMSE for the simulation experiment (balanced panel of forecasts). Benchmark approach: ew. Bold entries identify the best performing approaches, italic entries identify the second best and red color denotes forecasts worse than the benchmark.}
	%\resizebox*{0.9\linewidth}{!}{
		
\begin{tabular}[t]{cc>{}c|ccccccccccccc}
\toprule
\multicolumn{1}{c}{\textbf{}} & \multicolumn{1}{c}{\textbf{}} & \multicolumn{1}{c}{\textbf{}} & \multicolumn{13}{c}{\textbf{\textbf{Approach}}} \\
\cmidrule(l{0pt}r{0pt}){4-16}
\textbf{Sett.} & \textbf{p} & \textbf{N} & \rotatebox{90}{base$^{\ast}$} & \rotatebox{90}{base$^{\ast}_{\text{shr}}$} & \rotatebox{90}{base$_{\text{shr}}$} & \rotatebox{90}{ow$_{\text{var}}$} & \rotatebox{90}{ow$_{\text{cov}}$} & \rotatebox{90}{src} & \rotatebox{90}{scr$_{\text{ew}}$} & \rotatebox{90}{scr$_{\text{var}}$} & \rotatebox{90}{scr$_{\text{cov}}$} & \rotatebox{90}{occ$_{\text{wls}}$} & \rotatebox{90}{occ$_{\text{bv}}$} & \rotatebox{90}{occ$_{\text{shr}}$} & \rotatebox{90}{occ$_{\text{be}}$}\\
\midrule
1 & 4 & 50 & \textcolor{red}{1.712} & \textcolor{red}{1.130} & \textcolor{red}{1.120} & \textcolor{red}{1.004} & \textcolor{red}{1.056} & \em{0.847} & 0.901 & 0.905 & 0.939 & 0.893 & 0.924 & 0.954 & \textbf{0.818}\\
 &  & 100 & \textcolor{red}{1.708} & \textcolor{red}{1.078} & \textcolor{red}{1.074} & \textcolor{red}{1.001} & \textcolor{red}{1.029} & \em{0.837} & 0.895 & 0.896 & 0.913 & 0.892 & 0.912 & 0.883 & \textbf{0.801}\\
 &  & 200 & \textcolor{red}{1.702} & \textcolor{red}{1.052} & \textcolor{red}{1.049} & 1.000 & \textcolor{red}{1.013} & 0.832 & 0.888 & 0.888 & 0.895 & 0.891 & 0.904 & \em{0.825} & \textbf{0.791}\\
\addlinespace[0.25em]
 & 10 & 50 & \textcolor{red}{2.010} & \textcolor{red}{1.327} & \textcolor{red}{1.315} & \textcolor{red}{1.002} & \textcolor{red}{1.108} & \em{0.925} & 0.960 & 0.963 & \textcolor{red}{1.040} & 0.948 & \textcolor{red}{1.032} & \textcolor{red}{1.317} & \textbf{0.908}\\
 &  & 100 & \textcolor{red}{2.007} & \textcolor{red}{1.267} & \textcolor{red}{1.262} & \textcolor{red}{1.001} & \textcolor{red}{1.068} & \em{0.920} & 0.958 & 0.959 & \textcolor{red}{1.005} & 0.947 & \textcolor{red}{1.000} & \textcolor{red}{1.232} & \textbf{0.899}\\
\bottomrule
\end{tabular}

	%	}
\end{table}

\begin{table}[H]
	\centering
	\scriptsize
	\setlength{\tabcolsep}{2pt}
	\def\arraystretch{1}
	\caption{AvgRelMAE for the simulation experiment (unbalanced panel of forecasts). Benchmark approach: ew. Bold entries identify the best performing approaches, italic entries identify the second best and red color denotes forecasts worse than the benchmark.}
	%\resizebox*{0.9\linewidth}{!}{
		
\begin{tabular}[t]{cc>{}c|ccccccccc}
\toprule
\multicolumn{1}{c}{\textbf{}} & \multicolumn{1}{c}{\textbf{}} & \multicolumn{1}{c}{\textbf{}} & \multicolumn{9}{c}{\textbf{\textbf{Approach}}} \\
\cmidrule(l{0pt}r{0pt}){4-12}
\textbf{Sett.} & \textbf{p} & \textbf{N} & \rotatebox{90}{ow$_{\text{var}}$} & \rotatebox{90}{ow$_{\text{cov}}$} & \rotatebox{90}{scr$_{\text{ew}}$} & \rotatebox{90}{scr$_{\text{var}}$} & \rotatebox{90}{scr$_{\text{cov}}$} & \rotatebox{90}{occ$_{\text{wls}}$} & \rotatebox{90}{occ$_{\text{bv}}$} & \rotatebox{90}{occ$_{\text{shr}}$} & \rotatebox{90}{occ$_{\text{be}}$}\\
\midrule
1 & 4 & 50 & \textcolor{red}{1.031} & \textcolor{red}{1.069} & 0.905 & \textcolor{red}{1.054} & \textcolor{red}{1.092} & \em{0.897} & 0.899 & 0.898 & \textbf{0.892}\\
 &  & 100 & \textcolor{red}{1.024} & \textcolor{red}{1.070} & 0.901 & \textcolor{red}{1.033} & \textcolor{red}{1.087} & \em{0.893} & 0.895 & 0.893 & \textbf{0.885}\\
 &  & 200 & \textcolor{red}{1.010} & \textcolor{red}{1.062} & 0.897 & \textcolor{red}{1.003} & \textcolor{red}{1.072} & 0.890 & 0.892 & \em{0.889} & \textbf{0.880}\\
\addlinespace[0.25em]
 & 10 & 50 & \textcolor{red}{1.220} & \textcolor{red}{1.286} & 0.945 & \textcolor{red}{1.276} & \textcolor{red}{1.316} & \em{0.943} & 0.948 & \em{0.943} & \textbf{0.941}\\
 &  & 100 & \textcolor{red}{1.216} & \textcolor{red}{1.314} & 0.940 & \textcolor{red}{1.250} & \textcolor{red}{1.333} & \em{0.933} & 0.938 & 0.933 & \textbf{0.928}\\
 &  & 200 & \textcolor{red}{1.183} & \textcolor{red}{1.335} & 0.937 & \textcolor{red}{1.190} & \textcolor{red}{1.344} & 0.929 & 0.933 & \em{0.929} & \textbf{0.923}\\
\bottomrule
\end{tabular}

	%	}
\end{table}

\begin{table}[H]
	\centering
	\scriptsize
	\setlength{\tabcolsep}{2pt}
	\def\arraystretch{1}
	\caption{AvgRelMSE for the simulation experiment (unbalanced panel of forecasts). Benchmark approach: ew. Bold entries identify the best performing approaches, italic entries identify the second best and red color denotes forecasts worse than the benchmark.}
	%\resizebox*{0.9\linewidth}{!}{
		
\begin{tabular}[t]{cc>{}c|ccccccccc}
\toprule
\multicolumn{1}{c}{\textbf{}} & \multicolumn{1}{c}{\textbf{}} & \multicolumn{1}{c}{\textbf{}} & \multicolumn{9}{c}{\textbf{\textbf{Approach}}} \\
\cmidrule(l{0pt}r{0pt}){4-12}
\textbf{Sett.} & \textbf{p} & \textbf{N} & \rotatebox{90}{ow$_{\text{var}}$} & \rotatebox{90}{ow$_{\text{cov}}$} & \rotatebox{90}{scr$_{\text{ew}}$} & \rotatebox{90}{scr$_{\text{var}}$} & \rotatebox{90}{scr$_{\text{cov}}$} & \rotatebox{90}{occ$_{\text{wls}}$} & \rotatebox{90}{occ$_{\text{bv}}$} & \rotatebox{90}{occ$_{\text{shr}}$} & \rotatebox{90}{occ$_{\text{be}}$}\\
\midrule
1 & 4 & 50 & \textcolor{red}{1.063} & \textcolor{red}{1.145} & 0.815 & \textcolor{red}{1.111} & \textcolor{red}{1.193} & \em{0.800} & 0.804 & 0.804 & \textbf{0.792}\\
 &  & 100 & \textcolor{red}{1.046} & \textcolor{red}{1.145} & 0.807 & \textcolor{red}{1.062} & \textcolor{red}{1.178} & \em{0.792} & 0.796 & 0.795 & \textbf{0.779}\\
 &  & 200 & \textcolor{red}{1.020} & \textcolor{red}{1.129} & 0.801 & \textcolor{red}{1.002} & \textcolor{red}{1.148} & \em{0.789} & 0.792 & 0.789 & \textbf{0.771}\\
\addlinespace[0.25em]
 & 10 & 50 & \textcolor{red}{1.494} & \textcolor{red}{1.658} & 0.889 & \textcolor{red}{1.631} & \textcolor{red}{1.731} & 0.887 & 0.896 & \em{0.887} & \textbf{0.883}\\
 &  & 100 & \textcolor{red}{1.480} & \textcolor{red}{1.730} & 0.881 & \textcolor{red}{1.563} & \textcolor{red}{1.774} & \em{0.870} & 0.880 & 0.870 & \textbf{0.861}\\
 &  & 200 & \textcolor{red}{1.404} & \textcolor{red}{1.789} & 0.875 & \textcolor{red}{1.416} & \textcolor{red}{1.808} & 0.860 & 0.867 & \em{0.860} & \textbf{0.849}\\
\bottomrule
\end{tabular}

	%	}
\end{table}				% Appendix E

\clearpage
\section{Forecasting experiment on the Australian electricity generation dataset}
\label{app:fexp}
\begin{figure}[!htb]
	\centering
	\includegraphics[width = 0.9\linewidth]{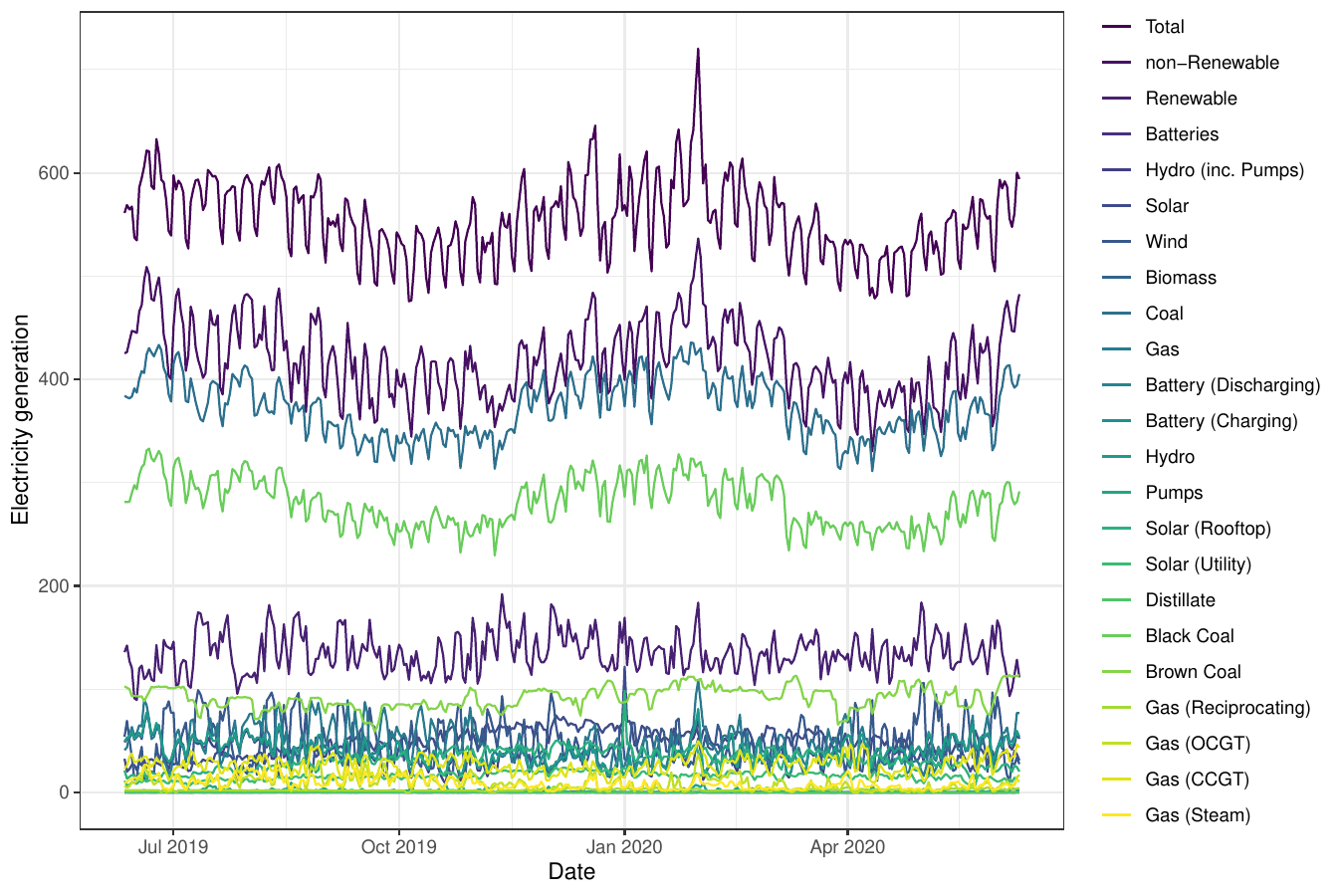}
	\caption{Australian daily electricity generation time series.}
	\label{fig:plot_ts_eg}
\end{figure}

\begin{figure}[!htb]
	\centering
	\includegraphics[width = 0.9\linewidth]{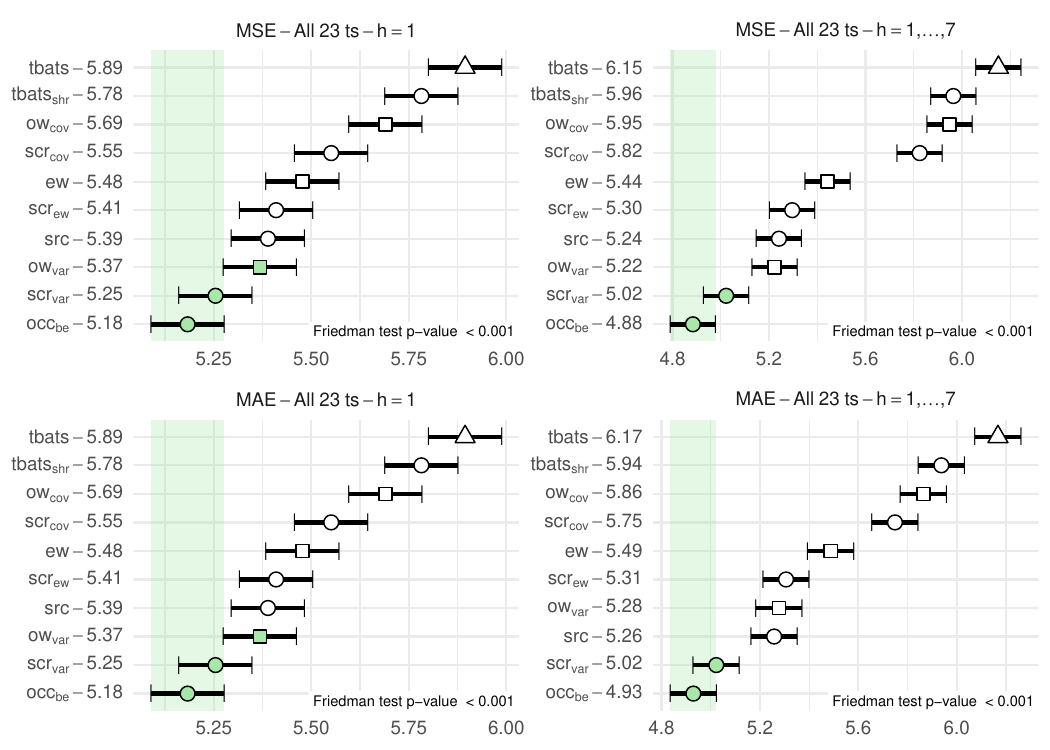}
	\caption{MCB Nemenyi test for the electricity generation dataset using the MSE (first row) and MAE (second row) at different forecast horizon ($h = 1$ for the first column and $h = 1,...,7$ for the second column). In each panel, the Friedman test p-value is reported in the lower-right corner. The mean rank of each approach is shown to the right of its name. Statistically significant differences in performance are indicated if the intervals of two forecasting approaches do not overlap. Thus, approaches that do not overlap with the green interval are considered significantly worse than the best, and vice versa.}\label{fig:energy_mcb_oa}
\end{figure}

\begin{table}[!htb]
	\centering
	\scriptsize
	\def\arraystretch{0.9}
	\caption{AvgRelMAE of daily forecasts for the Australian electricity generation dataset for all, upper and bottom time series. Benchmark approach: ew. Bold entries identify the best performing approaches, italic entries identify the second best and red color denotes forecasts worse than the benchmark.}\label{tab:oa_Energy_accuracy_mae}
	
\begin{tabular}[t]{>{}l|cccccccc}
\toprule
\multicolumn{1}{c}{\textbf{}} & \multicolumn{8}{c}{\textbf{Forecast horizon}} \\
\cmidrule(l{0pt}r{0pt}){2-9}
\multicolumn{1}{l|}{\textbf{Approach}} & 1 & 2 & 3 & 4 & 5 & 6 & 7 & 1:7\\
\midrule
\addlinespace[0.3em]
\multicolumn{9}{c}{\textbf{All 23 time series}}\\
\addlinespace[0em]
\multicolumn{9}{l}{\textit{Base (incoherent forecasts)}}\\
stlf & \textcolor{red}{1.0697} & \textcolor{red}{1.0737} & \textcolor{red}{1.0996} & \textcolor{red}{1.1029} & \textcolor{red}{1.0806} & \textcolor{red}{1.0672} & \textcolor{red}{1.0801} & \textcolor{red}{1.0830}\\
arima & \textcolor{red}{1.0561} & \textcolor{red}{1.0650} & \textcolor{red}{1.0501} & \textcolor{red}{1.0426} & \textcolor{red}{1.0459} & \textcolor{red}{1.0414} & \textcolor{red}{1.0339} & \textcolor{red}{1.0471}\\
tbats & \textcolor{red}{1.0447} & \textcolor{red}{1.0515} & \textcolor{red}{1.0348} & \textcolor{red}{1.0266} & \textcolor{red}{1.0305} & \textcolor{red}{1.0288} & \textcolor{red}{1.0201} & \textcolor{red}{1.0331}\\
\addlinespace[0em]
\multicolumn{9}{l}{\textit{Single-expert reconciliation}}\\
stlf$_{\text{shr}}$ & \textcolor{red}{1.0618} & \textcolor{red}{1.0696} & \textcolor{red}{1.0970} & \textcolor{red}{1.1001} & \textcolor{red}{1.0776} & \textcolor{red}{1.0625} & \textcolor{red}{1.0758} & \textcolor{red}{1.0788}\\
arima$_{\text{shr}}$ & \textcolor{red}{1.0512} & \textcolor{red}{1.0546} & \textcolor{red}{1.0334} & \textcolor{red}{1.0241} & \textcolor{red}{1.0232} & \textcolor{red}{1.0237} & \textcolor{red}{1.0182} & \textcolor{red}{1.0311}\\
tbats$_{\text{shr}}$ & \textcolor{red}{1.0320} & \textcolor{red}{1.0413} & \textcolor{red}{1.0231} & \textcolor{red}{1.0134} & \textcolor{red}{1.0212} & \textcolor{red}{1.0208} & \textcolor{red}{1.0188} & \textcolor{red}{1.0235}\\
\addlinespace[0em]
\multicolumn{9}{l}{\textit{Combination (incoherent forecasts)}}\\
ow$_{\text{var}}$ & 0.9927 & 0.9921 & 0.9982 & 0.9983 & 0.9967 & 0.9967 & 0.9990 & 0.9965\\
ow$_{\text{cov}}$ & \textcolor{red}{1.0216} & \textcolor{red}{1.0208} & \textcolor{red}{1.0390} & \textcolor{red}{1.0423} & \textcolor{red}{1.0307} & \textcolor{red}{1.0250} & \textcolor{red}{1.0325} & \textcolor{red}{1.0309}\\
\addlinespace[0em]
\multicolumn{9}{l}{\textit{Coherent combination}}\\
src & 0.9939 & 0.9941 & 0.9919 & 0.9895 & 0.9887 & \em{0.9908} & \em{0.9933} & 0.9915\\
scr$_{\text{ew}}$ & 0.9952 & 0.9959 & 0.9911 & 0.9908 & 0.9908 & 0.9932 & 0.9961 & 0.9930\\
scr$_{\text{var}}$ & \em{0.9819} & \em{0.9803} & \em{0.9869} & \em{0.9895} & \em{0.9887} & 0.9913 & 0.9972 & \em{0.9882}\\
scr$_{\text{cov}}$ & \textcolor{red}{1.0081} & \textcolor{red}{1.0081} & \textcolor{red}{1.0270} & \textcolor{red}{1.0327} & \textcolor{red}{1.0245} & \textcolor{red}{1.0197} & \textcolor{red}{1.0250} & \textcolor{red}{1.0215}\\
occ$_{\text{bv}}$ & \textcolor{red}{1.0157} & \textcolor{red}{1.0162} & \textcolor{red}{1.0344} & \textcolor{red}{1.0415} & \textcolor{red}{1.0289} & \textcolor{red}{1.0229} & \textcolor{red}{1.0283} & \textcolor{red}{1.0275}\\
occ$_{\text{shr}}$ & \textcolor{red}{1.0061} & \textcolor{red}{1.0111} & \textcolor{red}{1.0262} & \textcolor{red}{1.0325} & \textcolor{red}{1.0342} & \textcolor{red}{1.0335} & \textcolor{red}{1.0390} & \textcolor{red}{1.0269}\\
occ$_{\text{wls}}$ & 0.9885 & 0.9886 & 0.9956 & 0.9962 & 0.9945 & 0.9945 & 0.9975 & 0.9938\\
occ$_{\text{be}}$ & \textbf{0.9779} & \textbf{0.9745} & \textbf{0.9843} & \textbf{0.9852} & \textbf{0.9851} & \textbf{0.9880} & \textbf{0.9926} & \textbf{0.9843}\\
\midrule
\addlinespace[0.3em]
\multicolumn{9}{c}{\textbf{8 upper time series}}\\
\addlinespace[0em]
\multicolumn{9}{l}{\textit{Base (incoherent forecasts)}}\\
stlf & \textcolor{red}{1.0541} & \textcolor{red}{1.0736} & \textcolor{red}{1.1198} & \textcolor{red}{1.1272} & \textcolor{red}{1.1003} & \textcolor{red}{1.0785} & \textcolor{red}{1.0983} & \textcolor{red}{1.0953}\\
arima & \textcolor{red}{1.0900} & \textcolor{red}{1.0906} & \textcolor{red}{1.0664} & \textcolor{red}{1.0526} & \textcolor{red}{1.0568} & \textcolor{red}{1.0558} & \textcolor{red}{1.0439} & \textcolor{red}{1.0637}\\
tbats & \textcolor{red}{1.0662} & \textcolor{red}{1.0660} & \textcolor{red}{1.0423} & \textcolor{red}{1.0277} & \textcolor{red}{1.0366} & \textcolor{red}{1.0348} & \textcolor{red}{1.0219} & \textcolor{red}{1.0407}\\
\addlinespace[0em]
\multicolumn{9}{l}{\textit{Single-expert reconciliation}}\\
stlf$_{\text{shr}}$ & \textcolor{red}{1.0443} & \textcolor{red}{1.0654} & \textcolor{red}{1.1126} & \textcolor{red}{1.1224} & \textcolor{red}{1.0967} & \textcolor{red}{1.0753} & \textcolor{red}{1.0933} & \textcolor{red}{1.0893}\\
arima$_{\text{shr}}$ & \textcolor{red}{1.0684} & \textcolor{red}{1.0694} & \textcolor{red}{1.0448} & \textcolor{red}{1.0324} & \textcolor{red}{1.0297} & \textcolor{red}{1.0300} & \textcolor{red}{1.0194} & \textcolor{red}{1.0401}\\
tbats$_{\text{shr}}$ & \textcolor{red}{1.0379} & \textcolor{red}{1.0468} & \textcolor{red}{1.0316} & \textcolor{red}{1.0174} & \textcolor{red}{1.0299} & \textcolor{red}{1.0228} & \textcolor{red}{1.0177} & \textcolor{red}{1.0283}\\
\addlinespace[0em]
\multicolumn{9}{l}{\textit{Combination (incoherent forecasts)}}\\
ow$_{\text{var}}$ & 0.9856 & 0.9864 & 0.9982 & 0.9988 & 0.9955 & 0.9943 & 0.9989 & 0.9945\\
ow$_{\text{cov}}$ & \textcolor{red}{1.0129} & \textcolor{red}{1.0191} & \textcolor{red}{1.0529} & \textcolor{red}{1.0588} & \textcolor{red}{1.0422} & \textcolor{red}{1.0292} & \textcolor{red}{1.0434} & \textcolor{red}{1.0385}\\
\addlinespace[0em]
\multicolumn{9}{l}{\textit{Coherent combination}}\\
src & 0.9863 & 0.9911 & 0.9928 & 0.9918 & 0.9910 & 0.9891 & \textbf{0.9912} & 0.9904\\
scr$_{\text{ew}}$ & 0.9885 & 0.9923 & 0.9905 & 0.9913 & 0.9922 & 0.9909 & 0.9932 & 0.9912\\
scr$_{\text{var}}$ & \em{0.9663} & \em{0.9662} & \em{0.9841} & \em{0.9889} & \em{0.9894} & \em{0.9882} & 0.9986 & \em{0.9840}\\
scr$_{\text{cov}}$ & 0.9912 & 0.9953 & \textcolor{red}{1.0316} & \textcolor{red}{1.0433} & \textcolor{red}{1.0334} & \textcolor{red}{1.0243} & \textcolor{red}{1.0345} & \textcolor{red}{1.0237}\\
occ$_{\text{bv}}$ & \textcolor{red}{1.0026} & \textcolor{red}{1.0058} & \textcolor{red}{1.0410} & \textcolor{red}{1.0521} & \textcolor{red}{1.0371} & \textcolor{red}{1.0286} & \textcolor{red}{1.0372} & \textcolor{red}{1.0306}\\
occ$_{\text{shr}}$ & 0.9810 & 0.9942 & \textcolor{red}{1.0256} & \textcolor{red}{1.0397} & \textcolor{red}{1.0374} & \textcolor{red}{1.0353} & \textcolor{red}{1.0454} & \textcolor{red}{1.0246}\\
occ$_{\text{wls}}$ & 0.9774 & 0.9800 & 0.9967 & 0.9969 & 0.9959 & 0.9928 & 0.9959 & 0.9913\\
occ$_{\text{be}}$ & \textbf{0.9593} & \textbf{0.9610} & \textbf{0.9823} & \textbf{0.9849} & \textbf{0.9838} & \textbf{0.9842} & \em{0.9917} & \textbf{0.9791}\\
\midrule
\addlinespace[0.3em]
\multicolumn{9}{c}{\textbf{15 bottom time series}}\\
\addlinespace[0em]
\multicolumn{9}{l}{\textit{Base (incoherent forecasts)}}\\
stlf & \textcolor{red}{1.0781} & \textcolor{red}{1.0737} & \textcolor{red}{1.0890} & \textcolor{red}{1.0902} & \textcolor{red}{1.0703} & \textcolor{red}{1.0612} & \textcolor{red}{1.0704} & \textcolor{red}{1.0764}\\
arima & \textcolor{red}{1.0384} & \textcolor{red}{1.0517} & \textcolor{red}{1.0415} & \textcolor{red}{1.0372} & \textcolor{red}{1.0401} & \textcolor{red}{1.0337} & \textcolor{red}{1.0286} & \textcolor{red}{1.0384}\\
tbats & \textcolor{red}{1.0334} & \textcolor{red}{1.0438} & \textcolor{red}{1.0309} & \textcolor{red}{1.0260} & \textcolor{red}{1.0273} & \textcolor{red}{1.0256} & \textcolor{red}{1.0192} & \textcolor{red}{1.0290}\\
\addlinespace[0em]
\multicolumn{9}{l}{\textit{Single-expert reconciliation}}\\
stlf$_{\text{shr}}$ & \textcolor{red}{1.0712} & \textcolor{red}{1.0719} & \textcolor{red}{1.0888} & \textcolor{red}{1.0884} & \textcolor{red}{1.0676} & \textcolor{red}{1.0557} & \textcolor{red}{1.0666} & \textcolor{red}{1.0732}\\
arima$_{\text{shr}}$ & \textcolor{red}{1.0421} & \textcolor{red}{1.0468} & \textcolor{red}{1.0273} & \textcolor{red}{1.0197} & \textcolor{red}{1.0197} & \textcolor{red}{1.0203} & \textcolor{red}{1.0175} & \textcolor{red}{1.0263}\\
tbats$_{\text{shr}}$ & \textcolor{red}{1.0289} & \textcolor{red}{1.0384} & \textcolor{red}{1.0186} & \textcolor{red}{1.0112} & \textcolor{red}{1.0166} & \textcolor{red}{1.0197} & \textcolor{red}{1.0194} & \textcolor{red}{1.0210}\\
\addlinespace[0em]
\multicolumn{9}{l}{\textit{Combination (incoherent forecasts)}}\\
ow$_{\text{var}}$ & 0.9965 & 0.9951 & 0.9982 & 0.9980 & 0.9973 & 0.9979 & 0.9990 & 0.9975\\
ow$_{\text{cov}}$ & \textcolor{red}{1.0263} & \textcolor{red}{1.0218} & \textcolor{red}{1.0316} & \textcolor{red}{1.0336} & \textcolor{red}{1.0246} & \textcolor{red}{1.0228} & \textcolor{red}{1.0266} & \textcolor{red}{1.0269}\\
\addlinespace[0em]
\multicolumn{9}{l}{\textit{Coherent combination}}\\
src & 0.9979 & 0.9957 & 0.9915 & \em{0.9883} & \em{0.9875} & \em{0.9917} & \em{0.9944} & 0.9920\\
scr$_{\text{ew}}$ & 0.9988 & 0.9978 & 0.9914 & 0.9905 & 0.9900 & 0.9944 & 0.9976 & 0.9939\\
scr$_{\text{var}}$ & \em{0.9904} & \em{0.9878} & \em{0.9884} & 0.9898 & 0.9883 & 0.9929 & 0.9965 & \em{0.9905}\\
scr$_{\text{cov}}$ & \textcolor{red}{1.0172} & \textcolor{red}{1.0150} & \textcolor{red}{1.0245} & \textcolor{red}{1.0271} & \textcolor{red}{1.0198} & \textcolor{red}{1.0172} & \textcolor{red}{1.0201} & \textcolor{red}{1.0204}\\
occ$_{\text{bv}}$ & \textcolor{red}{1.0228} & \textcolor{red}{1.0218} & \textcolor{red}{1.0308} & \textcolor{red}{1.0358} & \textcolor{red}{1.0246} & \textcolor{red}{1.0198} & \textcolor{red}{1.0236} & \textcolor{red}{1.0258}\\
occ$_{\text{shr}}$ & \textcolor{red}{1.0197} & \textcolor{red}{1.0203} & \textcolor{red}{1.0266} & \textcolor{red}{1.0288} & \textcolor{red}{1.0325} & \textcolor{red}{1.0325} & \textcolor{red}{1.0356} & \textcolor{red}{1.0282}\\
occ$_{\text{wls}}$ & 0.9945 & 0.9933 & 0.9950 & 0.9958 & 0.9937 & 0.9955 & 0.9984 & 0.9952\\
occ$_{\text{be}}$ & \textbf{0.9879} & \textbf{0.9819} & \textbf{0.9853} & \textbf{0.9854} & \textbf{0.9859} & \textbf{0.9901} & \textbf{0.9931} & \textbf{0.9870}\\
\bottomrule
\end{tabular}

\end{table}

\begin{table}[!htb]
	\centering
	\scriptsize
	\def\arraystretch{0.9}
	\caption{AvgRelMSE of daily forecasts for the Australian electricity generation dataset for all, upper and bottom time series. Benchmark approach: ew. Bold entries identify the best performing approaches, italic entries identify the second best and red color denotes forecasts worse than the benchmark.}\label{tab:oa_Energy_accuracy_mse}
	
\begin{tabular}[t]{>{}l|cccccccc}
\toprule
\multicolumn{1}{c}{\textbf{}} & \multicolumn{8}{c}{\textbf{Forecast horizon}} \\
\cmidrule(l{0pt}r{0pt}){2-9}
\multicolumn{1}{l|}{\textbf{Approach}} & 1 & 2 & 3 & 4 & 5 & 6 & 7 & 1:7\\
\midrule
\addlinespace[0.3em]
\multicolumn{9}{c}{\textbf{All 23 time series}}\\
\addlinespace[0em]
\multicolumn{9}{l}{\textit{Base (incoherent forecasts)}}\\
stlf & \textcolor{red}{1.1205} & \textcolor{red}{1.1629} & \textcolor{red}{1.2351} & \textcolor{red}{1.2494} & \textcolor{red}{1.2252} & \textcolor{red}{1.2151} & \textcolor{red}{1.2280} & \textcolor{red}{1.2146}\\
arima & \textcolor{red}{1.1106} & \textcolor{red}{1.1158} & \textcolor{red}{1.0842} & \textcolor{red}{1.0701} & \textcolor{red}{1.0741} & \textcolor{red}{1.0726} & \textcolor{red}{1.0615} & \textcolor{red}{1.0808}\\
tbats & \textcolor{red}{1.0796} & \textcolor{red}{1.0780} & \textcolor{red}{1.0445} & \textcolor{red}{1.0270} & \textcolor{red}{1.0322} & \textcolor{red}{1.0288} & \textcolor{red}{1.0142} & \textcolor{red}{1.0393}\\
\addlinespace[0em]
\multicolumn{9}{l}{\textit{Single-expert reconciliation}}\\
stlf$_{\text{shr}}$ & \textcolor{red}{1.1022} & \textcolor{red}{1.1512} & \textcolor{red}{1.2258} & \textcolor{red}{1.2418} & \textcolor{red}{1.2183} & \textcolor{red}{1.2051} & \textcolor{red}{1.2204} & \textcolor{red}{1.2045}\\
arima$_{\text{shr}}$ & \textcolor{red}{1.0953} & \textcolor{red}{1.0904} & \textcolor{red}{1.0592} & \textcolor{red}{1.0398} & \textcolor{red}{1.0395} & \textcolor{red}{1.0433} & \textcolor{red}{1.0370} & \textcolor{red}{1.0532}\\
tbats$_{\text{shr}}$ & \textcolor{red}{1.0478} & \textcolor{red}{1.0577} & \textcolor{red}{1.0304} & \textcolor{red}{1.0108} & \textcolor{red}{1.0219} & \textcolor{red}{1.0213} & \textcolor{red}{1.0116} & \textcolor{red}{1.0257}\\
\addlinespace[0em]
\multicolumn{9}{l}{\textit{Combination (incoherent forecasts)}}\\
ow$_{\text{var}}$ & 0.9840 & 0.9881 & 0.9995 & \textcolor{red}{1.0032} & \textcolor{red}{1.0020} & \textcolor{red}{1.0028} & \textcolor{red}{1.0054} & 0.9995\\
ow$_{\text{cov}}$ & \textcolor{red}{1.0279} & \textcolor{red}{1.0494} & \textcolor{red}{1.0972} & \textcolor{red}{1.1103} & \textcolor{red}{1.1009} & \textcolor{red}{1.0993} & \textcolor{red}{1.1055} & \textcolor{red}{1.0908}\\
\addlinespace[0em]
\multicolumn{9}{l}{\textit{Coherent combination}}\\
src & 0.9827 & 0.9855 & 0.9863 & \em{0.9833} & \textbf{0.9852} & \textbf{0.9873} & \textbf{0.9911} & \em{0.9859}\\
scr$_{\text{ew}}$ & 0.9875 & 0.9898 & 0.9859 & 0.9859 & \em{0.9885} & \em{0.9905} & \em{0.9962} & 0.9890\\
scr$_{\text{var}}$ & \em{0.9586} & \em{0.9683} & \em{0.9838} & 0.9942 & 0.9982 & \textcolor{red}{1.0017} & \textcolor{red}{1.0114} & 0.9910\\
scr$_{\text{cov}}$ & \textcolor{red}{1.0026} & \textcolor{red}{1.0287} & \textcolor{red}{1.0795} & \textcolor{red}{1.0972} & \textcolor{red}{1.0942} & \textcolor{red}{1.0913} & \textcolor{red}{1.0981} & \textcolor{red}{1.0773}\\
occ$_{\text{bv}}$ & \textcolor{red}{1.0174} & \textcolor{red}{1.0420} & \textcolor{red}{1.0903} & \textcolor{red}{1.1056} & \textcolor{red}{1.0953} & \textcolor{red}{1.0915} & \textcolor{red}{1.0983} & \textcolor{red}{1.0834}\\
occ$_{\text{shr}}$ & 0.9914 & \textcolor{red}{1.0201} & \textcolor{red}{1.0626} & \textcolor{red}{1.0813} & \textcolor{red}{1.0930} & \textcolor{red}{1.0972} & \textcolor{red}{1.1090} & \textcolor{red}{1.0713}\\
occ$_{\text{wls}}$ & 0.9752 & 0.9810 & 0.9949 & \textcolor{red}{1.0012} & \textcolor{red}{1.0009} & \textcolor{red}{1.0015} & \textcolor{red}{1.0051} & 0.9960\\
occ$_{\text{be}}$ & \textbf{0.9481} & \textbf{0.9560} & \textbf{0.9754} & \textbf{0.9831} & 0.9891 & 0.9939 & 0.9993 & \textbf{0.9808}\\
\midrule
\addlinespace[0.3em]
\multicolumn{9}{c}{\textbf{8 upper time series}}\\
\addlinespace[0em]
\multicolumn{9}{l}{\textit{Base (incoherent forecasts)}}\\
stlf & \textcolor{red}{1.0772} & \textcolor{red}{1.1468} & \textcolor{red}{1.2474} & \textcolor{red}{1.2722} & \textcolor{red}{1.2344} & \textcolor{red}{1.2084} & \textcolor{red}{1.2267} & \textcolor{red}{1.2148}\\
arima & \textcolor{red}{1.1820} & \textcolor{red}{1.1781} & \textcolor{red}{1.1254} & \textcolor{red}{1.0947} & \textcolor{red}{1.0934} & \textcolor{red}{1.0978} & \textcolor{red}{1.0783} & \textcolor{red}{1.1143}\\
tbats & \textcolor{red}{1.1420} & \textcolor{red}{1.1319} & \textcolor{red}{1.0842} & \textcolor{red}{1.0533} & \textcolor{red}{1.0613} & \textcolor{red}{1.0573} & \textcolor{red}{1.0389} & \textcolor{red}{1.0744}\\
\addlinespace[0em]
\multicolumn{9}{l}{\textit{Single-expert reconciliation}}\\
stlf$_{\text{shr}}$ & \textcolor{red}{1.0549} & \textcolor{red}{1.1254} & \textcolor{red}{1.2294} & \textcolor{red}{1.2580} & \textcolor{red}{1.2232} & \textcolor{red}{1.1962} & \textcolor{red}{1.2165} & \textcolor{red}{1.1989}\\
arima$_{\text{shr}}$ & \textcolor{red}{1.1473} & \textcolor{red}{1.1357} & \textcolor{red}{1.0925} & \textcolor{red}{1.0600} & \textcolor{red}{1.0507} & \textcolor{red}{1.0535} & \textcolor{red}{1.0415} & \textcolor{red}{1.0749}\\
tbats$_{\text{shr}}$ & \textcolor{red}{1.0844} & \textcolor{red}{1.0924} & \textcolor{red}{1.0643} & \textcolor{red}{1.0367} & \textcolor{red}{1.0499} & \textcolor{red}{1.0417} & \textcolor{red}{1.0300} & \textcolor{red}{1.0529}\\
\addlinespace[0em]
\multicolumn{9}{l}{\textit{Combination (incoherent forecasts)}}\\
ow$_{\text{var}}$ & 0.9668 & 0.9766 & 0.9981 & \textcolor{red}{1.0065} & \textcolor{red}{1.0037} & \textcolor{red}{1.0032} & \textcolor{red}{1.0079} & 0.9981\\
ow$_{\text{cov}}$ & \textcolor{red}{1.0004} & \textcolor{red}{1.0377} & \textcolor{red}{1.1091} & \textcolor{red}{1.1301} & \textcolor{red}{1.1107} & \textcolor{red}{1.0957} & \textcolor{red}{1.1062} & \textcolor{red}{1.0939}\\
\addlinespace[0em]
\multicolumn{9}{l}{\textit{Coherent combination}}\\
src & 0.9730 & 0.9785 & 0.9864 & \em{0.9855} & \em{0.9862} & \textbf{0.9852} & \textbf{0.9906} & \em{0.9841}\\
scr$_{\text{ew}}$ & 0.9800 & 0.9839 & 0.9839 & 0.9871 & 0.9898 & 0.9890 & 0.9975 & 0.9874\\
scr$_{\text{var}}$ & \em{0.9292} & \em{0.9466} & \em{0.9769} & 0.9978 & \textcolor{red}{1.0026} & \textcolor{red}{1.0033} & \textcolor{red}{1.0195} & 0.9875\\
scr$_{\text{cov}}$ & 0.9589 & 0.9974 & \textcolor{red}{1.0738} & \textcolor{red}{1.1029} & \textcolor{red}{1.0962} & \textcolor{red}{1.0825} & \textcolor{red}{1.0936} & \textcolor{red}{1.0679}\\
occ$_{\text{bv}}$ & 0.9817 & \textcolor{red}{1.0163} & \textcolor{red}{1.0879} & \textcolor{red}{1.1151} & \textcolor{red}{1.1002} & \textcolor{red}{1.0878} & \textcolor{red}{1.0973} & \textcolor{red}{1.0783}\\
occ$_{\text{shr}}$ & 0.9385 & 0.9883 & \textcolor{red}{1.0541} & \textcolor{red}{1.0837} & \textcolor{red}{1.0892} & \textcolor{red}{1.0861} & \textcolor{red}{1.1064} & \textcolor{red}{1.0590}\\
occ$_{\text{wls}}$ & 0.9531 & 0.9618 & 0.9908 & \textcolor{red}{1.0042} & \textcolor{red}{1.0035} & \textcolor{red}{1.0015} & \textcolor{red}{1.0063} & 0.9919\\
occ$_{\text{be}}$ & \textbf{0.9117} & \textbf{0.9292} & \textbf{0.9664} & \textbf{0.9818} & \textbf{0.9857} & \em{0.9876} & \em{0.9971} & \textbf{0.9710}\\
\midrule
\addlinespace[0.3em]
\multicolumn{9}{c}{\textbf{15 bottom time series}}\\
\addlinespace[0em]
\multicolumn{9}{l}{\textit{Base (incoherent forecasts)}}\\
stlf & \textcolor{red}{1.1443} & \textcolor{red}{1.1716} & \textcolor{red}{1.2286} & \textcolor{red}{1.2375} & \textcolor{red}{1.2202} & \textcolor{red}{1.2186} & \textcolor{red}{1.2287} & \textcolor{red}{1.2145}\\
arima & \textcolor{red}{1.0744} & \textcolor{red}{1.0840} & \textcolor{red}{1.0628} & \textcolor{red}{1.0572} & \textcolor{red}{1.0639} & \textcolor{red}{1.0593} & \textcolor{red}{1.0526} & \textcolor{red}{1.0634}\\
tbats & \textcolor{red}{1.0477} & \textcolor{red}{1.0503} & \textcolor{red}{1.0240} & \textcolor{red}{1.0132} & \textcolor{red}{1.0170} & \textcolor{red}{1.0140} & \textcolor{red}{1.0013} & \textcolor{red}{1.0211}\\
\addlinespace[0em]
\multicolumn{9}{l}{\textit{Single-expert reconciliation}}\\
stlf$_{\text{shr}}$ & \textcolor{red}{1.1283} & \textcolor{red}{1.1652} & \textcolor{red}{1.2239} & \textcolor{red}{1.2332} & \textcolor{red}{1.2157} & \textcolor{red}{1.2099} & \textcolor{red}{1.2224} & \textcolor{red}{1.2075}\\
arima$_{\text{shr}}$ & \textcolor{red}{1.0686} & \textcolor{red}{1.0670} & \textcolor{red}{1.0418} & \textcolor{red}{1.0291} & \textcolor{red}{1.0336} & \textcolor{red}{1.0379} & \textcolor{red}{1.0345} & \textcolor{red}{1.0417}\\
tbats$_{\text{shr}}$ & \textcolor{red}{1.0287} & \textcolor{red}{1.0397} & \textcolor{red}{1.0127} & 0.9973 & \textcolor{red}{1.0073} & \textcolor{red}{1.0106} & \textcolor{red}{1.0020} & \textcolor{red}{1.0115}\\
\addlinespace[0em]
\multicolumn{9}{l}{\textit{Combination (incoherent forecasts)}}\\
ow$_{\text{var}}$ & 0.9934 & 0.9943 & \textcolor{red}{1.0003} & \textcolor{red}{1.0015} & \textcolor{red}{1.0011} & \textcolor{red}{1.0026} & \textcolor{red}{1.0041} & \textcolor{red}{1.0003}\\
ow$_{\text{cov}}$ & \textcolor{red}{1.0429} & \textcolor{red}{1.0556} & \textcolor{red}{1.0910} & \textcolor{red}{1.0998} & \textcolor{red}{1.0957} & \textcolor{red}{1.1013} & \textcolor{red}{1.1051} & \textcolor{red}{1.0892}\\
\addlinespace[0em]
\multicolumn{9}{l}{\textit{Coherent combination}}\\
src & 0.9880 & 0.9892 & \em{0.9863} & \textbf{0.9822} & \textbf{0.9847} & \textbf{0.9884} & \textbf{0.9913} & \em{0.9868}\\
scr$_{\text{ew}}$ & 0.9916 & 0.9930 & 0.9870 & 0.9853 & \em{0.9878} & \em{0.9913} & \em{0.9955} & 0.9898\\
scr$_{\text{var}}$ & \em{0.9747} & \em{0.9801} & 0.9875 & 0.9923 & 0.9959 & \textcolor{red}{1.0009} & \textcolor{red}{1.0072} & 0.9928\\
scr$_{\text{cov}}$ & \textcolor{red}{1.0268} & \textcolor{red}{1.0457} & \textcolor{red}{1.0826} & \textcolor{red}{1.0942} & \textcolor{red}{1.0931} & \textcolor{red}{1.0961} & \textcolor{red}{1.1005} & \textcolor{red}{1.0824}\\
occ$_{\text{bv}}$ & \textcolor{red}{1.0370} & \textcolor{red}{1.0559} & \textcolor{red}{1.0916} & \textcolor{red}{1.1006} & \textcolor{red}{1.0927} & \textcolor{red}{1.0935} & \textcolor{red}{1.0988} & \textcolor{red}{1.0862}\\
occ$_{\text{shr}}$ & \textcolor{red}{1.0207} & \textcolor{red}{1.0375} & \textcolor{red}{1.0672} & \textcolor{red}{1.0801} & \textcolor{red}{1.0951} & \textcolor{red}{1.1032} & \textcolor{red}{1.1103} & \textcolor{red}{1.0779}\\
occ$_{\text{wls}}$ & 0.9872 & 0.9914 & 0.9971 & 0.9996 & 0.9994 & \textcolor{red}{1.0016} & \textcolor{red}{1.0045} & 0.9982\\
occ$_{\text{be}}$ & \textbf{0.9682} & \textbf{0.9706} & \textbf{0.9802} & \em{0.9837} & 0.9909 & 0.9973 & \textcolor{red}{1.0005} & \textbf{0.9861}\\
\bottomrule
\end{tabular}

\end{table}

\begin{figure}[!htb]
	\centering
	\includegraphics[width = \linewidth]{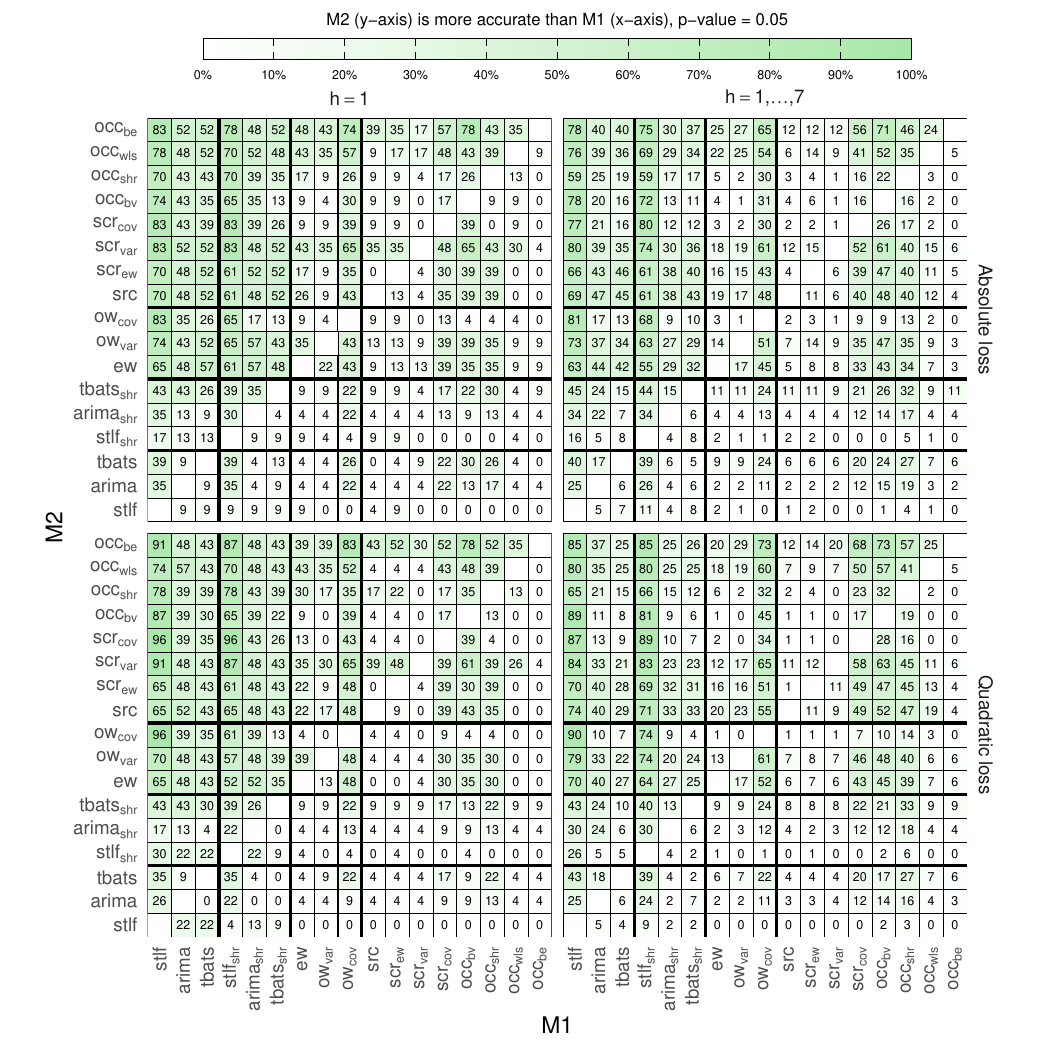}
	\caption{Pairwise DM-test results for the Australian electricity generation dataset, evaluated using absolute loss (top panels) and quadratic loss (bottom panels) across different forecast horizons. The left panel corresponds to forecast horizon $h = 1$, while the right panel is for $h = 1, \dots, 7$. Each cell reports the percentage of series for which the $p$-value of the DM-test is below $0.05$.}\label{fig:oa_energy_dm}
\end{figure}

\begin{table}[!tb]
	\centering
	\scriptsize
	\def\arraystretch{0.9}
	\caption{Model Confidence Set results for the Australian electricity generation dataset, evaluated using absolute loss across different forecast horizons ($h = 1$ and $h = 1, \dots, 7$) and aggregation level (all, upper and bottom times series). Each cell reports the percentage of series for which that approach is in the Model Confidence Set across different thresholds ($\delta \in \{95\%, 90\%, 80\%\}$).}\label{tab:oa_energy_mcs_mae}
	
\begin{tabular}[t]{>{}l|cc>{}c|ccc}
\toprule
\multicolumn{1}{c}{\textbf{ }} & \multicolumn{3}{c}{\textbf{$h =1$}} & \multicolumn{3}{c}{\textbf{$h =1:7$}} \\
\cmidrule(l{0pt}r{0pt}){2-4} \cmidrule(l{0pt}r{0pt}){5-7}
\multicolumn{1}{l|}{\textbf{Approach}} & $\delta = 95\%$ & $\delta = 90\%$ & $\delta = 80\%$ & $\delta = 95\%$ & $\delta = 90\%$ & $\delta = 80\%$\\
\midrule
\addlinespace[0.3em]
\multicolumn{7}{c}{\textbf{All 23 time series}}\\
\addlinespace[0em]
\multicolumn{7}{l}{\textit{Base (incoherent forecasts)}}\\
stlf & 47.8 & 39.1 & 26.1 & 47.8 & 43.5 & 34.8\\
arima & 60.9 & 60.9 & 56.5 & 73.9 & 73.9 & 69.6\\
tbats & 69.6 & 65.2 & 56.5 & 87.0 & 87.0 & 82.6\\
\addlinespace[0em]
\multicolumn{7}{l}{\textit{Single-expert reconciliation}}\\
stlf$_{\text{shr}}$ & 60.9 & 52.2 & 43.5 & 47.8 & 47.8 & 34.8\\
arima$_{\text{shr}}$ & 65.2 & 56.5 & 52.2 & 87.0 & 87.0 & 82.6\\
tbats$_{\text{shr}}$ & 91.3 & 87.0 & 73.9 & 87.0 & 87.0 & 82.6\\
\addlinespace[0em]
\multicolumn{7}{l}{\textit{Combination (incoherent forecasts)}}\\
ew & 91.3 & 91.3 & 82.6 & \em{95.7} & \em{91.3} & \em{91.3}\\
ow$_{\text{var}}$ & \em{95.7} & \em{95.7} & \em{91.3} & \em{95.7} & \em{91.3} & \em{91.3}\\
ow$_{\text{cov}}$ & 82.6 & 69.6 & 65.2 & 82.6 & 65.2 & 60.9\\
\addlinespace[0em]
\multicolumn{7}{l}{\textit{Coherent combination}}\\
src & \em{95.7} & 91.3 & 87.0 & \textbf{100.0} & \textbf{95.7} & \textbf{95.7}\\
scr$_{\text{ew}}$ & 91.3 & 91.3 & 87.0 & \textbf{100.0} & \textbf{95.7} & \em{91.3}\\
scr$_{\text{var}}$ & \textbf{100.0} & \textbf{100.0} & \textbf{95.7} & \em{95.7} & \textbf{95.7} & \textbf{95.7}\\
scr$_{\text{cov}}$ & 82.6 & 78.3 & 78.3 & 78.3 & 73.9 & 73.9\\
occ$_{\text{bv}}$ & 91.3 & 73.9 & 69.6 & 69.6 & 69.6 & 56.5\\
occ$_{\text{shr}}$ & 82.6 & 82.6 & 69.6 & 69.6 & 65.2 & 65.2\\
occ$_{\text{wls}}$ & \textbf{100.0} & \em{95.7} & 87.0 & 91.3 & \em{91.3} & \em{91.3}\\
occ$_{\text{be}}$ & \textbf{100.0} & \textbf{100.0} & \textbf{95.7} & \textbf{100.0} & \textbf{95.7} & \em{91.3}\\
\midrule
\addlinespace[0.3em]
\multicolumn{7}{c}{\textbf{8 upper time series}}\\
\addlinespace[0em]
\multicolumn{7}{l}{\textit{Base (incoherent forecasts)}}\\
stlf & 50.0 & 37.5 & 25.0 & 37.5 & 37.5 & 25.0\\
arima & 37.5 & 37.5 & 37.5 & 62.5 & 62.5 & 50.0\\
tbats & 50.0 & 50.0 & 50.0 & \em{87.5} & \em{87.5} & \em{87.5}\\
\addlinespace[0em]
\multicolumn{7}{l}{\textit{Single-expert reconciliation}}\\
stlf$_{\text{shr}}$ & 62.5 & 62.5 & 50.0 & 37.5 & 37.5 & 37.5\\
arima$_{\text{shr}}$ & 37.5 & 37.5 & 37.5 & \em{87.5} & \em{87.5} & 75.0\\
tbats$_{\text{shr}}$ & \em{87.5} & 75.0 & 62.5 & \em{87.5} & \em{87.5} & \em{87.5}\\
\addlinespace[0em]
\multicolumn{7}{l}{\textit{Combination (incoherent forecasts)}}\\
ew & \em{87.5} & \em{87.5} & \em{87.5} & \textbf{100.0} & \textbf{100.0} & \textbf{100.0}\\
ow$_{\text{var}}$ & \textbf{100.0} & \textbf{100.0} & \textbf{100.0} & \textbf{100.0} & \textbf{100.0} & \textbf{100.0}\\
ow$_{\text{cov}}$ & 75.0 & 75.0 & 62.5 & 75.0 & 62.5 & 62.5\\
\addlinespace[0em]
\multicolumn{7}{l}{\textit{Coherent combination}}\\
src & \em{87.5} & \em{87.5} & \em{87.5} & \textbf{100.0} & \textbf{100.0} & \textbf{100.0}\\
scr$_{\text{ew}}$ & \em{87.5} & \em{87.5} & \em{87.5} & \textbf{100.0} & \textbf{100.0} & \textbf{100.0}\\
scr$_{\text{var}}$ & \textbf{100.0} & \textbf{100.0} & \textbf{100.0} & \textbf{100.0} & \textbf{100.0} & \textbf{100.0}\\
scr$_{\text{cov}}$ & \em{87.5} & \em{87.5} & \em{87.5} & 75.0 & 75.0 & 75.0\\
occ$_{\text{bv}}$ & \em{87.5} & \em{87.5} & \em{87.5} & 62.5 & 62.5 & 62.5\\
occ$_{\text{shr}}$ & \em{87.5} & \em{87.5} & \em{87.5} & 62.5 & 62.5 & 62.5\\
occ$_{\text{wls}}$ & \textbf{100.0} & \textbf{100.0} & \em{87.5} & \textbf{100.0} & \textbf{100.0} & \textbf{100.0}\\
occ$_{\text{be}}$ & \textbf{100.0} & \textbf{100.0} & \textbf{100.0} & \textbf{100.0} & \textbf{100.0} & \textbf{100.0}\\
\midrule
\addlinespace[0.3em]
\multicolumn{7}{c}{\textbf{15 bottom time series}}\\
stlf & 46.7 & 40.0 & 26.7 & 53.3 & 46.7 & 40.0\\
arima & 73.3 & 73.3 & 66.7 & 80.0 & 80.0 & 80.0\\
tbats & 80.0 & 73.3 & 60.0 & 86.7 & \em{86.7} & 80.0\\
stlf$_{\text{shr}}$ & 60.0 & 46.7 & 40.0 & 53.3 & 53.3 & 33.3\\
arima$_{\text{shr}}$ & 80.0 & 66.7 & 60.0 & 86.7 & \em{86.7} & \em{86.7}\\
tbats$_{\text{shr}}$ & \em{93.3} & \em{93.3} & 80.0 & 86.7 & \em{86.7} & 80.0\\
ew & \em{93.3} & \em{93.3} & 80.0 & \em{93.3} & \em{86.7} & \em{86.7}\\
ow$_{\text{var}}$ & \em{93.3} & \em{93.3} & \em{86.7} & \em{93.3} & \em{86.7} & \em{86.7}\\
ow$_{\text{cov}}$ & 86.7 & 66.7 & 66.7 & 86.7 & 66.7 & 60.0\\
src & \textbf{100.0} & \em{93.3} & \em{86.7} & \textbf{100.0} & \textbf{93.3} & \textbf{93.3}\\
scr$_{\text{ew}}$ & \em{93.3} & \em{93.3} & \em{86.7} & \textbf{100.0} & \textbf{93.3} & \em{86.7}\\
scr$_{\text{var}}$ & \textbf{100.0} & \textbf{100.0} & \textbf{93.3} & \em{93.3} & \textbf{93.3} & \textbf{93.3}\\
scr$_{\text{cov}}$ & 80.0 & 73.3 & 73.3 & 80.0 & 73.3 & 73.3\\
occ$_{\text{bv}}$ & \em{93.3} & 66.7 & 60.0 & 73.3 & 73.3 & 53.3\\
occ$_{\text{shr}}$ & 80.0 & 80.0 & 60.0 & 73.3 & 66.7 & 66.7\\
occ$_{\text{wls}}$ & \textbf{100.0} & \em{93.3} & \em{86.7} & 86.7 & \em{86.7} & \em{86.7}\\
occ$_{\text{be}}$ & \textbf{100.0} & \textbf{100.0} & \textbf{93.3} & \textbf{100.0} & \textbf{93.3} & \em{86.7}\\
\bottomrule
\end{tabular}

\end{table}

\begin{table}[!tb]
	\centering
	\scriptsize
	\def\arraystretch{0.9}
	\caption{Model Confidence Set results for the Australian electricity generation dataset, evaluated using quadratic loss across different forecast horizons ($h = 1$ and $h = 1, \dots, 7$) and aggregation level (all, upper and bottom times series). Each cell reports the percentage of series for which that approach is in the Model Confidence Set across different thresholds ($\delta \in \{95\%, 90\%, 80\%\}$).}\label{tab:oa_energy_mcs_mse}
	
\begin{tabular}[t]{>{}l|cc>{}c|ccc}
\toprule
\multicolumn{1}{c}{\textbf{ }} & \multicolumn{3}{c}{\textbf{$h =1$}} & \multicolumn{3}{c}{\textbf{$h =1:7$}} \\
\cmidrule(l{0pt}r{0pt}){2-4} \cmidrule(l{0pt}r{0pt}){5-7}
\multicolumn{1}{l|}{\textbf{Approach}} & $\delta = 95\%$ & $\delta = 90\%$ & $\delta = 80\%$ & $\delta = 95\%$ & $\delta = 90\%$ & $\delta = 80\%$\\
\midrule
\addlinespace[0.3em]
\multicolumn{7}{c}{\textbf{All 23 time series}}\\
\addlinespace[0em]
\multicolumn{7}{l}{\textit{Base (incoherent forecasts)}}\\
stlf & 60.9 & 52.2 & 17.4 & 43.5 & 43.5 & 26.1\\
arima & 60.9 & 60.9 & 60.9 & 87.0 & 87.0 & 69.6\\
tbats & 69.6 & 65.2 & 65.2 & \textbf{100.0} & \em{95.7} & 87.0\\
\addlinespace[0em]
\multicolumn{7}{l}{\textit{Single-expert reconciliation}}\\
stlf$_{\text{shr}}$ & 69.6 & 60.9 & 47.8 & 52.2 & 47.8 & 34.8\\
arima$_{\text{shr}}$ & 56.5 & 56.5 & 56.5 & \textbf{100.0} & \textbf{100.0} & \textbf{95.7}\\
tbats$_{\text{shr}}$ & 78.3 & 78.3 & 69.6 & \textbf{100.0} & \em{95.7} & 87.0\\
\addlinespace[0em]
\multicolumn{7}{l}{\textit{Combination (incoherent forecasts)}}\\
ew & 91.3 & 87.0 & 82.6 & 91.3 & 91.3 & \em{91.3}\\
ow$_{\text{var}}$ & \textbf{100.0} & \textbf{95.7} & 87.0 & \em{95.7} & 91.3 & \em{91.3}\\
ow$_{\text{cov}}$ & 87.0 & 82.6 & 69.6 & 82.6 & 73.9 & 56.5\\
\addlinespace[0em]
\multicolumn{7}{l}{\textit{Coherent combination}}\\
src & \em{95.7} & \em{91.3} & \em{91.3} & \em{95.7} & \em{95.7} & \em{91.3}\\
scr$_{\text{ew}}$ & \em{95.7} & \em{91.3} & \em{91.3} & \em{95.7} & \em{95.7} & \em{91.3}\\
scr$_{\text{var}}$ & \em{95.7} & \textbf{95.7} & \textbf{95.7} & \em{95.7} & \em{95.7} & \textbf{95.7}\\
scr$_{\text{cov}}$ & 91.3 & \em{91.3} & 82.6 & 78.3 & 73.9 & 69.6\\
occ$_{\text{bv}}$ & 91.3 & \em{91.3} & 82.6 & 87.0 & 87.0 & 69.6\\
occ$_{\text{shr}}$ & 82.6 & 78.3 & 73.9 & 69.6 & 69.6 & 69.6\\
occ$_{\text{wls}}$ & \textbf{100.0} & \textbf{95.7} & \em{91.3} & 91.3 & 91.3 & \em{91.3}\\
occ$_{\text{be}}$ & \textbf{100.0} & \textbf{95.7} & \textbf{95.7} & \em{95.7} & \em{95.7} & \em{91.3}\\
\midrule
\addlinespace[0.3em]
\multicolumn{7}{c}{\textbf{8 upper time series}}\\
\addlinespace[0em]
\multicolumn{7}{l}{\textit{Base (incoherent forecasts)}}\\
stlf & 62.5 & 62.5 & 12.5 & 50.0 & 50.0 & 37.5\\
arima & 37.5 & 37.5 & 37.5 & \em{87.5} & \em{87.5} & 62.5\\
tbats & 50.0 & 50.0 & 50.0 & \textbf{100.0} & \em{87.5} & \em{87.5}\\
\addlinespace[0em]
\multicolumn{7}{l}{\textit{Single-expert reconciliation}}\\
stlf$_{\text{shr}}$ & 75.0 & 75.0 & 50.0 & 62.5 & 50.0 & 37.5\\
arima$_{\text{shr}}$ & 37.5 & 37.5 & 37.5 & \textbf{100.0} & \textbf{100.0} & \textbf{100.0}\\
tbats$_{\text{shr}}$ & 62.5 & 62.5 & 50.0 & \textbf{100.0} & \textbf{100.0} & \textbf{100.0}\\
\addlinespace[0em]
\multicolumn{7}{l}{\textit{Combination (incoherent forecasts)}}\\
ew & \em{87.5} & \em{87.5} & 75.0 & \textbf{100.0} & \textbf{100.0} & \textbf{100.0}\\
ow$_{\text{var}}$ & \textbf{100.0} & \textbf{100.0} & \em{87.5} & \textbf{100.0} & \textbf{100.0} & \textbf{100.0}\\
ow$_{\text{cov}}$ & \em{87.5} & 75.0 & 62.5 & \em{87.5} & \em{87.5} & 62.5\\
\addlinespace[0em]
\multicolumn{7}{l}{\textit{Coherent combination}}\\
src & \em{87.5} & \em{87.5} & \em{87.5} & \textbf{100.0} & \textbf{100.0} & \textbf{100.0}\\
scr$_{\text{ew}}$ & \em{87.5} & \em{87.5} & \em{87.5} & \textbf{100.0} & \textbf{100.0} & \textbf{100.0}\\
scr$_{\text{var}}$ & \textbf{100.0} & \textbf{100.0} & \textbf{100.0} & \textbf{100.0} & \textbf{100.0} & \textbf{100.0}\\
scr$_{\text{cov}}$ & \textbf{100.0} & \textbf{100.0} & \em{87.5} & \em{87.5} & 75.0 & 62.5\\
occ$_{\text{bv}}$ & \textbf{100.0} & \textbf{100.0} & \textbf{100.0} & \em{87.5} & \em{87.5} & 62.5\\
occ$_{\text{shr}}$ & \em{87.5} & \em{87.5} & \em{87.5} & 62.5 & 62.5 & 62.5\\
occ$_{\text{wls}}$ & \textbf{100.0} & \textbf{100.0} & \textbf{100.0} & \textbf{100.0} & \textbf{100.0} & \textbf{100.0}\\
occ$_{\text{be}}$ & \textbf{100.0} & \textbf{100.0} & \textbf{100.0} & \textbf{100.0} & \textbf{100.0} & \textbf{100.0}\\
\midrule
\addlinespace[0.3em]
\multicolumn{7}{c}{\textbf{15 bottom time series}}\\
stlf & 60.0 & 46.7 & 20.0 & 40.0 & 40.0 & 20.0\\
arima & 73.3 & 73.3 & 73.3 & 86.7 & 86.7 & 73.3\\
tbats & 80.0 & 73.3 & 73.3 & \textbf{100.0} & \textbf{100.0} & \em{86.7}\\
stlf$_{\text{shr}}$ & 66.7 & 53.3 & 46.7 & 46.7 & 46.7 & 33.3\\
arima$_{\text{shr}}$ & 66.7 & 66.7 & 66.7 & \textbf{100.0} & \textbf{100.0} & \textbf{93.3}\\
tbats$_{\text{shr}}$ & 86.7 & \em{86.7} & 80.0 & \textbf{100.0} & \em{93.3} & 80.0\\
ew & \em{93.3} & \em{86.7} & \em{86.7} & 86.7 & 86.7 & \em{86.7}\\
ow$_{\text{var}}$ & \textbf{100.0} & \textbf{93.3} & \em{86.7} & \em{93.3} & 86.7 & \em{86.7}\\
ow$_{\text{cov}}$ & 86.7 & \em{86.7} & 73.3 & 80.0 & 66.7 & 53.3\\
src & \textbf{100.0} & \textbf{93.3} & \textbf{93.3} & \em{93.3} & \em{93.3} & \em{86.7}\\
scr$_{\text{ew}}$ & \textbf{100.0} & \textbf{93.3} & \textbf{93.3} & \em{93.3} & \em{93.3} & \em{86.7}\\
scr$_{\text{var}}$ & \em{93.3} & \textbf{93.3} & \textbf{93.3} & \em{93.3} & \em{93.3} & \textbf{93.3}\\
scr$_{\text{cov}}$ & 86.7 & \em{86.7} & 80.0 & 73.3 & 73.3 & 73.3\\
occ$_{\text{bv}}$ & 86.7 & \em{86.7} & 73.3 & 86.7 & 86.7 & 73.3\\
occ$_{\text{shr}}$ & 80.0 & 73.3 & 66.7 & 73.3 & 73.3 & 73.3\\
occ$_{\text{wls}}$ & \textbf{100.0} & \textbf{93.3} & \em{86.7} & 86.7 & 86.7 & \em{86.7}\\
occ$_{\text{be}}$ & \textbf{100.0} & \textbf{93.3} & \textbf{93.3} & \em{93.3} & \em{93.3} & \em{86.7}\\
\bottomrule
\end{tabular}

\end{table}

\begin{figure}[!htbp]
	\centering
	\includegraphics[width = 0.9\linewidth]{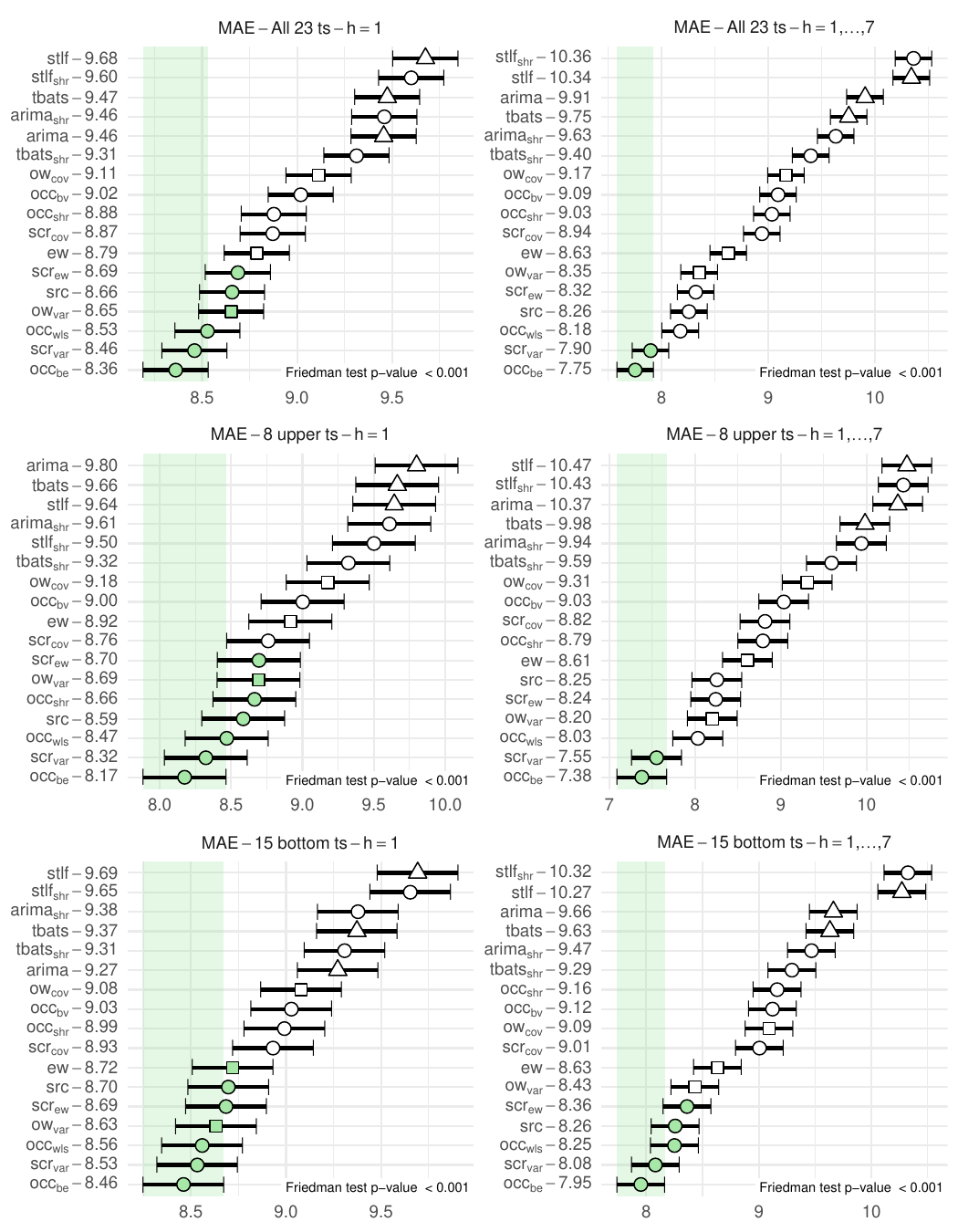}
	\caption{MCB Nemenyi test for the Australian electricity generation dataset using the MAE at different forecast horizon ($h = 1$ for the first column and $h = 1,...,7$ for the second column). In each panel, the Friedman test p-value is reported in the lower-right corner. The mean rank of each approach is shown to the right of its name. Statistically significant differences in performance are indicated if the intervals of two forecasting approaches do not overlap. Thus, approaches that do not overlap with the green interval are considered significantly worse than the best, and vice versa.}\label{fig:oa_energy_mcb_mae}
\end{figure}

\begin{figure}[!htbp]
\centering
\includegraphics[width = 0.9\linewidth]{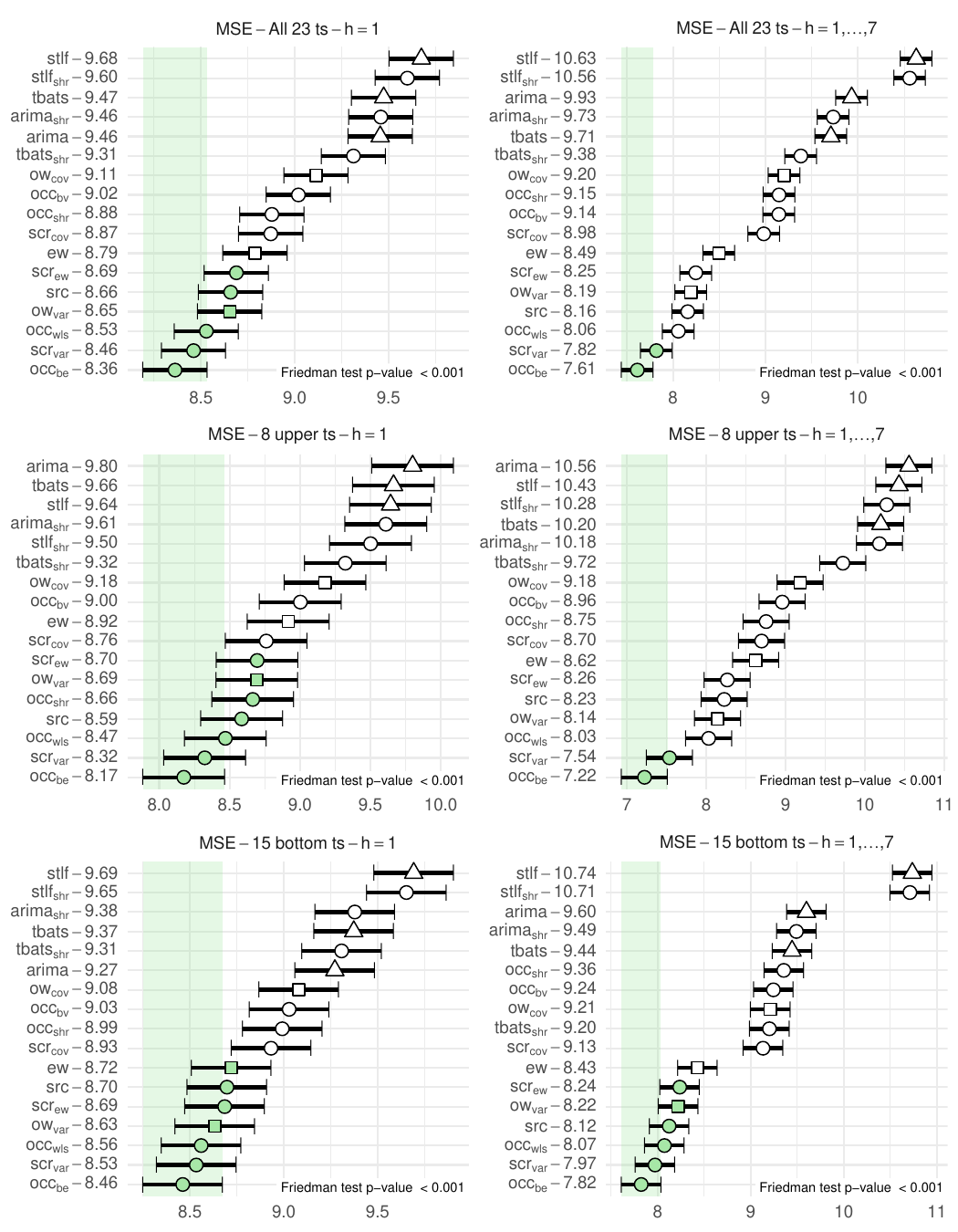}
\caption{MCB Nemenyi test for the Australian electricity generation dataset using the MSE at different forecast horizon ($h = 1$ for the first column and $h = 1,...,7$ for the second column). In each panel, the Friedman test p-value is reported in the lower-right corner. The mean rank of each approach is shown to the right of its name. Statistically significant differences in performance are indicated if the intervals of two forecasting approaches do not overlap. Thus, approaches that do not overlap with the green interval are considered significantly worse than the best, and vice versa.}\label{fig:oa_energy_mcb_mse}
\end{figure}

\clearpage
%\clearpage
\phantomsection\addcontentsline{toc}{section}{References}
\begingroup
\titlespacing*{\section}{0pt}{1\baselineskip}{0.5\baselineskip}
%\titlespacing*{\section}{0pt}{0.75\baselineskip}{0.75\baselineskip}
\spacingset{1.45}
\setlength{\bibsep}{0pt plus 0.1ex}
\bibliographystyle{apalike3link}

\bibliography{biblio.bib}

\endgroup
\end{document}